\documentclass[11pt,letter]{article}
\usepackage[utf8]{inputenc}
\pdfoutput=1
\usepackage{geometry,setspace,amsmath,amsfonts,amssymb,amsthm,lscape,graphicx,floatrow,bbm,bm,caption,wrapfig,fancyhdr,caption,booktabs,graphicx,dsfont,csquotes,zref,algorithm, algorithmic,subcaption,eucal,mathtools,wasysym,indentfirst,array,cite,rotating,etoolbox}
\usepackage{newfloat}
\newgeometry{left=1in, right=1in, top=1in, bottom=1in}
\usepackage[backref = page]{hyperref}
\usepackage[table,xcdraw]{xcolor}
\usepackage[square,sort,comma,longnamesfirst]{natbib}
\bibpunct{(}{)}{;}{a}{,}{,}

\definecolor{newred}{HTML}{E66100}
\definecolor{newgreen}{HTML}{44AA99}
\definecolor{newyellow}{HTML}{E899D5}
\makeatother
\hypersetup{colorlinks,linkcolor={newred},citecolor={newgreen},urlcolor={newyellow}} 
\DeclareMathOperator*{\E}{\mathbb{E}}

\newcommand{\argmin}{\arg\!\min}

\date{\today}
\usepackage[shortlabels]{enumitem}
\def\defeq{\equiv}
\newcommand{\EE}[2]{\mathbb{E}_{#1\!\!}\left[#2\right]}

\def\E#1{\EE{\,}{#1}}
\def\bR{{\mathbf R}}
\def\bI{{\mathbf I}}

\def\bA{{\mathbf A}}
\def\br{{\mathbf r}}
\def\bf{{\mathbf f}}

\def\bF{{\mathbf F}}
\def\bu{{\mathbf u}}

\def\bw{{\mathbf w}}
\def\bB{{\mathbf B}}

\def\bU{{\mathbf U}}
\def\bV{{\mathbf V}}

\def\bD{{\mathbf D}}
\def\bH{{\mathbf H}}

\def\bW{{\mathbf W}}
\def\bE{{\mathbf E}}

\def\bC{{\mathbf C}}

\def\bb{{\mathbf b}}
\def\by{{\mathbf y}}

\def\bLambda{{\bm \Lambda}}

\def\biota{{\bm \iota}}

\def\bbeta{{\bm \beta}}

\def\btheta{{\bm \theta}}
\def\bzeta{{\bm \zeta}}
\def\bsigma{{\bm \sigma}}
\def\bSigma{{\bm \Sigma}}
\def\bGamma{{\bm \Gamma}}
\def\bTheta{{\bm \Theta}}
\def\bPi{{\bm \Pi}}
\def\boldm{{\mathbf m}}

\def\bvarepsilon{{\bm \varepsilon}}
\DeclareMathAlphabet{\mathcal}{OMS}{cmsy}{m}{n}
\newtheorem{remark}{Remark}

\DeclarePairedDelimiter\abs{\lvert}{\rvert}%
\DeclarePairedDelimiter\norm{\lVert}{\rVert}%

\makeatletter
\let\oldabs\abs
\def\abs{\@ifstar{\oldabs}{\oldabs*}}
\let\oldnorm\norm
\def\norm{\@ifstar{\oldnorm}{\oldnorm*}}
\makeatother

\newcommand{\1}[1]{\mathds{1}\left[#1\right]}
\theoremstyle{plain}
\newtheorem{thm}{Theorem}
\newtheorem{prop}{Proposition}
\newtheorem{cor}{Corollary}
\newtheorem{lem}{Lemma}
\newcommand{\vertiii}[1]{{\left\vert\kern-0.25ex\left\vert\kern-0.25ex\left\vert #1 
		\right\vert\kern-0.25ex\right\vert\kern-0.25ex\right\vert}}
{
	\theoremstyle{plain}
	
}

\makeatletter
\def\mathcolor#1#{\@mathcolor{#1}}
\def\@mathcolor#1#2#3{%
	\protect\leavevmode
	\begingroup
	\color#1{#2}#3%
	\endgroup
}
\makeatother

\def\red#1{\mathcolor{black}{#1}}

\numberwithin{equation}{section}

\newcommand{\PreserveBackslash}[1]{\let\temp=\\#1\let\\=\temp}
\newcolumntype{C}[1]{>{\PreserveBackslash\centering}p{#1}}
\newcolumntype{R}[1]{>{\PreserveBackslash\raggedleft}p{#1}}
\newcolumntype{L}[1]{>{\PreserveBackslash\raggedright}p{#1}}

\newcommand\extrafootertext[1]{%
	\bgroup
	\renewcommand\thefootnote{\fnsymbol{footnote}}%
	\renewcommand\thempfootnote{\fnsymbol{mpfootnote}}%
	\footnotetext[0]{#1}%
	\egroup
}
\usepackage[toc, page, title]{appendix}
\begin{document}	
	\title{\textbf{Optimal Portfolio Using Factor Graphical Lasso}
	\extrafootertext{The authors would like to thank the editor Fabio Trojani and three anonymous referees for helpful and constructive comments on the paper.
		\vspace{.1cm}}}
	
\author{
	Tae-Hwy Lee\footnote{Department of Economics, University of California, Riverside. Email: tae.lee@ucr.edu.}\hskip 4mm \ and \hskip 2mm
	Ekaterina Seregina\footnote{Department of Economics, Colby College. Email: eseregin@colby.edu.}\hskip 8mm 
}

	\maketitle
	\setcounter{page}{1}	
	\begin{abstract}
\begin{spacing}{2}
Graphical models are a powerful tool to estimate a high-dimensional inverse covariance (precision) matrix, which has been applied for a portfolio allocation problem. The assumption made by these models is a sparsity of the precision matrix. However, when stock returns are driven by common factors, such assumption does not hold. We address this limitation and develop a framework, Factor Graphical Lasso (FGL), which integrates graphical models with the factor structure in the context of portfolio allocation by decomposing a precision matrix into low-rank and sparse components. Our theoretical results and simulations show that FGL consistently estimates the portfolio weights and risk exposure and also that FGL is robust to heavy-tailed distributions which makes our method suitable for financial applications. FGL-based portfolios are shown to exhibit superior performance over several prominent competitors including equal-weighted and Index portfolios in the empirical application for the S\&P500 constituents.	\end{spacing}
	
		\vskip 2mm
		
		
		\newpage
	\end{abstract} 
	
	\newpage 
	\setlength{\baselineskip}{22pt}
	\setstretch{2}
	\section{Introduction}
	Estimating the inverse covariance matrix, or \textit{precision} matrix, of excess stock returns is crucial for constructing weights of financial assets in a portfolio and estimating the out-of-sample Sharpe Ratio. In high-dimensional setting, when the number of assets, $p$, is greater than or equal to the sample size, $T$, using an estimator of \textit{covariance} matrix for obtaining portfolio weights leads to unstable investment allocations. This is known as the  Markowitz’ curse: a higher number of assets increases correlation between the investments, which calls for a more diversified portfolio, and yet unstable corner solutions for weights become more likely. The reason behind this curse is the need to invert a high-dimensional covariance matrix to obtain the optimal weights from the quadratic optimization problem: when $p\geq T$, the condition number of the covariance matrix (i.e., the absolute value of the ratio between maximal and minimal eigenvalues of the covariance matrix) is high. Hence, the inverted covariance matrix yields an unstable estimator of the precision matrix. To circumvent this issue one can estimate precision matrix directly, rather than inverting an estimated covariance matrix.
	
	 Graphical models were shown to provide consistent estimates of the precision matrix (\cite{GLASSO,meinshausen2006,cai2011constrained}). \cite{goto2015} estimated a sparse precision matrix for portfolio hedging using graphical models. They found out that their portfolio achieves significant out-of-sample risk reduction and higher return, as compared to the portfolios based on equal weights, shrunk covariance matrix, industry factor models, and no-short-sale constraints. \cite{AwoyePhD} used Graphical Lasso (\cite{GLASSO}) to estimate a sparse covariance matrix for the Markowitz mean-variance portfolio problem and reduce the realized portfolio risk. \cite{Millington} conducted an empirical study that applies Graphical Lasso for the estimation of covariance for the portfolio allocation. Their empirical findings suggest that portfolios using Graphical Lasso enjoy lower risk and higher returns compared to those using empirical covariance matrix. \cite{Millington} also construct a financial network using the estimated precision matrix to explore the relationship between the companies and show how the constructed network helps to make investment decisions. \cite{Caner2019} use the nodewise-regression method of \cite{meinshausen2006} to establish consistency of the estimated covariance matrix, weights and risk of high-dimensional financial portfolio. Their empirical application demonstrates that the precision matrix estimator based on the nodewise-regression outperforms the principal orthogonal complement thresholding estimator (POET) (\cite{fan2013POET}) and linear shrinkage (\cite{Ledoit2004}). \cite{cai2020high} use constrained $\ell_1$-minimization for inverse matrix estimation (Clime) of the precision matrix (\cite{cai2011constrained}) to develop a consistent estimator of the minimum variance for high-dimensional global minimum-variance portfolio. It is important to note that all the aforementioned methods impose some sparsity assumption on the precision matrix of excess returns.
	 
	 An alternative strategy to handle high-dimensional setting uses factor models to acknowledge common variation in the stock prices, which was documented in many empirical studies (see \cite{campbell1997} among many others). A common approach decomposes covariance matrix of excess returns into low-rank and sparse parts, the latter is further regularized since, after the common factors are accounted for, the remaining covariance matrix of the idiosyncratic components is still high-dimensional (\cite{fan2013POET,Fan2011,fan2018elliptical}). This stream of literature, however, focuses on the estimation of a covariance matrix. The accuracy of precision matrices obtained from inverting the factor-based covariance matrix was investigated by \cite{ait2017using}, but they did not study a high-dimensional case. \textit{Factor models are generally treated as competitors to graphical models}: as an example, \cite{Caner2019} find evidence of superior performance of nodewise-regression estimator of precision matrix over a factor-based estimator POET (\cite{fan2013POET}) in terms of the out-of-sample Sharpe Ratio and risk of financial portfolio. The root cause why factor models and graphical models are treated separately is the sparsity assumption on the precision matrix made in the latter. Specifically, as pointed out in \cite{koike2019biased}, \textit{when asset returns have common factors, the precision matrix cannot be sparse because all pairs of assets are partially correlated conditional on other assets through the common factors}. One attempt to integrate factor modeling and high-dimensional precision estimation was made by \cite{fan2018elliptical} (Section 5.2): the authors referred to such class of models as ``conditional graphical models". However, this was not the main focus of their paper which concentrated on covariance estimation through elliptical factor models. As \cite{fan2018elliptical} pointed out, ``\textit{though substantial amount of efforts have been made to understand the graphical model, little has been done for estimating conditional graphical model, which is more general and realistic}". Concretely, to the best of our knowledge there are no studies that examine theoretical and empirical performance of graphical models integrated with the factor structure in the context of portfolio allocation.

	 In this paper we fill this gap and develop a new conditional precision matrix estimator for the excess returns under the approximate factor model that combines the benefits of graphical models and factor structure. We call our algorithm the \textit{Factor Graphical Lasso (FGL)}. We use a factor model to remove the co-movements induced by the factors, and then we apply the Weighted Graphical Lasso for the estimation of the precision matrix of the idiosyncratic terms. We prove consistency of FGL in the spectral and $\ell_{1}$ matrix norms. In addition, we prove consistency of the estimated portfolio weights and risk exposure for three formulations of the optimal portfolio allocation.
	 
	 Our empirical application uses daily and monthly data for the constituents of the S\&P500: we demonstrate that FGL outperforms equal-weighted portfolio, index portfolio, portfolios based on other estimators of precision matrix (Clime, \cite{cai2011constrained}) and covariance matrix, including POET (\cite{fan2013POET}) and the shrinkage estimators adjusted to allow for the factor structure (\cite{Ledoit2004}, \cite{ledoit2017nonlinear}), in terms of the out-of-sample Sharpe Ratio. Furthermore, we find strong empirical evidence that relaxing the constraint that portfolio weights sum up to one leads to a large increase in the out-of-sample Sharpe Ratio, which, to the best of our knowledge, has not been previously well-studied in the empirical finance literature.
	 
	 From the theoretical perspective, our paper makes several important contributions to the existing literature on graphical models and factor models. First, to the best of out knowledge, there are no equivalent theoretical results that establish consistency of the portfolio weights and risk exposure in a high-dimensional setting \textit{without assuming sparsity on the covariance or precision matrix of stock returns}. Second, we extend the theoretical results of POET (\cite{fan2013POET}) to allow the number of factors to grow with the number of assets. Concretely, we establish uniform consistency for the factors and factor loadings estimated using PCA. Third, we are not aware of any other papers that provide convergence results for estimating a high-dimensional precision matrix using the Weighted Graphical Lasso under the approximate factor model with unobserved factors. Furthermore, all theoretical results established in this paper hold for a wide range of distributions: Sub-Gaussian family (including Gaussian) and elliptical family. Our simulations demonstrate that FGL is robust to very heavy-tailed distributions, which makes our method suitable for the financial applications. Finally, we demonstrate that in contrast to POET, the success of the proposed method does not heavily depend on the factor pervasiveness assumption: FGL is robust to the scenarios when the gap between the diverging and bounded eigenvalues decreases.

	 This paper is organized as follows: Section 2 reviews the basics of the Markowitz mean-variance portfolio theory. Section 3 provides a brief summary of the graphical models and introduces the Factor Graphical Lasso. Section 4 contains theoretical results and Section 5 validates these results using simulations. Section 6 provides empirical application. Section 7 concludes.
	
	\phantomsection
	\subsection*{Notation}
	\addcontentsline{toc}{section}{Notation} 
	For the convenience of the reader, we summarize the notation to be used throughout the paper. Let $\mathcal{S}_p$ denote the set of all $p \times p$ symmetric matrices, and $\mathcal{S}_{p}^{++}$ denotes the set of all $p \times p$ positive definite matrices. For any matrix $\bC$, its $(i,j)$-th element is denoted as $c_{ij}$.
	Given a vector $\bu\in \mathbb{R}^d$ and parameter $a\in \lbrack1,\infty)$, let $\norm{\bu}_a$ denote $\ell_a$-norm. Given a matrix $\bU \in\mathcal{S}_p$, let $\Lambda_{\text{max}}(\bU) \defeq \Lambda_1(\bU) \geq \Lambda_2(\bU)\geq \ldots \geq \Lambda_{\text{min}}(\bU) \defeq  \Lambda_p(\bU)$ be the eigenvalues of $\bU$, and $\text{eig}_K(\bU) \in \mathbb{R}^{K\times p}$ denote the first $K\leq p$ normalized eigenvectors corresponding to $\Lambda_1(\bU), \ldots ,\Lambda_K(\bU)$. Given parameters  $a,b\in \lbrack1,\infty)$, let $\vertiii{\bU}_{a,b}\defeq \max_{\norm{\by}_a=1}\norm{\bU\by}_{b}$ denote the induced matrix-operator norm. The special cases are $\vertiii{\bU}_1\defeq \max_{1\leq j\leq N}\sum_{i=1}^{N}\abs{u_{i,j}}$ for the $\ell_1/\ell_1$-operator norm; the operator norm ($\ell_2$-matrix norm) $\vertiii{\bU}_{2}^{2}\defeq\Lambda_{\text{max}}(\bU\bU')$ is equal to the maximal singular value of $\bU$; $\vertiii{\bU}_{\infty}\defeq \max_{1\leq j\leq N}\sum_{i=1}^{N}\abs{u_{j,i}}$ for the $\ell_{\infty}/\ell_{\infty}$-operator norm. Finally, $\norm{\bU}_{\text{max}}\defeq \max_{i,j}\abs{u_{i,j}}$ denotes the element-wise maximum, and $\vertiii{\bU}_{F}^{2}\defeq\sum_{i,j}u_{i,j}^{2}$ denotes the Frobenius matrix norm. 
	 
	\section{Optimal Portfolio Allocation}
	
	Suppose we observe $p$ assets (indexed by $i$) over $T$ period of time (indexed by $t$). Let $\widetilde{\br}_t=(\widetilde{r}_{1t}, \widetilde{r}_{2t},\ldots,\widetilde{r}_{pt})' \sim  \mathcal{D} (\boldm, \bSigma)$ be a $p \times 1$ vector of \textit{excess} returns drawn from a distribution $\mathcal{D}$, where $\boldm$ and $\bSigma$ are the unconditional mean and covariance matrix of the returns. The goal of the Markowitz theory is to choose asset weights in a portfolio \textit{optimally}. We will study two optimization problems: the well-known Markowitz weight-constrained (MWC) optimization problem, and the Markowitz risk-constrained (MRC) optimization that relaxes the constraint on portfolio weights.
	
	The first optimization problem searches for asset weights such that the portfolio achieves a desired expected rate of return with minimum risk, under the restriction that all weights sum up to one. This can be formulated as the following quadratic optimization problem:
	\begin{equation} \label{sys1}
	\min_{\bw}\frac{1}{2} \bw'\bSigma\bw,  \ \text{s.t.} \ \bw'\biota_p =1 \ \text{and} \ \boldm'\bw\geq\mu 
	\end{equation}
	where $\bw$ is a $p \times 1$ vector of asset weights in the portfolio, $\biota_p$ is a $p \times 1$ vector of ones, and $\mu$ is a desired expected rate of portfolio return. Let $\bTheta\defeq\bSigma^{-1}$ be the \textit{precision matrix}.
	
	If $\boldm'\bw>\mu$, then the solution to \eqref{sys1} yields the global minimum-variance (GMV) portfolio weights $\bw_{GMV}$:
	\begin{equation} \label{eq2}
	\bw_{GMV}=(\biota_p'\bTheta\biota_p)^{-1}\bTheta\biota_p.
	\end{equation}

	\indent If $\boldm'\bw=\mu$, the solution to \eqref{sys1} is a well-known two-fund separation theorem introduced by \cite{Tobin}:
	\begin{align}
	&\bw_{MWC}=(1-a_1)\bw_{GMV}+a_1 \bw_{M},\label{eq3}
	\end{align}
	where $\bw_{MWC}$ denotes the portfolio allocation with the constraint that the weights need to sum up to one, $\bw_{M}=(\biota_p'\bTheta\boldm)^{-1}\bTheta\boldm$, and $a_1=[\mu(\boldm'\bTheta\biota_p)(\biota_p'\bTheta\biota_p)-(\boldm'\bTheta\biota_p)^2]/[(\boldm'\bTheta\boldm)(\biota_p'\bTheta\biota_p)-(\boldm'\bTheta\biota_p)^2]$.
	
	The MRC problem maximizes Sharpe Ratio (SR) subject to either target risk or target return constraints, but portfolio weights are not required to sum up to one:
	\begin{equation} \label{ee19}
		\max_{\bw} \frac{\boldm'\bw}{\sqrt{\bw'\bSigma\bw}} \ \text{s.t.}\ \text{(i)} \ \boldm'\bw\geq\mu \ \text{or} \text{(ii)} \ \bw'\bSigma\bw\leq \sigma^2.
	\end{equation}
	When  $\mu=\sigma\sqrt{\boldm'\bTheta\boldm}$, the solution to either of the constraints is given by
		\begin{equation}\label{ee20}
		\bw_{MRC}=\frac{\sigma}{\sqrt{\boldm'\bTheta\boldm}}\bTheta\boldm.
	\end{equation}
	 Equation \eqref{ee19} tells us that once an investor specifies the desired return, $\mu$, and maximum risk-tolerance level, $\sigma$, the MRC weight maximizes the Sharpe Ratio of the portfolio.

	Therefore, we have three alternative portfolio allocations commonly used in the existing literature: GMV in \eqref{eq2}, MWC in \eqref{eq3} and MRC in \eqref{ee20}.
It is clear that all formulations require an estimate of the precision matrix $\bTheta$. 
	
	\section{Factor Graphical Lasso}
	In this section we introduce a framework for estimating precision matrix for the aforementioned financial portfolios which accounts for the fact that the returns follow approximate factor structure. We examine how to solve the Markowitz mean-variance portfolio allocation problems using factor structure in the returns. We also develop \textit{Factor Graphical Lasso} Algorithm that uses the estimated common factors to obtain a sparse precision matrix of the idiosyncratic component. The resulting estimator is used to obtain the precision of the asset returns necessary to form portfolio weights.

	The arbitrage pricing theory (APT), developed by \cite{APTRoss}, postulates that the expected returns on securities should be related to their covariance with the common components or factors. The goal of the APT is to model the tendency of asset returns to move together via factor decomposition. Assume that the return generating process ($\widetilde{\br}_t$) follows a $K$-factor model:
	\begin{align} \label{e5.1}
	&\underbrace{\widetilde{\br}_t}_{p \times 1}=\boldm + \bB \underbrace{\bf_t}_{K\times 1}+ \ \bvarepsilon_t,\quad t=1,\ldots,T
	\end{align}
	where $\bf_t=(f_{1t},\ldots, f_{Kt})'$ are the factors, $\bB$ is a $p \times K$ matrix of factor loadings, and $\bvarepsilon_t$ is the idiosyncratic component that cannot be explained by the common factors. Without loss of generality, we assume throughout the paper that unconditional means of factors and idiosyncratic component are zero. Factors in \eqref{e5.1} can be either observable, such as in \cite{Fama3Factor,Fama5Factor}, or can be estimated using statistical factor models. Unobservable factors and loadings are usually estimated by the principal component analysis (PCA), as studied in \cite{Bai2003}, \cite{Bai2002}, \cite{Connor1988}, and \cite{Stock2002}. 
	
    In this paper our main interest lies in establishing asymptotic properties of the estimators of precision matrix, portfolio weights and risk-exposure for the high-dimensional case. We assume that the number of common factors, $K=K_{p,T}\rightarrow\infty$ as $p\rightarrow\infty$, or $T\rightarrow\infty$, or both $p,T \rightarrow \infty$, but we require that $\max\{K/p,K/T\}\rightarrow0$ as  $p,T \rightarrow \infty$.
	
	Our setup is similar to the one studied in \cite{fan2013POET}: we consider a spiked covariance model when the first $K$ principal eigenvalues of $\bSigma$ are growing with $p$, while the remaining $p-K$ eigenvalues are bounded. 
	
	Rewrite equation \eqref{e5.1} in matrix form:
	\begin{equation} \label{5.2}
	\underbrace{\widetilde{\bR}}_{p\times T} = \boldm \biota'_T + \underbrace{\bB}_{p\times K} \bF + \bE,
	\end{equation}
	where $\biota_T$ is a $T\times 1$ vector of ones. We further demean the returns using the sample mean, $\widehat{\boldm}$, to obtain $\bR \defeq \widetilde{\bR} - \widehat{\boldm}\biota'_T$. We assume that $\norm{\widehat{\boldm}-\boldm}_{\text{max}}=\mathcal{O}_P(\sqrt{\log p/T})$,  which was proven to hold in \cite{CHANG2018} (see their Lemma 1).
	
	 Let $\bSigma_{\varepsilon}=T^{-1}\bE\bE'$ and $\bSigma_{f}=T^{-1}\bF\bF'$ be covariance matrices of the idiosyncratic components and factors, and let $\bTheta_{\varepsilon}=\bSigma_{\varepsilon}^{-1}$ and $\bTheta_{f}=\bSigma_{f}^{-1}$ be their inverses. The factors and loadings in \eqref{5.2} are estimated by solving the following minimization problem: $(\widehat{\bB},\widehat{\bF})=\argmin_{\bB,\bF}\norm{\bR-\bB\bF}^{2}_{F}$ s.t. $\frac{1}{T}\bF\bF'=\bI_K, \ \bB'\bB\ \text{is diagonal}$. The constraints are needed to identify the factors (\cite{fan2018elliptical}). It was shown (\cite{Stock2002}) that $\widehat{\bF}=\sqrt{T}\text{eig}_K(\bR'\bR)$ and $\widehat{\bB}=T^{-1}\bR\widehat{\bF}'$. Given $\widehat{\bF},\widehat{\bB}$, define $\widehat{\bE}=\bR-\widehat{\bB}\widehat{\bF}$. Given a sample of the estimated residuals $\{\widehat{\bvarepsilon}_t=\br_t-\widehat{\bB}\widehat{\bf_t}\}_{t=1}^{T}$ and the estimated factors $\{\widehat{\bf}_t\}_{t=1}^{T}$, let $\widehat{\bSigma}_{\varepsilon} = (1/T)\sum_{t=1}^{T}\widehat{\bvarepsilon}_t\widehat{\bvarepsilon}_t'$ and $\widehat{\bSigma}_{f}=(1/T)\sum_{t=1}^{T}\widehat{\bf}_t\widehat{\bf}_t'$ be the sample counterparts of the covariance matrices. 	Since our interest is in constructing portfolio weights, our goal is to estimate a precision matrix of the excess returns $\bTheta$.
	
    We impose a sparsity assumption on the precision matrix of the idiosyncratic errors, $\bTheta_{\varepsilon}$, which is obtained using the estimated residuals after removing the co-movements induced by the factors (see \cite{Brownlees2018EJS,Brownlees2018JAE,koike2019biased}). 
    
	Let us elaborate on three reasons justifying the assumption of sparsity on the precision matrix of residuals. First, from the technical viewpoint, this assumption is widely used in high-dimensional settings when $p>T$.
	Second, a more intuitive rationale for the sparsity assumption on $\bTheta_{\varepsilon}$ stems from its implication for the structure of corresponding optimal portfolios. Let $r_{t}^{\text{portf}}\defeq \widetilde{\br}_{t}'\bw_t$ be the optimal portfolio. Plugging in the definition of $\widetilde{\br}_{t}$ from \eqref{e5.1}, we get $r_{t}^{\text{portf}} = (\boldm+\bvarepsilon_{t})'\bw_t + \bf_{t}'\bB\bw_t$. Hence, after hedging factor risk, we can isolate the excess return component only loading on non-factor risk. In this context, since $\bw_t$ is a function of $\bTheta_{\varepsilon}$, imposing sparsity on $\bTheta_{\varepsilon}$ translates into reducing the contribution of more volatile non-factor risk on the optimal portfolio and thus leading to less sensitive (more robust) investment strategies.
	
	Third, another rationale comes from relatively high ``concentration" of S\&P 500 Composite Index: as evidenced from \href{https://www.spglobal.com/spdji/en/governance/methodologies/#methodology-information}{SP Global Index methodology} and \href{https://www.slickcharts.com/sp500}{financial data on S\&P 500 constituents by weight}, 15 large companies (top 3\%) comprise 30\% of the total index weights (starting from Apple that has the highest weight of nearly 7\%). As the number of firms, $p$, increases, one reasonable assumption is that the number of large firms increases at a rate slower than $p$ (\cite{Chudik,gabaix2011granular}). This suggests that one could divide the firms into dominant ones and followers. After the effect of common factors is accounted for, dominant firms still have significant idiosyncratic movements that influence other firms and must be taken into account when constructing a portfolio. When it comes to fringe firms (or market followers), idiosyncratic movements are smaller in magnitude and might be less relevant for portfolio allocation purposes. Hence, the network of the idiosyncratic returns is sparse and the sparsity increases with $p$. By imposing sparsity, we only keep relatively large partial correlations among idiosyncratic components: as illustrated in Supplemental Appendix \ref{appendixC4}, in our empirical application the estimated number of zeroes in off-diagonal elements of $\bTheta_{\varepsilon}$ varies over time from 74.5\%-98.8\%.

	\indent Henceforth, having established the need for a sparse precision of errors, we search for a tool that would help us recover its entries. This brings us to consider a family of graphical models, which have evolved from the connection between partial correlations and the entries of an adjacency matrix. The adjacency matrix has zero or one in its entries, with a zero entry indicating that two variables are independent conditional on the rest. The adjacency matrix is sometimes referred to as a ``graph". Graphical Lasso procedure (\cite{GLASSO}) described in Supplemental Appendix \ref{appendixAA} is a representative member of graphical models family: its theoretical and empirical properties have been thoroughly examined in a standard sparse setting (\cite{GLASSO}, \cite{DPGLASSO}, \cite{Sara2018}). One of the goals of our paper is to augment graphical models to non-sparse settings through integrating them with factor modeling. By doing so, graphical models would become adequate for applications in economics and finance.
	
	A common way to induce sparsity is by utilizing Lasso-type penalty. This strategy is used in the Graphical Lasso (GL) together with the objective function based on the Bregman divergence for estimating inverse covariance. 
	The discussion of GL is presented in Supplemental Appendix \ref{appendixAA}. We now elaborate on the Bregman divergence class which unifies many commonly used loss functions, including the quasi-likelihood function. Let $\bW_{\varepsilon}$ be an estimate of $\bSigma_{\varepsilon}$. \cite{ravikumar2011} showed that Bregman divergence of the form $\text{trace}(\bW_{\varepsilon}\bTheta_{\varepsilon})-\log\det(\bTheta_{\varepsilon})$, known as the log-determinant Bregman function,		
	 is suitable to be used as a measure of the quality of constructed sparse approximations of signals such as precision matrices. As pointed out by \cite{ravikumar2011}, in principle one could use other Bregman divergences including the von Neumann Entropy or the Frobenius divergence which would lead to alternative forms of divergence minimizations for estimating precision matrix. We proceed with the log-determinant Bregman function since (i) it ensures positive definite estimator of precision matrix; (ii) the population optimization problem involves only the population covariance and not its inverse; (iii) the log-determinant divergence gives rise to the likelihood function in the multivariate Gaussian case. At the same time, despite its resemblance with the Gaussian log-likelihood, Bregman divergence was shown to be applicable for non-Gaussian distributions (\cite{ravikumar2011}). Let $\widehat{\bD}_{\varepsilon}^{2}\defeq \textup{diag}(\bW_{\varepsilon})$. To sparsify entries of precision matrix of the idiosyncratic errors $\bTheta_{\varepsilon}$, we use the following penalized Bregman divergence with the Weighted Graphical Lasso penalty:	
	\begin{align} \label{e7.6}
	&\widehat{\bTheta}_{\varepsilon,\lambda}=\argmin_{\bTheta \in \mathcal{S}_{p}^{++}}\text{trace}(\bW_{\varepsilon}\bTheta_{\varepsilon})-\log\det(\bTheta_{\varepsilon})+\lambda\sum_{i\neq j}\widehat{d}_{\varepsilon,ii}\widehat{d}_{\varepsilon,jj}\abs{\theta_{\varepsilon,ij}}.
	\end{align}
	
	The subscript $\lambda$ in $\widehat{\bTheta}_{\varepsilon,\lambda}$ means that the solution of the optimization problem in \eqref{e7.6} will depend upon the choice of the tuning parameter which is discussed below. Section 4 establishes sparsity requirements that guarantee convergence of \eqref{e7.6}. In order to simplify notation, we will omit the subscript $\lambda$.
	
	The objective function in \eqref{e7.6} extends the family of linear shrinkage estimators of the first moment to linear shrinkage estimators of the inverse of the second moments. Instead of restricting the number of regressors for estimating conditional mean, equation \eqref{e7.6} restricts the number of edges in a graph by shrinking some off-diagonal entries of precision matrix to zero. Note that shrinkage occurs adaptively with respect to partial covariances.
	
	Let us discuss the choice of the tuning parameter $\lambda$ in \eqref{e7.6}. Let $\widehat{\bTheta}_{\varepsilon,\lambda}$ be the solution to \eqref{e7.6} for a fixed $\lambda$. Following \cite{koike2019biased}, we minimize the following Bayesian Information Criterion (BIC) using grid search:
	\begin{equation} \label{eq5.1}
		\text{BIC}(\lambda) \defeq T\Big[\text{trace}(\widehat{\bTheta}_{\varepsilon,\lambda}\widehat{\bSigma}_{\varepsilon})-\log\text{det}(\widehat{\bTheta}_{\varepsilon, \lambda}) \Big] + (\log T)\sum_{i\leq j}\1{\widehat{\theta}_{\varepsilon,\lambda,ij}\neq 0}.
	\end{equation}
	The grid $\mathcal{G}\defeq \{\lambda_1,\ldots,\lambda_{M}\}$ is constructed as follows: the maximum value in the grid, $\lambda_{M}$, is set to be the smallest value for which all the off-diagonal entries of $\widehat{\bTheta}_{\varepsilon,\lambda_{M}}$ are zero. The smallest value of the grid, $\lambda_{1}\in \mathcal{G}$, is determined as $\lambda_{1}\defeq\vartheta\lambda_{M}$ for a constant $0<\vartheta<1$. The remaining grid values $\lambda_1,\ldots,\lambda_{M}$ are constructed in the ascending order from $\lambda_{1}$ to $\lambda_{M}$ on the log scale:
	\begin{equation*}
		\lambda_i=\exp \Big(\log(\lambda_{1})+\frac{i-1}{M-1}\log(\lambda_{M}/\lambda_{1})  \Big), \quad i=2,\ldots,M-1.
	\end{equation*}
	We use $\vartheta=\omega_{3T}$ which is defined in Theorem 2 of the next section and $M=10$ in the simulations and the empirical exercise.
	
Having estimated factors, factor loadings and precision matrix of the idiosyncratic components, we combine them using the Sherman-Morrison-Woodbury formula to estimate the precision matrix of excess returns:	
	\begin{equation} \label{equa18}
		\widehat{\bTheta}=\widehat{\bTheta}_{\varepsilon}-\widehat{\bTheta}_{\varepsilon}\widehat{\bB}\lbrack\widehat{\bTheta}_f+\widehat{\bB}'\widehat{\bTheta}_{\varepsilon}\widehat{\bB}\rbrack^{-1}\widehat{\bB}'\widehat{\bTheta}_{\varepsilon}.
	\end{equation}

	 To solve \eqref{e7.6} we use the procedure based on the GL. However, the original algorithm developed by \cite{GLASSO} is not suitable under the factor structure. Our procedure called Factor Graphical Lasso (FGL), which is summarized in Procedure \ref{alg2}, augments the standard GL: it starts with estimating factors, loadings (low-rank part) and error terms (sparse part), then it proceeds by recovering sparse precision matrix of the errors using GL, and, finally, low-rank and sparse components are combined through Shermann-Morrison-Woodbury formula in \eqref{equa18}.
	\begin{spacing}{1.6}
		\begin{algorithm}[H]
		\floatname{algorithm}{Procedure}
		\caption{Factor Graphical Lasso}
		\label{alg2}
		\begin{algorithmic}[1]
			\STATE 	\textbf{(F}actor Model) Estimate $\widehat{\bf}_t$ and $\widehat{\bb}_i$ (Theorem \ref{theor1}). Get $\widehat{\bvarepsilon}_t=\br_t-\widehat{\bB}\widehat{\bf_t}$, $\widehat{\bSigma}_{\varepsilon}$, $\widehat{\bSigma}_f$ and $\widehat{\bTheta}_f=\widehat{\bSigma}_{f}^{-1}$.
			\STATE \textbf{(GL)} Use GL from \cite{GLASSO} (see Supplemental Appendix \ref{appendixAA} for more details) to get $\widehat{\bTheta}_{\varepsilon}$. (Theorem \ref{theor2})
			\STATE \textbf{(FGL)} Use $\widehat{\bTheta}_{\varepsilon}$, $\widehat{\bTheta}_f$ and $\widehat{\bb}_i$ from Steps 1-2 to get $\widehat{\bTheta}$ in Equation \eqref{equa18}. (Theorem \ref{theor3})
			\STATE Use $\widehat{\bTheta}$ to get $\widehat{\bw}_{\xi}$, $\xi\in\{\text{GMV, MWC, MRC}\}$. (Theorem \ref{theor4})
			\STATE Use $\widehat{\bSigma}=\widehat{\bTheta}^{-1}$ and $\widehat{\bw}_{\xi}$ to get portfolio exposure $\widehat{\bw}_{\xi}^{'}\widehat{\bSigma}\widehat{\bw}_{\xi}$. (Theorem \ref{theor5})
		\end{algorithmic}
	\end{algorithm}
	\end{spacing}	
The estimator produced by GL in general and FGL in particular is guaranteed to be positive definite. We have verified it in the simulations (Section 5) and the empirical application (Section 6). In Section 4, consistency properties of estimators are established for the factors and loadings (Theorem \ref{theor1}), the precision matrix of $\bvarepsilon$ (Theorem \ref{theor2}), the precision matrix $\bTheta$ (Theorem \ref{theor3}), portfolio weights (Theorem \ref{theor4}), and the portfolio risk exposure (Theorem \ref{theor5}). We can use $\widehat{\bTheta}$ obtained from \eqref{equa18} using Step 4 of Procedure \ref{alg2} to estimate portfolio weights in \eqref{eq2}, \eqref{eq3} and \eqref{ee20}:

\section{Asymptotic Properties}
In this section we first provide a brief review of the terminology used in the literature on graphical models and the approaches to estimate a precision matrix. After that we establish consistency of the Factor Graphical Lasso in Procedure \ref{alg2}. We also study consistency of the estimators of weights in \eqref{eq2}, \eqref{eq3} and \eqref{ee20} and the implications on the out-of sample Sharpe Ratio. Throughout the main text we assume that errors and factors have exponential-type tails (\ref{A3}\red{(c)}). Supplemental Appendix \ref{appendixA10} proves that the conclusions of all theorems studied in Section 4 continue to hold when this assumption is relaxed.

The review of the Gaussian graphical models is based on \cite{ESLII} and \cite{Bishop2006}. A \textit{graph} consists of a set of \textit{vertices} (nodes) and a set of \textit{edges} (arcs) that join some pairs of the vertices. In graphical models, each vertex represents a random variable, and the graph visualizes the joint distribution of the entire set of random variables.
The edges in a graph are parameterized by \textit{potentials} (values) that encode the strength of the conditional dependence between the random variables at the corresponding vertices. \textit{Sparse graphs} have a relatively small number of edges. Among the main challenges in working with the graphical models are choosing the structure of the graph (\textit{model selection}) and estimation of the edge parameters from the data.

Let $A\in \mathcal{S}_p$. Define the following set for $j=1,\ldots,p$:
\begin{align}\label{equ84}
&D_j(A)\defeq\{ i:A_{ij}\neq 0,\ i\neq j\}, \quad d_j(A)\defeq\text{card}(D_j(A)),\quad d(A)\defeq\max_{j=1,\ldots,p}d_j(A),
\end{align}
where $d_j(A)$ is the number of edges adjacent to the vertex $j$ (i.e., the \textit{degree} of vertex $j$), and $d(A)$ measures the maximum vertex degree. Define $S(A)\defeq \bigcup_{j=1}^{p}D_j(A)$ to be the overall off-diagonal sparsity pattern, and $s(A)\defeq \sum_{j=1}^{p}d_j(A)$ is the overall number of edges contained in the graph. Note that $\text{card}(S(A)) \leq s(A)$: when $s(A)=p(p-1)/2$ this would give a fully connected graph.

\subsection{Assumptions}
 We now list the assumptions on the model \eqref{e5.1}:
\begin{enumerate}[\textbf{({A}.1)}]
	\item \label{A1} (Spiked covariance model)
As $p \rightarrow \infty$, $\Lambda_1(\bSigma)>\Lambda_2(\bSigma)>\ldots>\Lambda_K(\bSigma)\gg \Lambda_{K+1}(\bSigma)\geq \ldots \geq \Lambda_p(\bSigma) \geq 0$, where $\Lambda_j(\bSigma)=\mathcal{O}(p)$ for $j \leq K$, while the non-spiked eigenvalues are bounded, that is, $c_0 \leq \Lambda_j(\bSigma) \leq C_0$, $j > K$ for constants $c_0, C_0 > 0$.
\end{enumerate}
\begin{enumerate}[\textbf{({A}.2)}]
	\item \label{A2}(Pervasive factors)
	There exists a positive definite $K \times K$ matrix $\breve{\bB}$ such that $\vertiii{p^{-1}\bB'\bB-\breve{\bB}}_{2}\rightarrow 0$ and $\Lambda_{\text{min}}(\breve{\bB})^{-1}=\mathcal{O}(1)$  as $p \rightarrow \infty$.
\end{enumerate}
\begin{enumerate}[\textbf{({A}.3)}]
	\item \label{A3}
	\begin{enumerate}[label=(\alph*)]
		\item $\{\bvarepsilon_t,\bf_{t}\}_{t\geq 1}$ is strictly stationary. Also, $\E{\varepsilon_{it}}=\E{\varepsilon_{it}f_{it}}=0$ $\forall i\leq p$, $j\leq K$ and $t\leq T$.
		\item There are constants $c_1, c_2 >0$ such that $\Lambda_{\text{min}}(\bSigma_{\varepsilon})>c_1$, $\vertiii{\bSigma_{\varepsilon}}_1<c_2$ and $\text{min}_{i\leq p, j\leq p} \text{var}(\varepsilon_{it}\varepsilon_{jt})>c_1$.
		\item There are $r_1,r_2>0$ and $b_1,b_2>0$ such that for any $s>0$, $i\leq p$, $j\leq K$,
		\begin{align*}
		\Pr{(\abs{\varepsilon_{it}}>s)\leq \exp\{-(s/b_1)^{r_1} \}}, \ \Pr{(\abs{f_{jt}}>s)\leq \exp\{-(s/b_2)^{r_2} \}}.
		\end{align*}
	\end{enumerate}
\end{enumerate}

We also impose the strong mixing condition. Let $\mathcal{F}_{-\infty}^{0}$ and $\mathcal{F}_{T}^{\infty}$ denote the $\sigma$-algebras that are generated by $\{(\bf_t,\bvarepsilon_{t}):t\leq 0\}$ and $\{(\bf_t,\bvarepsilon_{t}):t\geq T\}$ respectively. Define the mixing coefficient
\begin{equation}
\alpha(T)=\sup_{A\in \mathcal{F}_{-\infty}^{0}, B \in \mathcal{F}_{T}^{\infty}}\abs{\Pr{A}\Pr{B}-\Pr{AB}}.
\end{equation}
\begin{enumerate}[\textbf{({A}.4)}]
	\item \label{A4} (Strong mixing) There exists $r_3>0$ such that $3r_{1}^{-1}+1.5r_{2}^{-1}+3r_{3}^{-1}>1$, and $C>0$ satisfying, for all $T\in \mathbb{Z}^{+}$, $\alpha(T)\leq \exp (-CT^{r_3})$.
\end{enumerate}
\begin{enumerate}[\textbf{({A}.5)}]
	\item \label{A5}(Regularity conditions)
	There exists $M>0$ such that, for all $i\leq p$, $t\leq T$ and $s\leq T$, such that:
	\begin{enumerate}[label=(\alph*)]
		\item $\norm{\bb_i}_{\text{max}}<M$
		\item $\E{p^{-1/2}\{\bvarepsilon'_{s}\bvarepsilon_t-\E{\bvarepsilon'_{s}\bvarepsilon_t}\}}^4<M$ and
		\item $\E{\norm{p^{-1/2}\sum_{i=1}^{p}\bb_i\varepsilon_{it}}^4}<K^2M$.
	\end{enumerate}
\end{enumerate}
Some comments regarding the aforementioned assumptions are in order. Assumptions \ref{A1}-\ref{A4} are the same as in \cite{fan2013POET}, and assumption \ref{A5} is modified to account for the increasing number of factors. Assumption \ref{A1} divides the eigenvalues into the diverging and bounded ones. Without loss of generality, we assume that $K$ largest eigenvalues have multiplicity of 1. The assumption of a spiked covariance model is common in the literature on approximate factor models. However, we note that the model studied in this paper can be characterized as a \enquote{very spiked model}. In other words, the gap between the first $K$ eigenvalues and the rest is increasing with $p$. As pointed out by \cite{fan2018elliptical}, \ref{A1} is typically satisfied by the factor model with pervasive factors, which brings us to Assumption \ref{A2}: the factors impact a non-vanishing proportion of individual time-series. Supplemental Appendix \ref{appendixB3} explores the sensitivity of portfolios constructed using FGL when the pervasiveness assumption is relaxed, that is, when the gap between the diverging and bounded eigenvalues decreases. Assumption \ref{A3}\red{(a)} is slightly stronger than in \cite{Bai2003}, since it requires strict stationarity and non-correlation between $\{\bvarepsilon_{t}\}$ and $\{\bf_t\}$ to simplify technical calculations. In \ref{A3}\red{(b)} we require $\vertiii{\bSigma_{\varepsilon}}_1<c_2$ instead of $\lambda_{\textup{max}}(\bSigma_{\varepsilon})=\mathcal{O}(1)$ to estimate $K$ consistently. When $K$ is known, as in \cite{koike2019biased,Fan2011}, this condition can be relaxed. \ref{A3}\red{(c)} requires exponential-type tails to apply the large deviation theory to $(1/T)\sum_{t=1}^{T}\varepsilon_{it}\varepsilon_{jt}-\sigma_{\varepsilon,ij}$ and $(1/T)\sum_{t=1}^{T}f_{jt}\varepsilon_{it}$. However, in Supplemental Appendix \ref{appendixA10} we discuss the extension of our results to the setting with elliptical distribution family which is more appropriate for financial applications. Specifically, we discuss the appropriate modifications to the initial estimator of the covariance matrix of returns such that the bounds derived in this paper continue to hold. \ref{A4}-\ref{A5} are technical conditions which are needed to consistently estimate the common factors and loadings. The conditions \ref{A5}\red{(a-b)} are weaker than those in \cite{Bai2003} since our goal is to estimate a precision matrix, and \ref{A5}\red{(c)} differs from \cite{Bai2003} and \cite{Bai2006} in that the number of factors is assumed to slowly grow with $p$.

In addition, the following structural assumption on the population quantities is imposed:
\begin{enumerate}[\textbf{({B}.1)}]
	\item \label{B1} $\norm{\bSigma}_{\text{max}}=\mathcal{O}(1)$, $\norm{\bB}_{\text{max}}=\mathcal{O}(1)$, and $\norm{\boldm}_{\infty}=\mathcal{O}(1)$.
\end{enumerate}
The sparsity of $\bTheta_{\varepsilon}$ is controlled by the deterministic sequences $s_T$ and $d_T$: $s(\bTheta_{\varepsilon})=\mathcal{O}_{P}(s_T)$ for some sequence $s_T\in (0,\infty), \ T=1,2,\ldots$, and $d(\bTheta_{\varepsilon})=\mathcal{O}_{P}(d_T)$ for some sequence $d_T\in (0,\infty), \ T=1,2,\ldots$. We will impose restrictions on the growth rates of $s_T$ and $d_T$. Note that assumptions on $d_T$ are weaker since they are always satisfied when $s_T=d_T$. However, $d_T$ can generally be smaller than $s_T$. 
In contrast to \cite{fan2013POET} we do not impose sparsity on the covariance matrix of the idiosyncratic component. Instead, it is more realistic and relevant for error quantification in portfolio analysis to impose conditional sparsity on the precision matrix after the common factors are accounted for.
\subsection{The FGL Procedure}
 Recall the definition of the Weighted Graphical Lasso estimator in \eqref{e7.6} for the precision matrix of the idiosyncratic components.
Also, recall that to estimate $\bTheta$ we used equation \eqref{equa18}.
Therefore, in order to obtain the FGL estimator $\widehat{\bTheta}$ we take the following steps: \textbf{(1):} estimate unknown factors and factor loadings to get an estimator of $\bSigma_{\varepsilon}$. \textbf{(2):} use $\widehat{\bSigma}_{\varepsilon}$ to get an estimator of $\bTheta_{\varepsilon}$ in \eqref{e7.6}. \textbf{(3):} use $\widehat{\bTheta}_{\varepsilon}$ together with the estimators of factors and factor loadings from Step 1 to obtain the final precision matrix estimator $\widehat{\bTheta}$, portfolio weight estimator $\widehat{\bw}_{\xi}$, and risk exposure estimator $\widehat{\Phi}_{\xi} = \widehat{\bw}'_{\xi}\widehat{\bTheta}^{-1}\widehat{\bw}_{\xi}$ where $\xi\in\{\text{GMV, MWC, MRC}\}$.

Subsection 4.3 examines the theoretical foundations of the first step, and Subsections 4.4-4.5 are devoted to Steps 2 and 3.
\subsection{Convergence in Estimation of Factors and Loadings}
As pointed out in  \cite{Bai2003} and \cite{fan2013POET}, $K\times 1$-dimensional factor loadings $\{\bb_i\}_{i=1}^{p}$, which are the rows of the factor loadings matrix $\bB$, and $K\times 1$-dimensional common factors $\{\bf_{t}\}_{t=1}^{T}$, which are the columns of $\bF$, are not separately identifiable. Concretely, for any $K \times K$ matrix $\bH$ such that $\bH'\bH=\bI_K$, $\bB\bf_t=\bB\bH'\bH\bf_t$, therefore, we cannot identify the tuple $(\bB, \bf_t)$ from $(\bB\bH',\bH\bf_{t})$. Let $\widehat{K}\in \{1,\ldots,K_{\text{max}}\}$ denote the estimated number of factors, where $K_{\text{max}}$ is allowed to increase at a slower speed than $\min\{p,T\}$ such that $K_{\text{max}}=o(\min\{p^{1/3},T\})$ (see \cite{Li2017_Increasing_Factors} for the discussion about the rate).

 Define $\bV$ to be a $\widehat{K}\times \widehat{K}$ diagonal matrix of the first $\widehat{K}$ largest eigenvalues of the sample covariance matrix in decreasing order. Further, define a $\widehat{K}\times \widehat{K}$ matrix $\bH=(1/T)\bV^{-1}\widehat{\bF}'\bF\bB'\bB$. For $t\leq T$, $\bH\bf_t=T^{-1}\bV^{-1}\widehat{\bF}'(\bB\bf_{1},\ldots,\bB\bf_{T})'\bB\bf_{t}$, which depends only on the data $\bV^{-1}\widehat{\bF}'$ and an identifiable part of parameters $\{\bB\bf_{t}\}_{t=1}^{T}$. Hence, $\bH\bf_{t}$ does not have an identifiability problem regardless of the imposed identifiability condition.

Let $\gamma^{-1}=3r_{1}^{-1}+1.5r_{2}^{-1}+r_{3}^{-1}+1$. The following theorem is an extension of the results in \cite{fan2013POET} for the case when the number of factors is unknown and is allowed to grow. Proofs of all the theorems are in Supplemental Appendix \ref{appendixA}.

\begin{thm}\label{theor1}
	Suppose that $K_{\textup{max}}=o(\min\{p^{1/3},T\})$, $K^3\log p=o(T^{\gamma/6})$, $KT=o(p^2)$ and Assumptions \ref{A1}-\ref{A5} and \ref{B1} hold. Let $\omega_{1T}\defeq K^{3/2}\sqrt{\log p/T} +K/\sqrt{p}$ and $\omega_{2T}\defeq K/\sqrt{T}+KT^{1/4}/\sqrt{p}$. Then $\max_{i\leq p} \norm{\widehat{\bb}_i-\bH\bb_i}=\mathcal{O}_P(\omega_{1T})$ and $\max_{t\leq T} \norm{\widehat{\bf}_t-\bH\bf_t}=\mathcal{O}_P(\omega_{2T})$.
\end{thm}
The conditions $K^3\log p=o(T^{\gamma/6})$, $KT=o(p^2)$ are similar to \cite{fan2013POET}, the difference arises due to the fact that we do not fix $K$, hence, in addition to the factor loadings, there are $KT$ factors to estimate. Therefore, the number of parameters introduced by the unknown growing factors should not be \enquote{too large}, such that we can consistently estimate them uniformly. The growth rate of the number of factors is controlled by $K_{\text{max}}=o(\min\{p^{1/3},T\})$.

The bounds derived in Theorem \ref{theor1} help us establish the convergence properties of the estimated idiosyncratic covariance, $\widehat{\bSigma}_{\varepsilon}$, and precision matrix $\widehat{\bTheta}_{\varepsilon}$ which are presented in the next theorem:
\begin{thm}\label{theor2}
Let $\omega_{3T}\defeq K^{2}\sqrt{\log p/T} +K^3/\sqrt{p}$. Under the assumptions of Theorem \ref{theor1} and with $\lambda \asymp \omega_{3T}$ (where $\lambda$ is the tuning parameter in \eqref{e7.6}), the estimator $\widehat{\bSigma}_{\varepsilon}$ obtained by estimating factor model in \eqref{5.2} satisfies $\norm{\widehat{\bSigma}_{\varepsilon}-\bSigma_{\varepsilon}}_{\textup{max}}=\mathcal{O}_P(\omega_{3T}).$ Let $\varrho_{T}$ be a sequence of positive-valued random variables such that $\varrho_{T}^{-1}\omega_{3T}\xrightarrow{p}0$. If $s_T\varrho_{T}\xrightarrow{\text{p}}0$, then $\vertiii{\widehat{\bTheta}_{\varepsilon}-\bTheta_{\varepsilon}}_{l}=\mathcal{O}_P(\varrho_{T}s_T)$ as $T \rightarrow \infty$ for any $l \in [1,\infty]$.
\end{thm}
Note that the term containing $K^3/\sqrt{p}$ arises due to the need to estimate unknown factors. \cite{Fan2011} obtained a similar rate but for the case when factors are observable (in their work, $\omega_{3T}= K^{1/2}\sqrt{\log p/T}$). The second part of Theorem \ref{theor2} is based on the relationship between the convergence rates of the estimated covariance and precision matrices established in \cite{Sara2018} (Theorem 14.1.3). \cite{koike2019biased} obtained the convergence rate when factors are observable: the rate obtained in our paper is slower due to the fact that factors need to be estimated (concretely, the rate under observable factors would satisfy $\varrho_{T}^{-1}\sqrt{K\log p /T}\xrightarrow{p}0$ ). We now comment on the optimality of the rate in Theorem \ref{theor2}: as pointed out in \cite{koike2019biased}, in the standard Gaussian setting without factor structure, the minimax optimal rate is $d(\bTheta_{\varepsilon})\sqrt{\log p/T}$, which can be faster than the rate obtained in Theorem \ref{theor2} if $d(\bTheta_{\varepsilon})<s_T$. Using penalized nodewise regression could help achieve this faster rate. However, our empirical application to the monthly stock returns demonstrated superior performance of the Weighted Graphical Lasso compared to the nodewise regression in terms of the out-of-sample Sharpe Ratio and portfolio risk. Hence, in order not to divert the focus of this paper, we leave the theoretical properties of the nodewise regression for future research.
\subsection{Convergence in Estimation of Precision Matrix and Portfolio Weights}
Having established the convergence properties of $\widehat{\bSigma}_{\varepsilon}$ and $\widehat{\bTheta}_{\varepsilon}$, we now move to the estimation of the precision matrix of the factor-adjusted returns in equation \eqref{equa18}.
\begin{thm} \label{theor3}
Under the assumptions of Theorem \ref{theor2}, if $d_Ts_T\varrho_{T}\xrightarrow{\text{p}}0$, then $\vertiii{\widehat{\bTheta}-\bTheta}_{2}=\mathcal{O}_P(\varrho_{T}s_T)$ and $\vertiii{\widehat{\bTheta}-\bTheta}_{1}=\mathcal{O}_P(\varrho_{T}d_TK^{3/2}s_T )$.
\end{thm}
Note that since, by construction, the precision matrix obtained using the Factor Graphical Lasso is symmetric, $\vertiii{\widehat{\bTheta}-\bTheta}_{\infty}$ can be trivially obtained from the above theorem.

Using Theorem \ref{theor3}, we can then establish the consistency of the estimated weights of portfolios based on the Factor Graphical Lasso.
\begin{thm} \label{theor4}
	Under the assumptions of Theorem \ref{theor3}, we additionally assume $\vertiii{\bTheta}_2=\mathcal{O}(1)$ (this additional requirement essentially imposes $\Lambda_p(\bSigma)>0$ in \ref{A1}), and $\varrho_{T}d_{T}^{2}s_T=o(1)$. Procedure \ref{alg2} consistently estimates portfolio weights in \eqref{eq2}, \eqref{eq3} and \eqref{ee20}:\\ $\norm{\widehat{\bw}_{\text{GMV}}-\bw_{\text{GMV}}}_1=\mathcal{O}_P\Big(\varrho_{T}d_{T}^2K^{3}s_T\Big)=o_P(1)$, $\norm{\widehat{\bw}_{\text{MWC}}-\bw_{\text{MWC}}}_1=\mathcal{O}_P(\varrho_{T}d_{T}^2K^{3}s_T)=o_P(1)$, and $\norm{\widehat{\bw}_{\text{MRC}}-\bw_{\text{MRC}}}_1= \mathcal{O}_P \Big( d_{T}^{3/2}K^{3}\cdot\lbrack\varrho_{T}s_T\rbrack^{1/2} \Big)=o_P(1)$.
\end{thm}
We now comment on the rates in Theorem \ref{theor4}: first, the rates obtained by \cite{Caner2019} for GMV and MWC formulations, when no factor structure of stock returns is assumed, require $s(\bTheta)^{3/2}\sqrt{\log p/T}=o_P(1)$, where the authors imposed sparsity on the precision matrix of stock returns, $\bTheta$. Therefore, if the precision matrix of stock returns is not sparse, portfolio weights can be consistently estimated only if $p$ is less than $T^{1/3}$ (since $(p-1)^{3/2}\sqrt{\log p/T}=o(1)$ is required to ensure consistent estimation of portfolio weights). Our result in Theorem \ref{theor4} improves this rate and shows that as long as $d_{T}^{2}s_TK^{3}\sqrt{\log p/T}=o_P(1)$ we can consistently estimate weights of the financial portfolio. Specifically, when the precision of the factor-adjusted returns is sparse, we can consistently estimate portfolio weights when $p>T$ \textit{without} assuming sparsity on $\bSigma$ or $\bTheta$.
Second, note that GMV and MWC weights converge slightly slower than MRC weight. This result is further supported by our simulations presented in the next section.
\subsection{Implications on Portfolio Risk Exposure}
Having examined the properties of portfolio weights, it is natural to comment on the portfolio variance estimation error. It is determined by the errors in two components: the estimated covariance matrix and the estimated portfolio weights. Define $a = \biota'_{p}\bTheta\biota_p/p$, $b = \biota'_{p}\bTheta\boldm/p$, $d = \boldm'\bTheta\boldm/p$, $g = \sqrt{\boldm'\bTheta\boldm}/p$ and $\widehat{a} = \biota'_{p}\widehat{\bTheta}\biota_p/p$,  $\widehat{b} = \biota'_{p}\widehat{\bTheta}\widehat{\boldm}/p$, $\widehat{d}=\widehat{\boldm}'\widehat{\bTheta}\widehat{\boldm}/p$, $\widehat{g}=\sqrt{\widehat{\boldm}'\widehat{\bTheta}\widehat{\boldm}}/p$. Define $\Phi_{\text{GMV}} = \bw_{GMV}'\bSigma\bw_{GMV}=(pa)^{-1}$ to be the global minimum variance, $\Phi_{\text{MWC}} = \bw_{MWC}'\bSigma\bw_{MWC}=p^{-1}\Big[\frac{a\mu^2-2b\mu+d}{ad-b^2}\Big]$ is the MWC portfolio variance, and $\Phi_{\text{MRC}} = \bw_{MRC}'\bSigma\bw_{MRC}=\sigma^2(pg)$ is the MRC portfolio variance. We use the terms variance and risk exposure interchangeably. Let $\widehat{\Phi}_{\text{GMV}}$, $\widehat{\Phi}_{\text{MWC}}$, and $\widehat{\Phi}_{\text{MRC}}$ be the sample counterparts of the respective portfolio variances. The expressions for $\Phi_{\text{GMV}}$ and $\Phi_{\text{MWC}}$ were derived in \cite{FANLV2008} and \cite{Caner2019}. Theorem \ref{theor5} establishes the consistency of a large portfolio’s variance estimator.
	\begin{thm} \label{theor5}
	Under the assumptions of Theorem \ref{theor3}, FGL consistently estimates GMV, MWC, and MRC portfolio variance:\\ $\abs{\widehat{\Phi}_{\text{GMV}}/\Phi_{\text{GMV}} -1 }=\mathcal{O}_P(\varrho_{T}d_Ts_TK^{3/2} )=o_P(1)$,\\
		$\abs{\widehat{\Phi}_{\text{MWC}}/\Phi_{\text{MWC}} -1 }=\mathcal{O}_P(\varrho_{T}d_Ts_TK^{3/2} )=o_P(1)$,\\
		$\abs{\widehat{\Phi}_{\text{MRC}}/\Phi_{\text{MRC}} -1 }=\mathcal{O}_P\Big(\lbrack\varrho_{T}d_Ts_TK^{3/2}\rbrack^{1/2}\Big)=o_P(1)$.
\end{thm}
\cite{Caner2019} derived a similar result for $\Phi_{\text{GMV}}$ and $\Phi_{\text{MWC}}$ under the assumption that precision matrix of stock returns is sparse. Also, \cite{DING2020} derived the bounds for $\Phi_{\text{GMV}}$ under the factor structure assuming sparse covariance matrix of idiosyncratic components and gross exposure constraint on portfolio weights which limits negative positions.

The empirical application in Section 6 reveals that the portfolios constructed using MRC formulation have higher risk compared with GMV and MWC alternatives: using monthly and daily returns of the components of S\&P500 index, MRC portfolios exhibit higher out-of-sample risk and return compared to the alternative formulations. Furthermore, the empirical exercise demonstrates that the higher return of MRC portfolios outweighs higher risk for the monthly data which is evidenced by the increased out-of-sample Sharpe Ratio.

\section{Monte Carlo}
In order to validate our theoretical results, we perform several simulation studies which are divided into four parts. The first set of results computes the empirical convergence rates and compares them with the theoretical expressions derived in Theorems \ref{theor3}-\ref{theor5}. The second set of results compares the performance of the FGL with several alternative models for estimating covariance and precision matrix. To highlight the benefit of using the information about factor structure as opposed to standard graphical models, we include Graphical Lasso by \cite{GLASSO} (GL) that does not account for the factor structure. To explore the benefits of using FGL for error quantification in \eqref{equa18}, we consider several alternative estimators of covariance/precision matrix of the idiosyncratic component in \eqref{equa18}: (1) linear shrinkage estimator of covariance developed by \cite{Ledoit2004} further referred to as Factor LW or FLW; (2) nonlinear shrinkage estimator of covariance by \cite{ledoit2017nonlinear} (Factor NLW or FNLW); (3) POET (\cite{fan2013POET}); (4) constrained $\ell_{1}$-minimization for inverse matrix estimator, Clime (\cite{cai2011constrained}) (Factor Clime or FClime). Furthermore, we discovered that in certain setups the estimator of covariance produced by POET is not positive definite. In such cases we use the matrix symmetrization procedure as in \cite{fan2018elliptical} and then use eigenvalue cleaning as in \cite{Callot2017} and \cite{Hautsch2012}. This estimator is referred to as Projected POET; it coincides with POET when the covariance estimator produced by the latter is positive definite. The third set of results examines the performance of FGL and Robust FGL (described in Supplemental Appendix \ref{appendixA10}) when the dependent variable follows elliptical distribution. The fourth set of results explores the sensitivity of portfolios constructed using different covariance and precision estimators of interest when the pervasiveness assumption \ref{A2} is relaxed, that is, when the gap between the diverging and bounded eigenvalues decreases.  All exercises in this section use 100 Monte Carlo simulations.
We consider the following setup: let $p = T^{\delta}$, $\delta = 0.85$, $K = 2(\log T)^{0.5}$ and $T = \lbrack 2^h \rbrack, \ \text{for} \ h=7,7.5,8,\ldots,9.5$. A sparse precision matrix of the idiosyncratic components is constructed as follows: we first generate the adjacency matrix using a random graph structure. Define a $p \times p$ adjacency matrix $\bA_{\varepsilon}$ which is used to represent the structure of the graph:
\begin{align}
a_{\varepsilon,ij}=\begin{cases}
1, & \text{for} \ i\neq j\ \ \text{with probability $q$},\\
0, & \text{otherwise.}
\end{cases}
\end{align}
Let $a_{\varepsilon,ij}$ denote the $i,j$-th element of the adjacency matrix $\bA_{\varepsilon}$. We set $a_{\varepsilon,ij} = a_{\varepsilon,ji}=1, \ \text{for} \ i\neq j$ with probability $q$, and $0$ otherwise. Such structure results in $s_T = p(p-1)q/2$ edges in the graph. To control sparsity, we set $q = 1/(pT^{0.8})$, which makes $s_T = \mathcal{O}(T^{0.05})$. The adjacency matrix has all diagonal elements equal to zero. Hence, to obtain a positive definite precision matrix we apply the procedure described in \cite{HUGE}: using their notation, $\bTheta_{\varepsilon}=\bA_{\varepsilon}\cdot v+\bI(\abs{\tau}+0.1+u)$, where $u>0$ is a positive number added to the diagonal of the precision matrix to control the magnitude of partial correlations, $v$ controls the magnitude of partial correlations with $u$, and $\tau$ is the smallest eigenvalue of $\bA_{\varepsilon}\cdot v$.
In our simulations we use $u=0.1$ and $v=0.3$.

Factors are assumed to have the following structure:
\begin{align} 
&\bf_{t}=\phi_f\bf_{t-1}+\bzeta_t \label{e51}\\
&\underbrace{\br_t}_{p \times 1}=\boldm + \bB \underbrace{\bf_t}_{K\times 1}+\ \bvarepsilon_t, \quad t=1,\ldots,T \label{e52}
\end{align}
 where $m_i \sim \mathcal{N}(1,1)$ independently for each $i=1,\ldots,p$, $\bvarepsilon_{t}$ is a $p \times 1$ random vector of idiosyncratic errors following $\mathcal{N}(\bm{0},\bSigma_{\varepsilon})$, with sparse $\bTheta_{\varepsilon}$ that has a random graph structure described above, $\bf_{t}$ is a $K \times 1$ vector of factors, $\phi_f$ is an autoregressive parameter in the factors which is a scalar for simplicity, $\bB$ is a $p\times K$ matrix of factor loadings, $\bzeta_t$ is a $K \times 1$ random vector with each component independently following  $\mathcal{N}(0,\sigma^{2}_{\zeta})$. To create $\bB$ in \eqref{e52} we take the first $K$ rows of an upper triangular matrix from a Cholesky decomposition of the $p \times p$ Toeplitz matrix parameterized by $\rho$. For the first set of results we set $\rho = 0.2$, $\phi_f = 0.2$ and $\sigma^{2}_{\zeta} = 1$. The specification in \eqref{e52} leads to the low-rank plus sparse decomposition of the covariance matrix of stock returns $\br_t$.

As a first exercise, we compare the empirical and theoretical convergence rates of the precision matrix, portfolio weights and exposure.
A detailed description of the procedure and the simulation results is provided in Supplemental Appendix \ref{appendixB1}. We confirm that the empirical rates and theoretical rates from Theorems \ref{theor3}-\ref{theor5} are matched.

As a second exercise, we compare the performance of FGL with the alternative models listed at the beginning of this section. We consider two cases: \textbf{Case 1} is the same as for the first set of simulations ($p<T$): $p = T^{\delta}$, $\delta = 0.85$, $K = 2(\log T)^{0.5}$, $s_T = \mathcal{O}(T^{0.05})$. \textbf{Case 2} captures the cases when $p>T$ with $p = 3\cdot T^{\delta}$, $\delta = 0.85$, all else equal. The results for Case 2 are reported in \autoref{f4}-\ref{f5a}, and Case 1 is located in Supplemental Appendix \ref{appendixB1A}. FGL demonstrates superior performance for estimating precision matrix and portfolio weights in both cases, exhibiting consistency for both Case 1 and Case 2 settings. Also, FGL outperforms GL for estimating portfolio exposure and consistently estimates the latter, however, depending on the case under consideration some alternative models produce lower averaged error.


As a third exercise, we examine the performance of FGL and Robust FGL (described in Supplemental Appendix \ref{appendixA10}) when the dependent variable follows elliptical distributions. A detailed description of the data generating process (DGP) and simulation results are provided in Supplemental Appendix \ref{appendixB2}. We find that the performance of FGL for estimating the precision matrix is comparable with that of Robust FGL: this suggests that our FGL algorithm is robust to heavy-tailed distributions even without additional modifications.

As a final exercise, we explore the sensitivity of portfolios constructed using different covariance and precision estimators of interest when the pervasiveness assumption \ref{A2} is relaxed. A detailed description of the data generating process (DGP) and simulation results are provided in Supplemental Appendix \ref{appendixB3}. We verify that FGL exhibits robust performance when the gap between the diverging and bounded eigenvalues decreases. In contrast, POET and Projected POET are most sensitive to relaxing pervasiveness assumption which is consistent with our empirical findings and also with the simulation results by \cite{onatski2013POET}.
\section{Empirical Application}
In this section we examine the performance of the Factor Graphical Lasso for constructing a financial portfolio using daily data. The description and empirical results for monthly data can be found in Supplemental Appendix \ref{appendixC}. We first describe the data and the estimation methodology, then we list four metrics commonly reported in the finance literature, and, finally, we present the results.
\subsection{Data}
We use daily returns of the components of the S\&P500 index. The data on historical S\&P500 constituents and stock returns is fetched from CRSP and Compustat using SAS interface. For the daily data the full sample size has 5040 observations on 420 stocks from January 20, 2000 - January 31, 2020. We use January 20, 2000 - January 24, 2002 (504 obs) as the first training (estimation) period and January 25, 2002 - January 31, 2020 (4536 obs) as the out-of-sample (OOS) test period. Supplemental Appendix \ref{appendixC1} examines the performance of different competing methods for longer training periods. We roll the estimation window (training periods) over the test sample to rebalance the portfolios monthly. At the end of each month, prior to portfolio construction, we remove stocks with less than 2 years of historical stock return data. The performance of the competing models is compared with the Index -- the composite S\&P500 index listed as \textsuperscript{$\wedge$}GSPC.
  We take the risk-free rate and Fama/French factors from \href{https://mba.tuck.dartmouth.edu/pages/faculty/ken.french/data_library.html}{Kenneth R. French's data library.}
\subsection{Performance Measures}
Similarly to \cite{Caner2019}, we consider four metrics commonly reported in the finance literature: the Sharpe Ratio, the portfolio turnover, the average return and the risk of a portfolio (which is defined as the square root of the out-of-sample variance of the portfolio). We consider two scenarios: with and without transaction costs. Let $T$ denote the total number of observations, the training sample consists of $m=504$ observations, and the test sample is $n=T-m$.

When transaction costs are not taken into account, the out-of-sample average portfolio return, variance and SR are
\begin{align}\label{sharperatio}
\hat{\mu}_{\text{test}}=\frac{1}{n}\sum_{t=m}^{T-1}\widehat{\bw}'_{t}\br_{t+1}, \ \hat{\sigma}_{\text{test}}^{2}=\frac{1}{n-1}\sum_{t=m}^{T-1}(\widehat{\bw}'_{t}\br_{t+1}-\hat{\mu}_{\text{test}})^2,\
\text{SR} = \hat{\mu}_{\text{test}}/ \hat{\sigma}_{\text{test}}.
\end{align}

When transaction costs are considered, we follow \cite{ban2018machine}, \cite{Caner2019}, \cite{Demiguel2009optimal}, and \cite{Li2015sparse} to account for the transaction costs, further denoted as $\textup{tc}$. In line with the aforementioned papers, we set $\textup{tc}=10 \text{bps}$. Define the excess portfolio at time $t+1$ with transaction costs (tc) as 
\begin{align}
r_{t+1,\text{portfolio}} = \ \widehat{\bw}'_{t}\br_{t+1}-\textup{tc}(1+\widehat{\bw}'_{t}\br_{t+1})\sum_{j=1}^{p}\abs{\hat{w}_{t+1,j}-\hat{w}_{t,j}^{+}},
\end{align} 
where
\begin{align}\hat{w}_{t,j}^{+}=\hat{w}_{t,j}\frac{1+r_{t+1,j}+r^{f}_{t+1}}{1+r_{t+1,\text{portfolio}}+r^{f}_{t+1}},
\end{align} 
$r_{t+1,j}+r^{f}_{t+1}$ is sum of the excess return of the $j$-th asset and risk-free rate, and $r_{t+1,\text{portfolio}}+r^{f}_{t+1}$ is the sum of the excess return of the portfolio and risk-free rate. The out-of-sample average portfolio return, variance, Sharpe Ratio and turnover are defined accordingly:
\begin{align}
\hat{\mu}_{\text{test,tc}}=\frac{1}{n}\sum_{t=m}^{T-1}r_{t,\text{portfolio}},\ &\hat{\sigma}_{\text{test,tc}}^{2}=\frac{1}{n-1}\sum_{t=m}^{T-1}(r_{t,\text{portfolio}}-\hat{\mu}_{\text{test,tc}})^2,\
\text{SR}_{\text{tc}} =\hat{\mu}_{\text{test,tc}}/ \hat{\sigma}_{\text{test,tc}},\\
&\text{Turnover}=\frac{1}{n}\sum_{t=m}^{T-1}\sum_{j=1}^{p}\abs{\hat{w}_{t+1,j}-\hat{w}_{t,j}^{+}}.
\end{align}
\subsection{Description of Empirical Design}
In the empirical application for constructing financial portfolio we consider two scenarios, when the factors are unknown and estimated using the standard PCA (statistical factors), and when the factors are known. The number of statistical factors, $\hat{K}$, is estimated in accordance with Remark \ref{remark1} in Supplemental Appendix \ref{appendixC0}. For the scenario with known factors we include up to 5 Fama-French factors: FF1 includes the excess return on the market, FF3 includes FF1 plus size factor (Small Minus Big, SMB) and value factor (High Minus Low, HML), and FF5 includes FF3 plus profitability factor (Robust Minus Weak, RMW) and risk factor (Conservative Minus Agressive, CMA).\\
\indent We examine the performance of Factor Graphical Lasso for three alternative portfolio allocations \eqref{eq2}, \eqref{eq3} and \eqref{ee20} and compare it with the equal-weighted portfolio (EW), index portfolio (Index), FClime, FLW, FNLW (as in the simulations, we use alternative covariance and precision estimators that incorporate the factor structure through Sherman-Morrison inversion formula), POET, Projected POET, and factor models without sparsity restriction on the residual risk (FF1, FF3, and FF5).\\
\indent In \autoref{tab3} and Supplemental Appendix \ref{appendixC}, we report the daily and monthly portfolio performance for three alternative portfolio allocations in \eqref{eq2}, \eqref{eq3} and \eqref{ee20}. We consider a relatively risk-averse investor in a sense that they are willing to tolerate no more risk than that incurred by holding the S\&P500 Index: the target level of risk for the weight-constrained and risk-constrained Markowitz portfolio (MWC and MRC) is set at $\sigma=0.013$ which is the standard deviation of the daily excess returns of the S\&P500 index in the first training set. A return target $\mu=0.0378\%$ which is equivalent to $10\%$ yearly return when compounded. Transaction costs for each individual stock are set to be a constant $0.1\%$. Supplemental Appendix \ref{appendixC1} provides the results for less risk-averse investors that have higher target levels of risk and return for both monthly and daily data. 

To compare the relative performance of investment strategies induced by different precision matrix estimators, we use a stepwise multiple testing procedure developed in \cite{RomanoWolf2005} and further covered in \cite{RomanoWolf2016}. Let $\text{SR}^{P} = \mu_{\text{test}}/\sigma_{\text{test}}$ be the population counterpart of the sample Sharpe Ratio defined in \eqref{sharperatio}. We compare each strategy $s$, $1\leq s \leq S$, with the benchmark (Index) strategy, indexed as $S+1$. Define $\chi_{s} \defeq \text{SR}_{s}^{P} - \text{SR}_{S+1}^{P}$. The test statistic is $\hat{\chi}_{s} \defeq \text{SR}_{s} - \text{SR}_{S+1}$. For a given strategy $s$, we consider the individual testing problem $\mathbb{H}_0: \chi_{s} \leq 0 \quad \text{vs.} \quad \mathbb{H}_A: \chi_{s} > 0$.
Using the stepwise multiple testing procedure we aim at identifying as many strategies as possible for which $\chi_{s} > 0$: we relabel the strategies according to the size of the individual test statistics, from largest to smallest, and make the individual decisions in a stepdown manner starting with the null hypothesis that corresponds to the largest test statistic. P-values for competing methods are reported in the tables with empirical results. We note that by construction of the stepwise multiple testing procedure, the resulting p-values are relatively conservative, consistent with Remark 3.1 of \cite{RomanoWolf2005}.
\subsection{Empirical Results}
This section explores the performance of the Factor Graphical Lasso for the financial portfolio using daily data.

Let us summarize the results for daily data in \autoref{tab3}: \textbf{(1)} MRC portfolios produce higher return and higher risk, compared to MWC and GMV. However, the out-of-sample Sharpe Ratio for MRC is lower than that of MWC and GMV, which implies that the higher risk of MRC portfolios is not fully compensated by the higher return. \textbf{(2)} FGL outperforms all the competitors, including EW and Index. Specifically, our method has the lowest risk and turnover (compared to FClime, FLW, FNLW and POET), and the highest out-of-sample Sharpe Ratio compared with all alternative methods. \textbf{(3)} The implementation of POET for MRC resulted in the erratic behavior of this method for estimating portfolio weights; many entries in the weight matrix had \enquote{NaN} entries. We elaborate on the reasons behind such performance below.
\textbf{(4)} Using the observable Fama-French factors in the FGL, in general, produces portfolios with higher return and higher out-of-sample Sharpe Ratio compared to the portfolios based on statistical factors. Interestingly, this increase in return is not followed by higher risk. \textbf{(5)} FGL strongly dominates all factor models that do not impose sparsity on the precision of the idiosyncratic component.
The results for monthly data are provided in Supplemental Appendix \ref{appendixC}: all the conclusions are similar to the ones for daily data.

We now examine possible reasons behind the observed puzzling behavior of POET and Projected POET. The erratic behavior of the former is caused by the fact that POET estimator of covariance matrix was not positive-definite which produced poor estimates of GMV and MWC weights and made it infeasible to compute MRC weights (recall, by construction MRC weight in \eqref{ee20} requires taking a square root). To explore deteriorated behavior of Projected POET, let us highlight two findings outlined by the existing closely related literature. First, \cite{bailey2020factorstrength} examined ``pervasiveness" degree, or strength, of 146 factors commonly used in the empirical finance literature, and found that only the market factor was strong, while all other factors were semi-strong. This indicates that the factor pervasiveness assumption \ref{A2} might be unrealistic in practice. Second, as pointed out by \cite{onatski2013POET}, ``the quality of POET dramatically deteriorates as the systematic-idiosyncratic eigenvalue gap becomes small". Therefore, being guided by the two aforementioned findings, we attribute deteriorated performance of POET and Projected POET to the decreased gap between the diverging and bounded eigenvalues documented in the past studies on financial returns. High sensitivity of these two covariance estimators in such settings was further supported by our additional simulation study (Supplemental Appendix \ref{appendixB3}) examining the robustness of portfolios constructed using different covariance and precision estimators.

\autoref{tab5} compares the performance of MRC portfolios for the daily data for different time periods of interesting episodes in terms of the cumulative excess return (CER), risk, and SR. To demonstrate the performance of all methods during the periods of recession and expansion, we chose four periods and recorded CER for the whole year in each period of interest. Two years, 2002 and 2008 correspond to the recession periods, which is why we we refer to them as \enquote{Downturns}. We note that the references to Argentine Great Depression and The Financial Crisis do not intend to limit these economic downturns to only one year. They merely provide the context for the recessions. The other two years, 2017 and 2019, correspond to the years which were relatively favorable to the stock market (\enquote{Booms}). Overall, it is easier to beat the Index in Downturns than in Booms. In most cases FGL shows superior performance in terms of CER and SR for Downturn \#1, Boom \#1 and Boom \#2. For Downturn \#2, even though FGL has the highest CER, its SR is smaller than SR of some other competing methods. One explanation would be the following: as evidenced by high risk of the competing methods during Boom \#2, there were high positive and negative returns during the period, with high returns driving up the average used in computing the SR. However, if one were to use the alternative strategies ignoring CER statistics, then the return on the money deposited at the beginning of 2008 would either be negative (e.g. FClime, Projected POET) or smaller than the CER of FGL-based strategies. This exercise demonstrates that SR statistics alone, especially during recession periods characterized by higher volatility, could be misleading. Another interesting finding from such exercise is that FGL exhibits smaller risk compared to most competing methods even during the periods of recession, which holds for all portfolio formulations. This allows FGL to minimize cumulative losses during economic downturns. Subperiod analyses for MWC and GMV portfolio formulations is presented in Supplemental Appendix \ref{appendixC3}.
%
%
\section{Conclusion}
 In this paper, we propose a new conditional precision matrix estimator for the excess returns under the approximate factor model with unobserved factors that combines the benefits of graphical models and factor structure. We established consistency of FGL in the spectral and $\ell_{1}$ matrix norms. In addition, we proved consistency of the portfolio weights and risk exposure for three formulations of the optimal portfolio allocation without assuming sparsity on the covariance or precision matrix of stock returns. All theoretical results established in this paper hold for a wide range of distributions: sub-Gaussian family (including Gaussian) and elliptical family. Our simulations demonstrate that FGL is robust to very heavy-tailed distributions, which makes our method suitable for the financial applications. Furthermore, we demonstrate that in contrast to POET and Projected POET, the success of the proposed method does not heavily depend on the factor pervasiveness assumption: FGL is robust to the scenarios when the gap between the diverging and bounded eigenvalues decreases.

The empirical exercise uses the constituents of the S\&P500 index and demonstrates superior performance of FGL compared to several alternative models for estimating precision (FClime) and covariance (FLW, FNLW, POET) matrices, Equal-Weighted (EW) portfolio and Index portfolio in terms of the OOS SR and risk. This result is robust to monthly and daily data.
We examine three portfolio formulations and discover that the only portfolios that produce positive CER during recessions are the ones that relax the constraint requiring portfolio weights sum up to one.



\cleardoublepage

\phantomsection

\addcontentsline{toc}{section}{References}
\setlength{\baselineskip}{14pt}
\bibliographystyle{apalike}
\bibliography{RevisedFGL_Arxiv}
\cleardoublepage
\begin{figure}[!htbp]
	\phantomsection
	\addcontentsline{toc}{section}{Figures}
	\centering
	\includegraphics[width=0.98\textwidth]{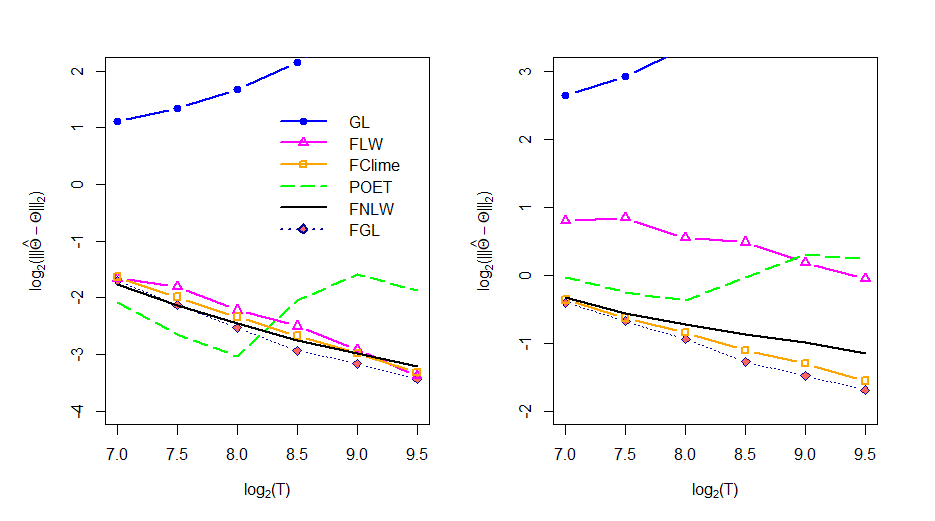}
	\bigskip
	\caption{\textbf{Averaged errors of the estimators of $\bTheta$ for Case 2 on logarithmic scale: $p = 3\cdot T^{0.85}$, $K = 2(\log T)^{0.5}$, $s_T = \mathcal{O}(T^{0.05})$.}}
	\label{f4}
\end{figure}
\begin{figure}[!htbp]
	\centering
	\includegraphics[width=0.98\textwidth]{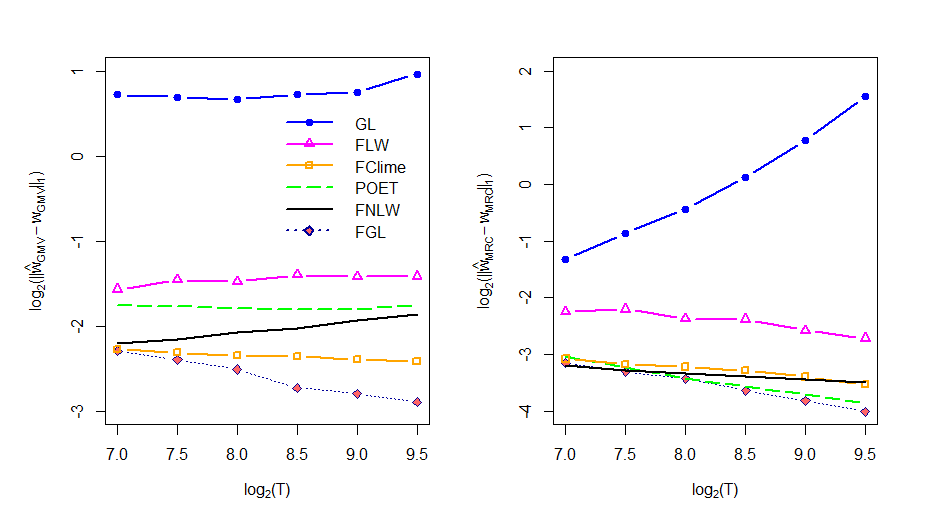}
	\bigskip
	\caption{\textbf{Averaged errors of the estimators of $\bw_{\text{GMV}}$ (left) and $\bw_{\text{MRC}}$ (right) for Case 2 on logarithmic scale: $p = 3\cdot T^{0.85}$, $K = 2(\log T)^{0.5}$, $s_T = \mathcal{O}(T^{0.05})$.}}
	\label{f5}
\end{figure}
\newpage
\begin{figure}[!htbp]
	\centering
	\includegraphics[width=0.98\textwidth]{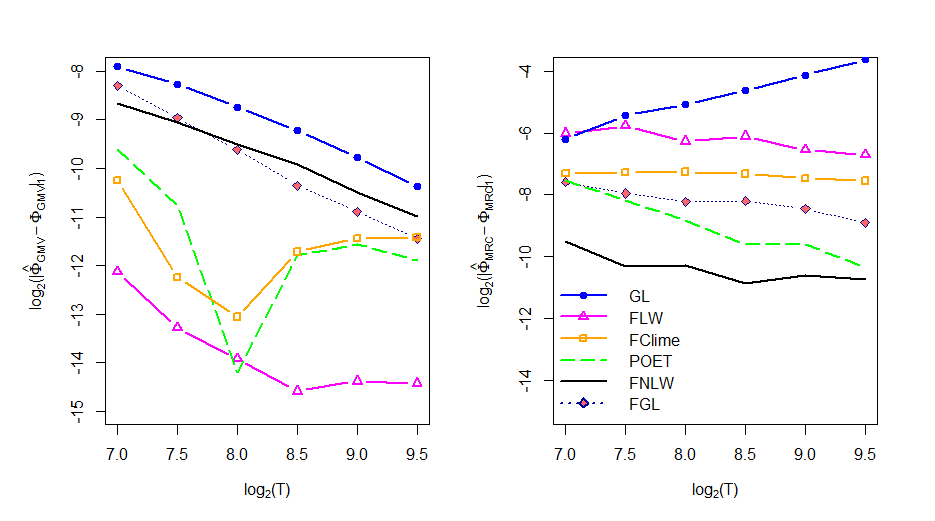}
	\bigskip
	\caption{\textbf{Averaged errors of the estimators of $\Phi_{\text{GMV}}$ (left) and $\Phi_{\text{MRC}}$ (right) for Case 2 on logarithmic scale: $p = 3\cdot T^{0.85}$, $K = 2(\log T)^{0.5}$, $s_T = \mathcal{O}(T^{0.05})$.}}
	\label{f5a}
\end{figure}
	\begin{landscape}
	\begin{table}[]
		\addcontentsline{toc}{section}{Tables}
		\centering
		\caption{\small {Daily portfolio returns, risk, SR and turnover. In the upper part corresponding to the results w/o transactions costs, p-values are in parentheses. In the lower part corresponding to the results with transaction costs, $^{***}$ indicates p-value $<$ 0.01, $^{**}$ indicates p-value  $<$ 0.05, and $^{*}$ indicates p-value $<$ 0.10. In-sample: January 20, 2000 - January 24, 2002  (504 obs), Out-of-sample: January 17, 2002 - January 31, 2020 (4536 obs).}}
		\label{tab3}
		\resizebox{0.9\textwidth}{!}{%
\begin{tabular}{ccccccccccccc} 
	\toprule
	& \multicolumn{4}{c}{Markowitz Risk-Constrained} & \multicolumn{4}{c}{Markowitz Weight-Constrained} & \multicolumn{4}{c}{Global Minimum-Variance} \\ 
	\midrule
	& \textbf{Return} & \textbf{Risk} & \textbf{SR} & \textbf{Turnover} & \textbf{Return} & \textbf{Risk} & \textbf{SR} & \textbf{Turnover} & \textbf{Return} & \textbf{Risk} & \textbf{SR} & \textbf{Turnover} \\ 
	\midrule
	\textbf{Without TC} &  &  &  &  &  &  &  &  &  &  &  &  \\
	EW & 2.33E-04 & 1.90E-02 & 0.0123 & - & 2.33E-04 & 1.90E-02 & 0.0123 & - & 2.33E-04 & 1.90E-02 & 0.0123 & - \\
	Index & 1.86E-04 & 1.17E-02 & 0.0159 & - & 1.86E-04 & 1.17E-02 & 0.0159 & - & 1.86E-04 & 1.17E-02 & 0.0159 & - \\
	FGL & 8.12E-04 & 2.66E-02 & \begin{tabular}[c]{@{}c@{}}0.0305\\(0.0579)\end{tabular} & - & 2.95E-04 & 8.21E-03 & \begin{tabular}[c]{@{}c@{}}0.0360\\(0.024)\end{tabular} & - & 2.94E-04 & 7.51E-03 & \begin{tabular}[c]{@{}c@{}}0.0392\\(0.0279)\end{tabular} & - \\
	FClime & 2.15E-03 & 8.46E-02 & \begin{tabular}[c]{@{}c@{}}0.0254\\(0.0758)\end{tabular} & - & 2.02E-04 & 9.85E-03 & \begin{tabular}[c]{@{}c@{}}0.0205\\(0.0299)\end{tabular} & - & 2.73E-04 & 1.07E-02 & \begin{tabular}[c]{@{}c@{}}0.0255\\(0.0419)\end{tabular} & - \\
	FLW & 4.34E-04 & 2.65E-02 & \begin{tabular}[c]{@{}c@{}}0.0164\\(0.1782)\end{tabular} & - & 3.12E-04 & 9.96E-03 & \begin{tabular}[c]{@{}c@{}}0.0313\\(0.024)\end{tabular} & - & 3.10E-04 & 9.38E-03 & \begin{tabular}[c]{@{}c@{}}0.0330\\(0.0279)\end{tabular} & - \\
	FNLW & 4.91E-04 & 6.66E-02 & \begin{tabular}[c]{@{}c@{}}0.0074\\(0.5515 )\end{tabular} & - & 2.98E-04 & 1.24E-02 & \begin{tabular}[c]{@{}c@{}}0.0241\\(0.0419)\end{tabular} & - & 3.06E-04 & 1.32E-02 & \begin{tabular}[c]{@{}c@{}}0.0231\\(0.0419)\end{tabular} & - \\
	POET & NaN & NaN & NaN & - & -7.06E-04 & 2.74E-01 & \begin{tabular}[c]{@{}c@{}}-0.0026\\(0.9137)\end{tabular} & - & 1.07E-03 & 2.71E-01 & \begin{tabular}[c]{@{}c@{}}0.0039\\(0.7912)\end{tabular} & - \\
	Projected POET & 1.20E-03 & 1.71E-01 & \begin{tabular}[c]{@{}c@{}}0.0070\\(0.5515)\end{tabular} & - & -8.06E-05 & 1.61E-02 & \begin{tabular}[c]{@{}c@{}}-0.0050\\(0.9337)\end{tabular} & - & -7.57E-05 & 1.93E-02 & \begin{tabular}[c]{@{}c@{}}-0.0039\\(0.9482)\end{tabular} & - \\
	FGL (FF1) & 7.96E-04 & 2.80E-02 & \begin{tabular}[c]{@{}c@{}}0.0285\\(0.0758)\end{tabular} & - & 3.73E-04 & 8.73E-03 & \begin{tabular}[c]{@{}c@{}}0.0427\\(0.024)\end{tabular} & - & 3.52E-04 & 8.62E-03 & \begin{tabular}[c]{@{}c@{}}0.0408\\(0.0259)\end{tabular} & - \\
	FGL (FF3) & 6.51E-04 & 2.74E-02 & \begin{tabular}[c]{@{}c@{}}0.0238\\(0.0758)\end{tabular} & - & 3.52E-04 & 8.96E-03 & \begin{tabular}[c]{@{}c@{}}0.0393\\(0.024)\end{tabular} & - & 3.39E-04 & 8.94E-03 & \begin{tabular}[c]{@{}c@{}}0.0379\\(0.022)\end{tabular} & - \\
	FGL (FF5) & 5.87E-04 & 2.70E-02 & \begin{tabular}[c]{@{}c@{}}0.0217\\(0.0758)\end{tabular} & - & 3.47E-04 & 9.38E-03 & \begin{tabular}[c]{@{}c@{}}0.0370\\(0.024)\end{tabular} & - & 3.36E-04 & 9.29E-03 & \begin{tabular}[c]{@{}c@{}}0.0362\\(0.022)\end{tabular} & - \\
	FF1 & 7.38E-04 & 1.11E-01 & \begin{tabular}[c]{@{}c@{}}0.0067\\(0.5821)\end{tabular} & - & 3.30E-05 & 1.62E-02 & \begin{tabular}[c]{@{}c@{}}0.0020\\(0.7139)\end{tabular} & - & 2.49E-05 & 1.61E-02 & \begin{tabular}[c]{@{}c@{}}0.0015\\(0.8430)\end{tabular} & - \\
	FF3 & 7.52E-04 & 1.11E-01 & \begin{tabular}[c]{@{}c@{}}0.0068\\(0.5821)\end{tabular} & - & 2.68E-05 & 1.62E-02 & \begin{tabular}[c]{@{}c@{}}0.0017\\(0.7139)\end{tabular} & - & 2.06E-05 & 1.61E-02 & \begin{tabular}[c]{@{}c@{}}0.0013\\(0.8430)\end{tabular} & - \\
	FF5 & 7.59E-04 & 1.11E-01 & \begin{tabular}[c]{@{}c@{}}0.0069\\(0.5821)\end{tabular} & - & 2.01E-05 & 1.62E-02 & \begin{tabular}[c]{@{}c@{}}0.0012\\(0.7139)\end{tabular} & - & 1.38E-05 & 1.61E-02 & \begin{tabular}[c]{@{}c@{}}0.0009\\(0.8430)\end{tabular} & - \\ 
	\midrule
	\textbf{With TC} &  &  &  &  &  &  &  &  &  &  &  &  \\
	EW & 2.01E-04 & 1.90E-02 & 0.0106 & 0.0292 & 2.01E-04 & 1.90E-02 & 0.0106 & 0.0292 & 2.01E-04 & 1.90E-02 & 0.0106 & 0.0292 \\
	FGL & 4.47E-04 & 2.66E-02 & 0.0168 & 0.3655 & 2.30E-04 & 8.22E-03 & 0.0280* & 0.0666 & 2.32E-04 & 7.52E-03 & 0.0309* & 0.0633 \\
	FClime & 1.18E-03 & 8.48E-02 & 0.0139 & 1.0005 & 1.67E-04 & 9.86E-03 & 0.0170 & 0.0369 & 2.46E-04 & 1.07E-02 & 0.0230* & 0.0290 \\
	FLW & -5.54E-05 & 2.65E-02 & -0.0021 & 0.4874 & 1.92E-04 & 9.98E-03 & 0.0193 & 0.1207 & 1.92E-04 & 9.39E-03 & 0.0204* & 0.1194 \\
	FNLW & -2.39E-03 & 7.03E-02 & -0.0340 & 3.6370 & 5.50E-05 & 1.25E-02 & 0.0044 & 0.2441 & 6.08E-05 & 1.33E-02 & 0.0046 & 0.2457 \\
	POET & NaN & NaN & NaN & NaN & -2.28E-02 & 5.55E-01 & -0.0411 & 113.3848 & -2.81E-02 & 4.21E-01 & -0.0666 & 132.8215 \\
	Projected POET & -1.59E-02 & 3.64E-01 & -0.0437 & 35.9692 & -1.03E-03 & 1.68E-02 & -0.0616 & 0.9544 & -1.37E-03 & 2.06E-02 & -0.0666 & 1.2946 \\
	FGL (FF1) & 3.86E-04 & 2.80E-02 & 0.0138 & 0.4068 & 2.82E-04 & 8.74E-03 & 0.0323** & 0.0903 & 2.63E-04 & 8.63E-03 & 0.0305* & 0.0887 \\
	FGL (FF3) & 2.47E-04 & 2.74E-02 & 0.0090 & 0.4043 & 2.60E-04 & 8.98E-03 & 0.0290** & 0.0928 & 2.49E-04 & 8.96E-03 & 0.0278* & 0.0911 \\
	FGL (FF5) & 1.83E-04 & 2.71E-02 & 0.0068 & 0.4032 & 2.53E-04 & 9.40E-03 & 0.0269* & 0.0952 & 2.43E-04 & 9.30E-03 & 0.0262* & 0.0937 \\
	FF1 & -6.69E-03 & 1.28E-01 & -0.0639 & 8.5721 & -5.27E-04 & 1.65E-02 & -0.0319 & 0.5704 & -5.30E-04 & 1.64E-02 & -0.0323 & 0.5641 \\
	FF3 & -6.65E-03 & 1.28E-01 & -0.0635 & 8.5411 & -5.33E-04 & 1.65E-02 & -0.0323 & 0.5701 & -5.34E-04 & 1.64E-02 & -0.0326 & 0.5638 \\
	FF5 & -6.63E-03 & 1.28E-01 & -0.0634 & 8.5262 & -5.40E-04 & 1.65E-02 & -0.0327 & 0.5703 & -5.41E-04 & 1.64E-02 & -0.0330 & 0.5646 \\
	\bottomrule
\end{tabular}%
		}
	\end{table}
\end{landscape}
\clearpage
 \newpage
\begin{sidewaystable}[ph!]
	\centering
	\caption{Cumulative excess return (CER) and risk of MRC portfolios using daily data. Targeted risk is set at $\sigma=0.013$, daily targeted return is $0.0378\%$. P-values are in parentheses. In-sample: January 20, 2000 - January 24, 2002  (504 obs), Out-of-sample: January 17, 2002 - January 31, 2020 (4536 obs).}
	\label{tab5}
	\resizebox{\textwidth}{!}{%
\begin{tabular}{clllcccccccccc} 
	\toprule
	& EW & Index & FGL & FClime & FLW & FNLW & ProjPOET & FGL(FF1) & FGL(FF3) & FGL(FF5) & FF1 & FF3 & FF5 \\ 
	\midrule
	\multicolumn{7}{l}{\textbf{Downturn \#1: Argentine Great Depression (2002)}} & \multicolumn{1}{l}{} & \multicolumn{1}{l}{} & \multicolumn{1}{l}{} & \multicolumn{1}{l}{} & \multicolumn{1}{l}{} & \multicolumn{1}{l}{} & \multicolumn{1}{l}{} \\
	CER & -0.1633 & -0.2418 & 0.2909 & -0.0079 & 0.0308 & 0.0728 & -0.6178 & 0.3375 & 0.3423 & 0.3401 & -0.0860 & -0.0860 & -0.0860 \\
	Risk & 0.0160 & 0.0168 & 0.0206 & 0.0348 & 0.0231 & 0.0213 & 0.0545 & 0.0211 & 0.0211 & 0.0212 & 0.0495 & 0.0495 & 0.0495 \\
	SR & -0.0393 & -0.0615 & \begin{tabular}[c]{@{}l@{}}0.0629\\(0.0619)\end{tabular} & \begin{tabular}[c]{@{}c@{}}0.0164\\(0.0759)\end{tabular} & \begin{tabular}[c]{@{}c@{}}0.0171\\(0.0759)\end{tabular} & \begin{tabular}[c]{@{}c@{}}0.0246\\(0.0759)\end{tabular} & \begin{tabular}[c]{@{}c@{}}-0.0467\\(0.4852)\end{tabular} & \begin{tabular}[c]{@{}c@{}}0.0689\\(0.0619)\end{tabular} & \begin{tabular}[c]{@{}c@{}}0.0696\\(0.0619)\end{tabular} & \begin{tabular}[c]{@{}c@{}}0.0692\\(0.0619)\end{tabular} & \begin{tabular}[c]{@{}c@{}}0.0169\\(0.0759)\end{tabular} & \begin{tabular}[c]{@{}c@{}}0.0169\\(0.0759)\end{tabular} & \begin{tabular}[c]{@{}c@{}}0.0169\\(0.0759)\end{tabular} \\ 
	\midrule
	\multicolumn{7}{l}{\textbf{Downturn \#2: Financial Crisis (2008)}} & \multicolumn{1}{l}{} & \multicolumn{1}{l}{} & \multicolumn{1}{l}{} & \multicolumn{1}{l}{} & \multicolumn{1}{l}{} & \multicolumn{1}{l}{} & \multicolumn{1}{l}{} \\
	CER & -0.5622 & -0.4746 & 0.2938 & -0.8912 & 0.2885 & 0.2075 & -0.9999 & 0.2665 & 0.2650 & 0.2560 & 0.0404 & 0.0404 & 0.0404 \\
	Risk & 0.0310 & 0.0258 & 0.0282 & 0.1484 & 0.0315 & 0.0392 & 0.1963 & 0.0320 & 0.0319 & 0.0319 & 0.0986 & 0.0986 & 0.0986 \\
	SR & -0.0857 & -0.0857 & \begin{tabular}[c]{@{}l@{}}0.0315\\(0.0889)\end{tabular} & \begin{tabular}[c]{@{}c@{}}0.1045\\(0.1079)\end{tabular} & \begin{tabular}[c]{@{}c@{}}0.0282\\(0.1079)\end{tabular} & \begin{tabular}[c]{@{}c@{}}0.0392\\(0.1079)\end{tabular} & \begin{tabular}[c]{@{}c@{}}0.1963\\(0.1079)\end{tabular} & \begin{tabular}[c]{@{}c@{}}0.0320\\(0.0889)\end{tabular} & \begin{tabular}[c]{@{}c@{}}0.0319\\(0.0889)\end{tabular} & \begin{tabular}[c]{@{}c@{}}0.0319\\(0.0889)\end{tabular} & \begin{tabular}[c]{@{}c@{}}0.0006\\(0.1079)\end{tabular} & \begin{tabular}[c]{@{}c@{}}0.0006\\(0.1079)\end{tabular} & \begin{tabular}[c]{@{}c@{}}0.0006\\(0.1079)\end{tabular} \\ 
	\midrule
	\multicolumn{7}{l}{\textbf{Boom \#1 (2017)}} & \multicolumn{1}{l}{} & \multicolumn{1}{l}{} & \multicolumn{1}{l}{} & \multicolumn{1}{l}{} & \multicolumn{1}{l}{} & \multicolumn{1}{l}{} & \multicolumn{1}{l}{} \\
	CER & 0.0627 & 0.1752 & 0.7267 & 0.5331 & 0.3164 & 0.5796 & -0.7599 & 0.6568 & 0.6607 & 0.6486 & -0.5070 & -0.5294 & -0.4755 \\
	Risk & 0.0218 & 0.0042 & 0.0142 & 0.0383 & 0.0118 & 0.0497 & 0.1197 & 0.0135 & 0.0134 & 0.0132 & 0.0720 & 0.0721 & 0.0710 \\
	SR & 0.0220 & 0.1536 & \begin{tabular}[c]{@{}l@{}}0.1606\\(0.5544)\end{tabular} & \begin{tabular}[c]{@{}c@{}}0.1231\\(0.6264)\end{tabular} & \begin{tabular}[c]{@{}c@{}}0.0987\\(0.6264)\end{tabular} & \begin{tabular}[c]{@{}c@{}}0.1008\\(0.5465)\end{tabular} & \begin{tabular}[c]{@{}c@{}}0.0151\\(0.9815)\end{tabular} & \begin{tabular}[c]{@{}c@{}}0.1563\\(0.5455)\end{tabular} & \begin{tabular}[c]{@{}c@{}}0.1581\\(0.5455)\end{tabular} & \begin{tabular}[c]{@{}c@{}}0.1575\\(0.5455)\end{tabular} & \begin{tabular}[c]{@{}c@{}}-0.0022\\(0.9985)\end{tabular} & \begin{tabular}[c]{@{}c@{}}-0.0046\\(0.9985)\end{tabular} & \begin{tabular}[c]{@{}c@{}}0.0002\\(0.9815)\end{tabular} \\ 
	\midrule
	\multicolumn{7}{l}{\textbf{Boom \#2 (2019)}} & \multicolumn{1}{l}{} & \multicolumn{1}{l}{} & \multicolumn{1}{l}{} & \multicolumn{1}{l}{} & \multicolumn{1}{l}{} & \multicolumn{1}{l}{} & \multicolumn{1}{l}{} \\
	CER & 0.1642 & 0.2934 & 0.6872 & 0.2346 & 0.5520 & 0.9315 & 1.8592 & 0.5166 & 0.5168 & 0.5037 & 0.2690 & 0.2682 & 0.2730 \\
	Risk & 0.0185 & 0.0086 & 0.0263 & 0.0557 & 0.0287 & 0.0355 & 0.1177 & 0.0247 & 0.0248 & 0.0248 & 0.1094 & 0.1094 & 0.1094 \\
	SR & 0.0418 & 0.1228 & \begin{tabular}[c]{@{}l@{}}0.0919\\(0.1738)\end{tabular} & \begin{tabular}[c]{@{}c@{}}0.0436\\(0.2298)\end{tabular} & \begin{tabular}[c]{@{}c@{}}0.0753\\(0.2298)\end{tabular} & \begin{tabular}[c]{@{}c@{}}0.0905\\(0.2298)\end{tabular} & \begin{tabular}[c]{@{}c@{}}0.0898\\(0.2298)\end{tabular} & \begin{tabular}[c]{@{}c@{}}0.0793\\(0.1728)\end{tabular} & \begin{tabular}[c]{@{}c@{}}0.0792\\(0.1728)\end{tabular} & \begin{tabular}[c]{@{}c@{}}0.0779\\(0.1728)\end{tabular} & \begin{tabular}[c]{@{}c@{}}0.0798\\(0.2298)\end{tabular} & \begin{tabular}[c]{@{}c@{}}0.0798\\(0.2298)\end{tabular} & \begin{tabular}[c]{@{}c@{}}0.0799\\(0.2298)\end{tabular} \\
	\bottomrule
\end{tabular}%
	}
\end{sidewaystable}
\clearpage
\renewcommand{\appendixpagename}{Supplemental Appendix}
\renewcommand\appendixtocname{Supplemental Appendix}
\begin{appendices} 
	\renewcommand{\thesection}{\Alph{section}}
	\renewcommand{\thesubsection}{\Alph{section}.\arabic{subsection}}
	\renewcommand{\theequation}{\Alph{section}.\arabic{equation}}
	\captionsetup{%
		figurewithin=section,
		tablewithin=section
	}
	\begin{spacing}{2}
This Online Supplemental Appendix is structured as follows: Appendix \ref{appendixAA} summarizes Graphical Lasso algorithm, Appendix \ref{appendixA} contains proofs of the theorems, accompanying lemmas, and an extension of the theorems to elliptical distributions. Appendix \ref{appendixB} provides additional simulations for Section 5, additional empirical results for Section 6 are located in Appendix \ref{appendixC}.

	\section{Graphical Lasso Algorithm} \label{appendixAA}
To solve \eqref{e7.6} we use the procedure based on the weighted Graphical Lasso which was first proposed in \cite{GLASSO} and further studied in \cite{DPGLASSO} and \cite{Sara2018} among others. Define the following partitions of $\bW_{\varepsilon}$, $\widehat{\bSigma}_{\varepsilon}$ and $\bTheta_{\varepsilon}$:
\begin{equation} \label{eq43}
	\bW_{\varepsilon}=\begin{pmatrix}
		\underbrace{\bW_{\varepsilon,11}}_{(p-1)\times(p-1)}&\underbrace{\bw_{\varepsilon,12}}_{(p-1)\times 1}\\\bw_{\varepsilon,12}'&w_{\varepsilon,22}
	\end{pmatrix}, \widehat{\bSigma}_{\varepsilon}=\begin{pmatrix}
		\underbrace{\widehat{\bSigma}_{\varepsilon,11}}_{(p-1)\times(p-1)}&\underbrace{\widehat{\bsigma}_{\varepsilon,12}}_{(p-1)\times 1}\\\widehat{\bsigma}_{\varepsilon,12}'&\widehat{\sigma}_{\varepsilon,22}
	\end{pmatrix}, \bTheta=\begin{pmatrix}
		\underbrace{\bTheta_{\varepsilon,11}}_{(p-1)\times(p-1)}&\underbrace{\btheta_{\varepsilon,12}}_{(p-1)\times 1}\\\btheta_{\varepsilon,12}'&\theta_{\varepsilon,22}
	\end{pmatrix}.
\end{equation}
Let $\bbeta\defeq -\btheta_{\varepsilon,12}/\theta_{\varepsilon,22}$. The idea of GLASSO is to set $\bW_{\varepsilon}= \widehat{\bSigma}_{\varepsilon}+\lambda\bI$ in \eqref{e7.6} and combine the gradient of \eqref{e7.6} with the formula for partitioned inverses to obtain the following $\ell_1$-regularized quadratic program
\begin{equation}\label{eq50}
	\widehat{\bbeta}=\argmin_{\bbeta\in \mathbb{R}^{p-1}}\Bigl\{ \frac{1}{2}\bbeta'\bW_{\varepsilon,11}\bbeta-\bbeta'\widehat{\bsigma}_{\varepsilon,12}+\lambda\norm{\bbeta}_1\Bigr\}.
\end{equation} 
As shown by \cite{GLASSO}, \eqref{eq50} can be viewed as a LASSO regression, where the LASSO estimates are functions of the inner products of $\bW_{\varepsilon,11}$ and $\widehat{\sigma}_{\varepsilon,12}$. Hence, \eqref{e7.6} is equivalent to $p$ coupled LASSO problems. Once we obtain $\widehat{\bbeta}$, we can estimate the entries $\bTheta_{\varepsilon}$ using the formula for partitioned inverses. The procedure to obtain sparse $\bTheta_{\varepsilon}$ is summarized in Algorithm \ref{alg1a}.
\begin{spacing}{1.25}
	\setcounter{algorithm}{0} 
	\begin{algorithm}[H]
		\caption{Graphical Lasso \cite{GLASSO}}
		\label{alg1a}
		\begin{algorithmic}[1]
			\STATE 	Initialize $\bW_{\varepsilon}= \widehat{\bSigma}_{\varepsilon}+\lambda\bI$. The diagonal of $\bW_{\varepsilon}$ remains the same in what follows.
			\STATE Repeat for $j=1,\ldots,p,1,\ldots,p,\ldots$ until convergence:
			\begin{itemize}
				\item Partition $\bW_{\varepsilon}$ into part 1: all but the $j$-th row and column, and part 2: the $j$-th row and column.
				\item  Solve the score equations using the cyclical coordinate descent: $	\bW_{\varepsilon,11}\bbeta-\widehat{\bsigma}_{\varepsilon,12}+\lambda\cdot\text{Sign}(\bbeta)=\mathbf{0}$.
				This gives a $(p-1) \times 1$ vector solution $\widehat{\bbeta}.$
				\item Update $\widehat{\bw}_{\varepsilon,12}=\bW_{\varepsilon,11}\widehat{\bbeta}$.
			\end{itemize}
			\STATE In the final cycle (for $i=1,\ldots,p$) solve for $\frac{1}{\widehat{\theta}_{22}}=w_{\varepsilon,22}-\widehat{\bbeta}'\widehat{\bw}_{\varepsilon,12}$ and $	\widehat{\btheta}_{12}=-\widehat{\theta}_{22}\widehat{\bbeta}$.
		\end{algorithmic}
	\end{algorithm}
\end{spacing}
As was shown in \cite{GLASSO} and the follow-up paper by \cite{DPGLASSO}, the estimator produced by Graphical Lasso is guaranteed to be positive definite. 
\section{Proofs of the Theorems} \label{appendixA}
	\subsection{Lemmas for Theorem 1}
	\begin{lem}\label{lemma4fan}
		Under the assumptions of Theorem 1,
		\begin{enumerate}[label=(\alph*)]
			\item $\max_{i,j\leq K}\abs{(1/T)\sum_{t=1}^{T}f_{it}f_{jt}-\E{f_{it}f_{jt}}}=\mathcal{O}_P(\sqrt{1/T})$,
			\item $\max_{i,j\leq p}\abs{(1/T)\sum_{t=1}^{T}\varepsilon_{it}\varepsilon_{jt}-\E{\varepsilon_{it}\varepsilon_{jt}}}=\mathcal{O}_P(\sqrt{\log p/T})$,
			\item $\max_{i\leq K,j\leq p}\abs{(1/T)\sum_{t=1}^{T}f_{it}\varepsilon_{jt}}=\mathcal{O}_P(\sqrt{\log p/T})$.
		\end{enumerate}
	\end{lem}	
	\begin{proof}
		The proof of Lemma \ref{lemma4fan} can be found in Fan et al. (2011) (Lemma B.1).
	\end{proof}
	\begin{lem}\label{lemma6fan}
		Under Assumption \textbf{(A.4)}, $\max_{t\leq T}\sum_{s=1}^{K}\abs{\E{\bvarepsilon'_{s}\bvarepsilon_{t}}}/p=\mathcal{O}(1)$.
	\end{lem}
	\begin{proof}
		The proof of Lemma \ref{lemma6fan} can be found in Fan et al. (2013) (Lemma A.6).
	\end{proof}
	\begin{lem}\label{lemma7fan}
		For $\widehat{K}$ defined in expression (3.6),
		\begin{equation*}
		\Pr\Big(\widehat{K}=K\Big)\rightarrow 1.
		\end{equation*}
	\end{lem}
	\begin{proof}
		The proof of Lemma \ref{lemma7fan} can be found in Li et al. (2017) (Theorem 1 and Corollary 1).
	\end{proof}
	Using the expressions (A.1) in Bai (2003) and (C.2) in Fan et al. (2013), we have the following identity:
	\begin{equation} \label{eqA1}
	\widehat{\bf}_t-\bH\bf_t = \Big(\frac{\bV}{p}\Big)^{-1}\Bigg[\frac{1}{T}\sum_{s=1}^{T}\widehat{\bf}_{s}\frac{\E{\bvarepsilon'_{s}\bvarepsilon_{t}}}{p}+\frac{1}{T}\sum_{s=1}^{T} \widehat{\bf}_{s}\zeta_{st}+ \frac{1}{T}\sum_{s=1}^{T} \widehat{\bf}_{s}\eta_{st} + \frac{1}{T}\sum_{s=1}^{T} \widehat{\bf}_{s}\xi_{st} \Bigg],
	\end{equation}
	where $\zeta_{st} = \bvarepsilon'_{s}\bvarepsilon_{t}/p-\E{\bvarepsilon'_{s}\bvarepsilon_{t}}/p$, $\eta_{st} = \bf'_{s}\sum_{i=1}^{p}\bb_i\varepsilon_{it}/p$ and $\xi_{st} = \bf'_{t}\sum_{i=1}^{p}\bb_i\varepsilon_{is}/p$.
	\begin{lem}\label{lemma8fan}
		For all $i\leq \widehat{K}$,
		\begin{enumerate}[label=(\alph*)]
			\item $(1/T)\sum_{t=1}^{T}\Big[(1/T)\sum_{t=1}^{T}\hat{f}_{is}\E{\bvarepsilon'_{s}\bvarepsilon_{t}}/p \Big]^2=\mathcal{O}_P(T^{-1})$,
			\item $(1/T)\sum_{t=1}^{T}\Big[(1/T)\sum_{t=1}^{T}\hat{f}_{is}\zeta_{st}/p \Big]^2=\mathcal{O}_P(p^{-1})$,
			\item $(1/T)\sum_{t=1}^{T}\Big[(1/T)\sum_{t=1}^{T}\hat{f}_{is}\eta_{st}/p \Big]^2=\mathcal{O}_P(K^2/p)$,
			\item $(1/T)\sum_{t=1}^{T}\Big[(1/T)\sum_{t=1}^{T}\hat{f}_{is}\xi_{st}/p \Big]^2=\mathcal{O}_P(K^2/p)$.
		\end{enumerate}
	\end{lem}
	\begin{proof}
		We only prove (c) and (d), the proof of (a) and (b) can be found in Fan et al. (2013) (Lemma 8).
		\begin{enumerate}[label=(\alph*),start=3]
			\item Recall, $\eta_{st} = \bf'_{s}\sum_{i=1}^{p}\bb_i\varepsilon_{it}/p$. Using Assumption \textbf{(A.5)}, we get $\E{(1/T)\times \sum_{t=1}^{T}\norm*{\sum_{i=1}^{p}\bb_i\varepsilon_{it}}^2} = \E{\norm{\sum_{i=1}^{p}\bb_{i}\varepsilon_{it}}^2}=\mathcal{O}(pK)$. Therefore, by the Cauchy-Schwarz inequality and the facts that $(1/T)\sum_{t=1}^{T}\norm*{\bf_t}^2=\mathcal{O}(K)$, and, $\forall i$, $\sum_{s=1}^{T}\hat{f}_{is}^{2}=T$,
			\begin{align*}
			\frac{1}{T}\sum_{t=1}^{T}\Big(\frac{1}{T}\sum_{s=1}^{T}\hat{f}_{is}\eta_{st}  \Big)^2&\leq \norm{\frac{1}{T}\sum_{s=1}^{T}\norm*{\hat{f}_{is}\bf'_{s}   }^2\frac{1}{T}\sum_{t=1}^{T}\frac{1}{p}\norm*{\sum_{j=1}^{p}\bb_{i}\varepsilon_{jt}}  }^2\\
			&\leq \frac{1}{Tp^2} \sum_{t=1}^{T}\norm{\sum_{j=1}^{p}\bb_{i}\varepsilon_{jt}}^2\Bigg(\frac{1}{T}\sum_{s=1}^{T}\hat{f}_{is}^2 \frac{1}{T}\sum_{s=1}^{T}\norm{\bf_{s}}^2  \Bigg)\\
			&=\mathcal{O}_P\Big(\frac{K}{p}\cdot K\Big)=\mathcal{O}_P\Big(\frac{K^2}{p}\Big).
			\end{align*}
			\item Using a similar approach as in part (c):
			\begin{align*}
			\frac{1}{T}\sum_{t=1}^{T}\Big(\frac{1}{T}\sum_{s=1}^{T}\hat{f}_{is}\xi_{st}  \Big)^2&= \frac{1}{T}\sum_{t=1}^{T}  \abs{\frac{1}{T}\sum_{s=1}^{T}\bf'_t\sum_{j=1}^{p}\varepsilon_{js} \frac{1}{p}\hat{f}_{is} }^2\leq \Big(\frac{1}{T}\sum_{t=1}^{T}\norm{\bf_t}^2  \Big)\norm{\frac{1}{T}\sum_{s=1}^{T}\sum_{j=1}^{p}\bb_j\varepsilon_{js}\frac{1}{p}\hat{f}_{is}}^2\\
			&\leq \Big(\frac{1}{T}\sum_{t=1}^{T}\norm{\bf_t}^2  \Big) \frac{1}{T}\sum_{s=1}^{T}\norm{\sum_{j=1}^{p}\bb_j\varepsilon_{js}\frac{1}{p}}^2 \Big(\frac{1}{T}\sum_{s=1}^{T}\hat{f}_{is}^2  \Big)\\
			&=\mathcal{O}_P\Big(K\cdot \frac{pK}{p^2}\cdot 1  \Big) = \mathcal{O}_P\Big(\frac{K^2}{p} \Big)
			\end{align*}
		\end{enumerate}
	\end{proof}
	\begin{lem}\label{lemma9fan} 
		~
		\begin{enumerate}[label=(\alph*)]
			\item	$\max_{t\leq T}\norm{(1/(Tp))\sum_{s=1}^{T}\widehat{\bf}^{'}_{s}\E{\bvarepsilon'_{s}\bvarepsilon_{t}}}=\mathcal{O}_P(K/\sqrt{T})$.
			\item	$\max_{t\leq T}\norm{(1/(Tp))\sum_{s=1}^{T}\widehat{\bf}^{'}_{s}\zeta_{st}}=\mathcal{O}_P(\sqrt{K}T^{1/4}/\sqrt{p})$.
			\item	$\max_{t\leq T}\norm{(1/(Tp))\sum_{s=1}^{T}\widehat{\bf}^{'}_{s}\eta_{st}}=\mathcal{O}_P(KT^{1/4}/\sqrt{p})$.
			\item	$\max_{t\leq T}\norm{(1/(Tp))\sum_{s=1}^{T}\widehat{\bf}^{'}_{s}\xi_{st}}=\mathcal{O}_P(KT^{1/4}/\sqrt{p})$.
		\end{enumerate}
	\end{lem}
	\begin{proof}
		Our proof is similar to the proof in Fan et al. (2013). However, we relax the assumptions of fixed $K$.
		\begin{enumerate}[label=(\alph*)]
			\item Using the Cauchy-Schwarz inequality, Lemma \ref{lemma6fan}, and the fact that $(1/T)\sum_{t=1}^{T}\norm*{\widehat{\bf}_t}^2=\mathcal{O}_P(K)$, we get
			\begin{align*}
			&\max_{t\leq T} \norm{\frac{1}{Tp}\sum_{s=1}^{T}\widehat{\bf}'_{s} \E{\bvarepsilon'_{s}\bvarepsilon_{t}} }\leq 	\max_{t\leq T} \Bigg[ \frac{1}{T}\sum_{s=1}^{T}\norm{\widehat{\bf}_{s}}\frac{1}{T}\sum_{s=1}^{T}\Bigg(\frac{\E{\bvarepsilon'_{s}\bvarepsilon_{t}}}{p}  \Bigg)^2 \Bigg]^{1/2}
			\leq \mathcal{O}_P(K) \max_{t\leq T} \Bigg[ \frac{1}{T}\sum_{s=1}^{T}\Bigg(\frac{\E{\bvarepsilon'_{s}\bvarepsilon_{t}}}{p}  \Bigg)^2 \Bigg]^{1/2}\\
			&\leq \mathcal{O}_P(K) \max_{s,t} \sqrt{\abs{\frac{\E{\bvarepsilon'_{s}\bvarepsilon_{t}}}{p}  }} \max_{t\leq T}\Bigg[ \frac{1}{T}\sum_{s=1}^{T}\abs{\frac{\E{\bvarepsilon'_{s}\bvarepsilon_{t}}}{p}}  \Bigg]^{1/2} = \mathcal{O}_P\Big(K\cdot 1\cdot \frac{1}{\sqrt{T}}\Big) = \mathcal{O}_P\Big(\frac{K}{\sqrt{T}}\Big).
			\end{align*}
			\item Using the Cauchy-Schwarz inequality,
			\begin{align*}
			\max_{t\leq T}\norm{\frac{1}{T}\sum_{s=1}^{T}\widehat{\bf}^{'}_{s}\zeta_{st}}\leq \max_{t\leq T} \frac{1}{T}\Bigg(\sum_{s=1}^{T}\norm{\widehat{\bf}_s}^2\sum_{s=1}^{T}\zeta_{st}^{2}  \Bigg)^{1/2}&\leq \Bigg(\mathcal{O}_P(K)\max_t\frac{1}{T} \sum_{s=1}^{T}\zeta_{st}^{2}  \Bigg)^{1/2}\\ &= \mathcal{O}_P\Big(\sqrt{K}\cdot T^{1/4}/\sqrt{p}\cdot \Big).
			\end{align*}
			To obtain the last inequality we used Assumption \textbf{(A.5)}(b) to get $\E{(1/T)\sum_{s=1}^{T}\zeta_{st}^{2}}^2\leq \max_{s,t\leq T}\E{\zeta_{st}^{4}}=\mathcal{O}(1/p^2)$, and then applied the Chebyshev inequality and Bonferroni's method that yield $\max_{t}(1/T)\sum_{s=1}^{T}\zeta_{st}^{2}=\mathcal{O}_P\Big(\sqrt{T}/p\Big)$.
			\item Using the definition of $\eta_{st}$ we get
			\begin{equation*}
			\max_{t\leq T}\norm{\frac{1}{T}\sum_{s=1}^{T}\widehat{\bf}^{'}_{s}\eta_{st}}\leq \norm{\frac{1}{T}\sum_{s=1}^{T}\widehat{\bf}_s\bf'_{s} }\max_t\norm{\frac{1}{p}\sum_{i=1}^{p}\bb_{i}\varepsilon_{it}  }=\mathcal{O}_P\Big(K \cdot T^{1/4}/\sqrt{p}\Big).
			\end{equation*}
			To obtain the last rate we used Assumption \textbf{(A.5)}(c) together with the Chebyshev inequality and Bonferroni's method to get $\max_{t\leq T}\norm{\sum_{i=1}^{p}\bb_{i}\varepsilon_{it}}=\mathcal{O}_P\Big(T^{1/4}\sqrt{p}\Big)$.
			\item In the proof of Lemma \ref{lemma8fan} we showed that $\norm*{(1/T)\times \sum_{t=1}^{T}\sum_{i=1}^{p}\bb_i\varepsilon_{it}(1/p)\widehat{\bf}_s}^2 = \mathcal{O}\Big(\sqrt{K/p}\Big)$. Furthermore, Assumption \textbf{(A.3)} implies $\E{K^{-2}\bf_{t}}^4<M$, therefore, $\max_{t\leq T}\norm{\bf_t}=\mathcal{O}_P\Big(T^{1/4}\sqrt{K}\Big)$. Using these bounds we get
			\begin{align*}
			\max_{t\leq T}\norm{\frac{1}{T}\sum_{s=1}^{T}\widehat{\bf}^{'}_{s}\xi_{st}}\leq \max_{t\leq T}\norm{\bf_t}\cdot \norm{\sum_{s=1}^{T}\sum_{i=1}^{p}\bb_i\varepsilon_{it}\frac{1}{p} \widehat{\bf}_s }=\mathcal{O}_P \Big(T^{1/4}\sqrt{K}\cdot \sqrt{K/p}\Big)=\mathcal{O}_P\Big(T^{1/4}K/\sqrt{p}\Big).
			\end{align*}
		\end{enumerate}
	\end{proof}	
	\begin{lem}\label{lemma10fan} 
		~	
		\begin{enumerate}[label=(\alph*)]
			\item $\max_{i\leq K}(1/T)\sum_{t=1}^{T}(\widehat{\bf}_t-\bH\bf_{t})_{i}^{2}=\mathcal{O}_P(1/T+K^2/p)$.
			\item $(1/T)\sum_{t=1}^{T}\norm*{\widehat{\bf}_t-\bH\bf_{t}}^{2} = \mathcal{O}_P(K/T+K^3/p)$.
			\item $\max_{t\leq T}(1/T)\norm*{\widehat{\bf}_t-\bH\bf_{t}}=\mathcal{O}_P(K/\sqrt{T}+KT^{1/4}/\sqrt{p})$.
		\end{enumerate}
	\end{lem}
	\begin{proof}
		Similarly to Fan et al. (2013), we prove this lemma conditioning on the event $\hat{K}=K$. Since $\Pr(\hat{K}\neq K)=o(1)$, the unconditional arguments are implied.
		\begin{enumerate}[label=(\alph*)]
			\item Using \eqref{eqA1}, for some constant $C>0$,
			\begin{align*}
			\max_{i\leq K}(1/T)\sum_{t=1}^{T}(\widehat{\bf}_t-\bH\bf_{t})_{i}^{2}&\leq C \max_{i\leq K}\frac{1}{T} \sum_{t=1}^{T}\Bigg(\frac{1}{T} \sum_{s=1}^{T}\hat{f}_{is} \frac{\E{\bvarepsilon'_{s}\bvarepsilon_{t}}}{p} \Bigg)^2+C \max_{i\leq K}\frac{1}{T} \sum_{t=1}^{T}\Bigg(\frac{1}{T} \sum_{s=1}^{T}\hat{f}_{is} \zeta_{st} \Bigg)^2\\
			&+ C \max_{i\leq K}\frac{1}{T} \sum_{t=1}^{T}\Bigg(\frac{1}{T} \sum_{s=1}^{T}\hat{f}_{is} \zeta_{st} \Bigg)^2 +C \max_{i\leq K}\frac{1}{T} \sum_{t=1}^{T}\Bigg(\frac{1}{T} \sum_{s=1}^{T}\hat{f}_{is} \xi_{st} \Bigg)^2\\
			&=\mathcal{O}_P\Bigg(\frac{1}{T}+\frac{1}{p}+\frac{K^2}{p}+ \frac{K^2}{p} \Bigg) = \mathcal{O}_P(1/T+K^2/p).
			\end{align*}
			\item Part (b) follows from part (a) and
			\begin{equation*}
			\frac{1}{T}\sum_{t=1}^{T}\norm*{\widehat{\bf}_t-\bH\bf_{t}}^{2}\leq K\max_{i\leq K}\frac{1}{T}\sum_{t=1}^{T}(\widehat{\bf}_t-\bH\bf_{t})_{i}^{2}.
			\end{equation*}
			\item Part (c) is a direct consequence of \ref{eqA1} and Lemma \ref{lemma9fan}.
		\end{enumerate}
	\end{proof}		
	\begin{lem}\label{lemma11fan} 
		~	
		\begin{enumerate}[label=(\alph*)]
			\item $\bH\bH' = \bI_{\hat{K}}+\mathcal{O}_P(K^{5/2}/\sqrt{T}+K^{5/2}/\sqrt{p})$.
			\item $\bH\bH' = \bI_{K}+\mathcal{O}_P(K^{5/2}/\sqrt{T}+K^{5/2}/\sqrt{p})$.
		\end{enumerate}
	\end{lem}
	\begin{proof}
		Similarly to Lemma \ref{lemma10fan}, we first condition on $\hat{K}=K$.
		\begin{enumerate}[label=(\alph*)]
			\item 	The key observation here is that, according to the definition of $\bH$, its rank grows with $K$, that is, $\norm{\bH}=\mathcal{O}_P(K)$. Let $\widehat{\text{cov}}(\bH\bf_{t})=(1/T)\sum_{t=1}^{T}\bH\bf_{t}(\bH\bf_{t})'$. Using the triangular inequality we get
			\begin{equation}\label{eqA2}
			\norm{\bH\bH'-\bI_{\hat{K}}}_F\leq \norm{\bH\bH'- \widehat{\text{cov}}(\bH\bf_{t})}_F+\norm{\widehat{\text{cov}}(\bH\bf_{t})-\bI_{\hat{K}}}_F.
			\end{equation}
			To bound the first term in \eqref{eqA2}, we use Lemma \ref{lemma4fan}: $\norm{\bH\bH'- \widehat{\text{cov}}(\bH\bf_{t})}_F\leq \norm{\bH}^2\norm{\bI_{K}-\widehat{\text{cov}}(\bH\bf_{t})}_F=\mathcal{O}_P(K^{5/2}/\sqrt{T})$.\\
			To bound the second term in \eqref{eqA2}, we use the Cauchy-Schwarz inequality and Lemma \ref{lemma10fan}:
			\begin{align*}
			&\norm{\frac{1}{T} \sum_{t=1}^{T}\bH\bf_{t}(\bH\bf_{t})'-\frac{1}{T} \sum_{t=1}^{T}\widehat{\bf}_t\widehat{\bf}'_t }_F \leq \norm{\frac{1}{T} \sum_{t=1}^{T}(\bH\bf_{t}- \widehat{\bf}_t)(\bH\bf_{t})'  }_F + \norm{\frac{1}{T} \sum_{t} \widehat{\bf}_t(\widehat{\bf}'_t-(\bH\bf_{t})') }_F\\
			&\leq \Bigg(\frac{1}{T} \sum_{t=1}\norm{\bH\bf_{t}- \widehat{\bf}_t}^2 \frac{1}{T} \sum_{t=1} \norm{\bH\bf_{t}}^2   \Bigg)^{1/2} + \Bigg(\frac{1}{T} \sum_{t=1}\norm{\bH\bf_{t}- \widehat{\bf}_t}^2 \frac{1}{T} \sum_{t=1} \norm{\widehat{\bf}_t}^2   \Bigg)^{1/2}\\
			&=\mathcal{O}_P\Bigg(\Big(\frac{K}{T}+\frac{K^3}{p}\cdot K  \Big)^{1/2}+ \Big(\frac{K}{T}+\frac{K^3}{p}\cdot K^2  \Big)^{1/2}  \Bigg)=\mathcal{O}_P\Bigg(\frac{K^{3/2}}{\sqrt{T}}+\frac{K^{5/2}}{\sqrt{p}}  \Bigg).
			\end{align*}
			\item The proof of (b) follows from $\Pr(\hat{K}=K)\rightarrow 1$ and the arguments made in Fan et al. (2013), (Lemma 11) for fixed $K$.
		\end{enumerate}
	\end{proof}	
	\subsection{Proof of Theorem 1}
	The second part of Theorem 1 was proved in Lemma \ref{lemma10fan}. We now proceed to the convergence rate of the first part. Using the following definitions: $\widehat{\bb}_i=(1/T)\sum_{t=1}^{T}r_{it}\widehat{\bf}_t$ and $(1/T)\sum_{t=1}^{T}\widehat{\bf}_t\widehat{\bf}'_t=\bI_{K}$, we obtain
	\begin{equation} \label{eqA3}
	\widehat{\bb}_i-\bH\bb_{i}=\frac{1}{T}\sum_{t=1}^{T}\bH\bf_t\varepsilon_{it}+\frac{1}{T}\sum_{t=1}^{T}r_{it}(\widehat{\bf}_t-\bH\bf_t)+\bH\Big(\frac{1}{T}\sum_{t=1}^{T}\bf_t\bf'_t-\bI_{K}  \Big)\bb_i.
	\end{equation}
	Let us bound each term on the right-hand side of \eqref{eqA3}. The first term is
	\begin{align*}
	\max_{i\leq p} \norm{ \bH\bf_t\varepsilon_{it} }\leq \norm{\bH}\max_i\sqrt{ \sum_{k=1}^{K}\Bigg(\frac{1}{T}\sum_{t=1}^{T}f_{kt}\varepsilon_{it}  \Bigg)^2} &\leq \norm{\bH}\sqrt{K}\max_{i\leq p, j\leq K}\abs{\frac{1}{T}\sum_{t=1}^{T}f_{jt}\varepsilon_{it}} \\
	&=\mathcal{O}_P\Big(K\cdot K^{1/2}\cdot \sqrt{\log p/T}\Big),
	\end{align*}
	where we used Lemmas \ref{lemma4fan} and \ref{lemma11fan} together with Bonferroni's method. For the second term,
	\begin{align*}
	\max_i\norm{\frac{1}{T}\sum_{t=1}^{T}r_{it}\Big(\widehat{\bf}_t-\bH\bf_t\Big)} \leq \max_i\Bigg(\frac{1}{T}\sum_{t=1}^{T}r_{it}^{2} \frac{1}{T}\sum_{t=1}^{T}\norm{\widehat{\bf}_t-\bH\bf_t }^2  \Bigg)^{1/2} = \mathcal{O}_P\Bigg(\frac{1}{T}+\frac{K^2}{p} \Bigg)^{1/2},
	\end{align*}
	where we used Lemma \ref{lemma10fan} and the fact that $\max_i T^{-1}\sum_{t=1}^{T}r_{it}^2=\mathcal{O}_P(1)$ since $\E{r_{it}^2}=\mathcal{O}(1)$.\\
	Finally, the third term is $\mathcal{O}_P(K^2T^{-1/2})$ since $\norm*{(1/T)\sum_{t=1}^{T}\bf_t\bf'_t-\bI_{K}}=\mathcal{O}_P\Big(KT^{-1/2}\Big)$, $\norm{\bH}=\mathcal{O}_P(K)$ and $\max_i\norm{\bb}_i=\mathcal{O}(1)$ by Assumption \textbf{(B.1)}.	
	\subsection{Corollary 1}
	As a consequence of Theorem 1, we get the following corollary:
	\begin{cor}\label{cor1}
		Under the assumptions of Theorem 1,
		\begin{equation*}
		\max_{i\leq p, t\leq T} \norm{\widehat{\bb}_{i}^{'}\widehat{\bf}_t-\bb'_i\bf_t}=\mathcal{O}_P(\log T^{1/r_2}K^{2}\sqrt{\log p/T}+K^2T^{1/4}/\sqrt{p}).
		\end{equation*}	
	\end{cor}
	\begin{proof}
		Using Assumption \textbf{(A.4)} and  Bonferroni's method, we have $\max_{t\leq T}\norm*{\bf_t}=\mathcal{O}_P(\sqrt{K}\log T^{1/r_2})$. By Theorem 1, uniformly in $i$ and $t$:
		\begin{align*}
		\norm{\widehat{\bb}'_{i}\widehat{\bf}_t-\bb'_{i}\bf_t} &\leq \norm{\widehat{\bb}_{i}-\bH\bb_i }\norm{\widehat{\bf}_t-\bH\bf_t }+\norm{\bH\bb_i}\norm{\widehat{\bf}_t-\bH\bf_t }\\
		&+\norm{\widehat{\bb}_{i}-\bH\bb_i }\norm{\bH\bf_t}+\norm{\bb_i}\norm{\bf_t}\norm{\bH'\bH-\bI_{K}}\\
		&=\mathcal{O}_P\Bigg(\Big(K^{3/2}\sqrt{\frac{\log p}{T}}+\frac{K}{\sqrt{p}}  \Big) \cdot \Big( \frac{K}{\sqrt{T}}+\frac{KT^{1/4}}{\sqrt{p}} \Big) \Bigg) + \mathcal{O}_P\Bigg(K\cdot \Big( \frac{K}{\sqrt{T}}+\frac{KT^{1/4}}{\sqrt{p}} \Big)  \Bigg)\\
		&+ \mathcal{O}_P\Bigg( \Big(K^{3/2}\sqrt{\frac{\log p}{T}}+\frac{K}{\sqrt{p}}  \Big) \cdot \log T^{1/r_2}K^{1/2}  \Bigg)+ \mathcal{O}_P\Bigg(\log T^{1/r_2}K^{1/2} \Big(\frac{K^{5/2}}{\sqrt{T}}+\frac{K^{5/2}}{\sqrt{p}}  \Big) \Bigg)\\
		&= \mathcal{O}_P\Big(\log T^{1/r_2}K^{2}\sqrt{\log p/T}+K^2T^{1/4}/\sqrt{p}\Big).
		\end{align*}
	\end{proof}		
	\subsection{Proof of Theorem 2}
	Using the definition of the idiosyncratic components we have $\varepsilon_{it}-\hat{\varepsilon}_{it}=\bb'_i\bH'(\widehat{\bf}_t-\bH\bf_t )+(\widehat{\bb}'_{i}-\bb'_i\bH')\widehat{\bf}_t+\bb'_i(\bH'\bH-\bI_{K})\bf_t$. We bound the maximum element-wise difference as follows:
	\begin{align*}
	\max_{i\leq p}\frac{1}{T}\sum_{t=1}^{T}(\varepsilon_{it}-\hat{\varepsilon}_{it})^2&\leq 4\max_{i}\norm{\bb'_i\bH'}^{2}\frac{1}{T}\sum_{t=1}^{T}\norm{\widehat{\bf}_t-\bH\bf_t}^2 + 4\max_{i}\norm{\widehat{\bb}'_{i}-\bb'_i\bH' }^2 \frac{1}{T}\sum_{t=1}^{T}\norm{\widehat{\bf}_t}^2\\
	&+4\max_{i}\norm{\bb'_i}\frac{1}{T}\sum_{t=1}^{T}\norm{\bf_t}^2\norm{\bH'\bH-\bI_{K}}_{F}^{2}\\
	&=\mathcal{O}\Bigg(K^2\cdot \Big(\frac{K}{T}+\frac{K^3}{p} \Big) \Bigg)+ \mathcal{O}\Bigg(\Big(\frac{K^3\log p}{T} + \frac{K^2}{p}\Big)\cdot K \Bigg) + \mathcal{O}\Bigg(K\cdot \Big(\frac{K^5}{T}+\frac{K^5}{p} \Big) \Bigg)\\
	&=\mathcal{O}\Bigg(\frac{K^4\log p}{T} + \frac{K^6}{p}\Bigg).
	\end{align*}
	Let $\omega_{3T}\defeq K^{2}\sqrt{\log p/T} +K^3/\sqrt{p}$. Denote $\max_{i\leq p}(1/T)\sum_{t=1}^{T}(\varepsilon_{it}-\hat{\varepsilon}_{it})^2=\mathcal{O}_P(\omega_{3T}^{2})$. Then, $\max_{i,t}\abs{\varepsilon_{it}-\hat{\varepsilon}_{it}}=\mathcal{O}_P(\omega_{3T})=o_P(1)$, where the last equality is implied by Corollary \ref{cor1}.\\
	As pointed out in the main text, the second part of Theorem 2 is based on the relationship between the convergence rates of the estimated covariance and precision matrices established in Jankov\'{a} and van de Geer (2018) (Theorem 14.1.3).
	\subsection{Lemmas for Theorem 3}
	\begin{lem}\label{lemmaC4koike} 
		Under the assumptions of Theorem 1, we have the following results:
		\begin{enumerate}[label=(\alph*)]
			\item $\norm{\bB}=\norm{\bB\bH'}=\mathcal{O}(\sqrt{p})$.
			\item $\varrho_{T}^{-1}\max_{1\leq i\leq p}\norm{\widehat{\bb}_i-\bH'\bb_{i}  }=o_P(1/\sqrt{K})$ and $\max_{1\leq i\leq p}\norm{\widehat{\bb}_i}=\mathcal{O}_P(\sqrt{K})$.
			\item $\varrho_{T}^{-1}\norm{\widehat{\bB}-\bB\bH'} = o_P\Big(\sqrt{p/K}\Big)$ and $\norm{\widehat{\bB}}=\mathcal{O}_P(\sqrt{p})$.
		\end{enumerate}
	\end{lem}
	\begin{proof}
		Part (c) is direct consequences of (a)-(b), therefore, we only prove the first two parts in what follows.
		\begin{enumerate}[label=(\alph*)]
			\item Part (a) easily follows from \textbf{(B.1)}: $\text{tr}(\bSigma-\bB\bB')=\text{tr}(\bSigma)-\norm{\bB}^2\geq 0$, since $\text{tr}(\bSigma)=\mathcal{O}(p)$ by \textbf{(B.1)}, we get $\norm{\bB}^2=\mathcal{O}(p)$. Part (a) follows from the fact that the linear space spanned by the rows of $\bB$ is the same as that by the rows of $\bB\bH'$, hence, in practice, it does not matter which one is used.
			\item From Theorem 1, we have $\max_{i\leq p} \norm{\widehat{\bb}_i-\bH\bb_i}=\mathcal{O}_P(\omega_{1T})$. Using the definition of $\varrho_{T}$ from Theorem 2, it follows that $\varrho_{T}^{-1}\omega_{1T}=o_P(\omega_{1T}\omega_{3T}^{-1})$. Let $\widetilde{z}_T\defeq \omega_{1T}\omega_{3T}^{-1}$. Consider\\ $\varrho_{T}^{-1}\max_{1\leq i\leq p}\norm{\widehat{\bb}_i-\bH\bb_{i}  }=o_P(z_T)$. The latter holds for any $z_t\geq \widetilde{z}_T$, with the tightest bound obtained when $z_T=\widetilde{z}_T$. For the ease of representation, we use $z_T=1/\sqrt{K}$ instead of $\widetilde{z}_T$.\\
			The second result in Part (b) is obtained using the fact that $\max_{1\leq i\leq p}\norm{\widehat{\bb}_i}\leq \sqrt{K}\norm{\bB}_{\text{max}}$, where $\norm{\bB}_{\text{max}}=\mathcal{O}(1)$ by \textbf{(B.1)}.
		\end{enumerate}
	\end{proof}
	\begin{lem}\label{lemmaC5koike} 
		Let $\bPi\defeq \Big[ \bTheta_f+(\bB\bH')'\bTheta_{\varepsilon}(\bB\bH') \Big]^{-1}$, $\widehat{\bPi}\defeq \Big[\widehat{\bTheta}_f+\widehat{\bB}'\widehat{\bTheta}_{\varepsilon}\widehat{\bB}\Big]^{-1}$. Also, define $\bSigma_f = (1/T)\sum_{t=1}^{T}\bH\bf_{t}(\bH\bf_{t})'$, $\bTheta_f = \bSigma_{f}^{-1}$, $\widehat{\bSigma}_f\defeq (1/T)\sum_{t=1}^{T}\widehat{\bf}_t\widehat{\bf}'_t$, and $\widehat{\bTheta}_f=\widehat{\bSigma}_{f}^{-1}$.  Under the assumptions of Theorem 2, we have the following results:
		\begin{enumerate}[label=(\alph*)]
			\item $\Lambda_{\text{min}}(\bB'\bB)^{-1}=\mathcal{O}(1/p)$.
			\item $\vertiii{\bPi}_2=\mathcal{O}(1/p)$.
			\item  $\varrho_{T}^{-1}\vertiii{\widehat{\bTheta}_f-\bTheta_f}_2=o_P\Big(1/\sqrt{K}\Big)$.
			\item $\varrho_{T}^{-1}\vertiii{\widehat{\bPi}-\bPi}_2=\mathcal{O}_P\Big(s_T/p\Big)$ and $\vertiii{\widehat{\bPi}}_2=\mathcal{O}_P(1/p)$.
		\end{enumerate}
	\end{lem}
	\begin{proof}
		~
		\begin{enumerate}[label=(\alph*)]
			\item Using Assumption \textbf{(A.2)} we have $\abs{\Lambda_{\text{min}}(p^{-1}\bB'\bB)-\Lambda_{\text{min}}(\breve{\bB})}\leq \vertiii{p^{-1}\bB'\bB-\breve{\bB}}_2$, which implies Part (a).
			\item First, notice that $\vertiii{\bPi}_2=\Lambda_{\text{min}}(\bTheta_f+(\bB\bH')'\bTheta_{\varepsilon}(\bB\bH'))^{-1}$. Therefore, we get
			\begin{equation*}
			\vertiii{\bPi}_2 \leq \Lambda_{\text{min}}((\bB\bH')'\bTheta_{\varepsilon}(\bB\bH'))^{-1}\leq \Lambda_{\text{min}}(\bB'\bB)^{-1}\Lambda_{\text{min}}(\bTheta_{\varepsilon})^{-1} = \Lambda_{\text{min}}(\bB'\bB)^{-1}\Lambda_{\text{max}}(\bSigma_{\varepsilon}),
			\end{equation*}
			where the second inequality is due to the fact that the linear space spanned by the rows of $\bB$ is the same as that by the rows of $\bB\bH'$, hence, in practice, it does not matter which one is used. Therefore, the result in Part (b) follows from Part (a), Assumptions \textbf{(A.1)} and \textbf{(A.2)}.
			\item From Lemma \ref{lemma11fan} we obtained:
			\begin{align*}
			&\norm{\frac{1}{T} \sum_{t=1}^{T}\bH\bf_{t}(\bH\bf_{t})'-\frac{1}{T} \sum_{t=1}^{T}\widehat{\bf}_t\widehat{\bf}'_t }_F =\mathcal{O}_P\Bigg(\frac{K^{3/2}}{\sqrt{T}}+\frac{K^{5/2}}{\sqrt{p}}  \Bigg).
			\end{align*}
			Since $\vertiii{\bTheta_f(\widehat{\bSigma}_f-\bSigma_f)}_2<1$, we have
			\begin{align*}
			\vertiii{\widehat{\bTheta}_f-\bTheta_f}_2 \leq \frac{\vertiii{\bTheta_f}_2\vertiii{ \bTheta_f(\widehat{\bSigma}_f-\bSigma_f) }_2}{1- \vertiii{ \bTheta_f(\widehat{\bSigma}_f-\bSigma_f) }_2 }=\mathcal{O}_P\Bigg(\frac{K^{3/2}}{\sqrt{T}}+\frac{K^{5/2}}{\sqrt{p}}  \Bigg).
			\end{align*}		
			Let $\omega_{4T}=K^{3/2}/\sqrt{T}+K^{5/2}/\sqrt{p}$. Using the definition of $\varrho_{T}$ from Theorem 2, it follows that $\varrho_{T}^{-1}\omega_{4T}=o_P(\omega_{4T}\omega_{3T}^{-1})$. Let $\widetilde{\gamma}_T\defeq \omega_{4T}\omega_{3T}^{-1}$. Consider $\varrho_{T}^{-1}\vertiii{\widehat{\bTheta}_f-\bTheta_f}_2=o_P(\gamma_T)$. The latter holds for any $\gamma_t\geq \widetilde{\gamma}_T$, with the tightest bound obtained when $\gamma_T=\widetilde{\gamma}_T$. For the ease of representation, we use $\gamma_T=1/\sqrt{K}$ instead of $\widetilde{\gamma}_T$.
			\item We will bound each term in the definition of $\widehat{\bPi}-\bPi$. First, we have
			\begin{align} \label{eqA4}
			\vertiii{\widehat{\bB}'\widehat{\bTheta}_{\varepsilon}\widehat{\bB} - (\bB\bH')'\bTheta_{\varepsilon}(\bB\bH') }_2 &\leq \vertiii{\widehat{\bB}-\bB\bH' }_2\vertiii{\widehat{\bTheta}_{\varepsilon} }_2\vertiii{\widehat{\bB} }_2+\vertiii{\bB\bH'}_2\vertiii{ \widehat{\bTheta}_{\varepsilon}-\bTheta_{\varepsilon}}_2\vertiii{\widehat{\bB} }_2 \nonumber\\
			&+\vertiii{\bB\bH'}_2\vertiii{\bTheta_{\varepsilon}}_2\vertiii{\widehat{\bB}-\bB\bH' }_2 = \mathcal{O}_P\Bigg(p\cdot s_T\cdot \varrho_{T}  \Bigg). 
			\end{align}
			Now we combine \eqref{eqA4} with the results from Parts (b)-(c):
			\begin{align*}
			\varrho_{T}^{-1}\vertiii{\bPi\Big(\widehat{\bPi}^{-1} - \bPi^{-1} \Big)}_2 = \mathcal{O}_P\Big(s_t \Big).
			\end{align*}
			Finally, since $\vertiii{\bPi\Big(\widehat{\bPi}^{-1} - \bPi^{-1} \Big)}_2<1$, we have
			\begin{align*}
			\varrho_{T}^{-1}\vertiii{\widehat{\bPi}-\bPi}_2\leq 	\varrho_{T}^{-1}\frac{\vertiii{\bPi}_2 \vertiii{\bPi\Big(\widehat{\bPi}^{-1} - \bPi^{-1} \Big)}_2}{1-\vertiii{\bPi\Big(\widehat{\bPi}^{-1} - \bPi^{-1} \Big)}_2} =\mathcal{O}_P \Bigg(\frac{s_t}{p} \Bigg).
			\end{align*}
		\end{enumerate}
	\end{proof}	
	\subsection{Proof of Theorem 3}
	Using the Sherman-Morrison-Woodbury formula, we have
	\begin{align} \label{eqA5}
	\vertiii{\widehat{\bTheta}-\bTheta}_l&\leq \vertiii{\widehat{\bTheta}_{\varepsilon}-\bTheta_{\varepsilon}}_l+\vertiii{(\widehat{\bTheta}_{\varepsilon}-\bTheta_{\varepsilon} )\widehat{\bB}\widehat{\bPi}\widehat{\bB}'\widehat{\bTheta}_{\varepsilon}}_l + \vertiii{\bTheta_{\varepsilon}(\widehat{\bB}-\bB\bH')\widehat{\bPi}\widehat{\bB}' \widehat{\bTheta}_{\varepsilon}}_l\nonumber\\
	&+\vertiii{\bTheta_{\varepsilon}\bB\bH'(\widehat{\bPi}-\bPi)\widehat{\bB}'\widehat{\bTheta}_{\varepsilon} }_l + \vertiii{\bTheta_{\varepsilon}\bB\bH'\bPi(\widehat{\bB}-\bB )'\widehat{\bTheta}_{\varepsilon} }_l + \vertiii{\bTheta_{\varepsilon}\bB\bH'\bPi(\bB\bH')'(\widehat{\bTheta}_{\varepsilon}-\bTheta_{\varepsilon}) }_l \nonumber\\
	&=\Delta_1+\Delta_2+\Delta_3+\Delta_4+\Delta_5+\Delta_6.
	\end{align}
	We now bound the terms in \eqref{eqA5} for $l=2$ and $l=\infty$.
	We start with $l=2$. First, note that $\varrho_{T}^{-1}\Delta_1=\mathcal{O}_P(s_T)$ by Theorem 2. Second, using Lemmas \ref{lemmaC4koike}-\ref{lemmaC5koike} together with Theorem 2, we have $\varrho_{T}^{-1}(\Delta_2+\Delta_6) = \mathcal{O}_P(s_T\cdot \sqrt{p}\cdot (1/p)\cdot \sqrt{p}\cdot 1)=\mathcal{O}_P(s_T)$. Third, $\varrho_{T}^{-1}(\Delta_3+\Delta_5)$ is negligible according to Lemma \ref{lemmaC4koike}\red{(c)}. Finally, $\varrho_{T}^{-1}\Delta_4=\mathcal{O}_P\Big(1\cdot \sqrt{p}\cdot(s_T/p)\cdot \sqrt{p}\cdot 1\Big) = \mathcal{O}_P(s_T)$ by Lemmas \ref{lemmaC4koike}-\ref{lemmaC5koike} and Theorem 2.\\
	Now consider $l=\infty$. First, similarly to the previous case, $\varrho_{T}^{-1}\Delta_1=\mathcal{O}_P(s_T)$. Second, $\varrho_{T}^{-1}(\Delta_2+\Delta_6) = \mathcal{O}_P\Big(s_T\cdot \sqrt{pK}\cdot (\sqrt{K}/p) \cdot \sqrt{pK}\cdot \sqrt{d_T}\Big) = \mathcal{O}_P(s_TK^{3/2}\sqrt{d_T})$, where we used the fact that for any $\bA \in \mathcal{S}_p$ we have $\vertiii{\bA}_1=\vertiii{\bA}_{\infty}\leq \sqrt{d(\bA)}\vertiii{\bA}_2$, where $d(\bA)$ measures the maximum vertex degree as described at the beginning of Section 4. Third, the term  $\varrho_{T}^{-1}(\Delta_3+\Delta_5)$ is negligible according to Lemma \ref{lemmaC4koike}\red{(c)}. Finally, $\varrho_{T}^{-1}\Delta_4=\mathcal{O}_P (\sqrt{d_T}\cdot \sqrt{pK}\cdot \sqrt{K}(s_T)/p\cdot \sqrt{pK} \cdot \sqrt{d_T}  ) = \mathcal{O}_P(d_TK^{3/2}s_T  )$.	
	\subsection{Lemmas for Theorem 4}
	\begin{lem}\label{theorA1caner}
		Under the assumptions of Theorem 4, $\vertiii{\bTheta}_1=\mathcal{O}(d_TK^{3/2})$, where $d_T$ was defined in Section 4.	
	\end{lem}
	\begin{proof}
		~
		 We use the Sherman-Morrison-Woodbury formula:
			\begin{align}\label{A6new}
			\vertiii{\bTheta}_1 &\leq \vertiii{\bTheta_{\varepsilon}}_1+\vertiii{\bTheta_{\varepsilon}\bB\lbrack\bTheta_{f}+\bB'\bTheta_{\varepsilon}\bB \rbrack^{-1}\bB'\bTheta_{\varepsilon}}_1 \nonumber\\
			&=\mathcal{O}(\sqrt{d_T})+\mathcal{O}\Big(\sqrt{d_T}\cdot p \cdot \frac{\sqrt{K}}{p} \cdot K \cdot \sqrt{d_T}\Big) = \mathcal{O}(d_TK^{3/2}).
			\end{align}
			The last equality in \eqref{A6new} is obtained under the assumptions of Theorem 4. This result is important in several aspects: it shows that the sparsity of the precision matrix of stock returns is controlled by the sparsity in the precision of the idiosyncratic returns. Hence, one does not need to impose an unrealistic sparsity assumption on the precision of returns a priori when the latter follow a factor structure - sparsity of the precision once the common movements have been taken into account would suffice.
	\end{proof}	
	\begin{lem}\label{lemA1caner}
		Define $a = \biota'_{p}\bTheta\biota_p/p$, $b = \biota'_{p}\bTheta\boldm/p$, $d = \boldm'\bTheta\boldm/p$, $g = \sqrt{\boldm'\bTheta\boldm}/p$ and $\widehat{a} = \biota'_{p}\widehat{\bTheta}\biota_p/p$,  $\widehat{b} = \biota'_{p}\widehat{\bTheta}\widehat{\boldm}/p$, $\widehat{d}=\widehat{\boldm}'\widehat{\bTheta}\widehat{\boldm}/p$, $\widehat{g}=\sqrt{\widehat{\boldm}'\widehat{\bTheta}\widehat{\boldm}}/p$ . Under the assumptions of Theorem 4 and assuming $(ad-b^2)>0$,
		\begin{enumerate}[label=(\alph*)]
			\item $a \geq C_0>0$, $b =\mathcal{O}(1)$, $d=\mathcal{O}(1)$, where $C_0$ is a positive constant representing the minimal eigenvalue of $\bTheta$.
			\item $\abs{\widehat{a}-a}=\mathcal{O}_P(\varrho_{T}d_TK^{3/2}s_T )=o_P(1)$.
			\item $\abs{\widehat{b}-b}=\mathcal{O}_P(\varrho_{T}d_TK^{3/2}s_T )=o_P(1)$
			\item  $\abs{\widehat{d}-d}=\mathcal{O}_P(\varrho_{T}d_TK^{3/2}s_T )=o_P(1)$.
			\item $\abs{\widehat{g}-g}=\mathcal{O}_P\Big(\lbrack\varrho_{T}d_TK^{3/2}s_T\rbrack^{1/2}\Big)=o_P(1)$.
			\item $\abs{(\widehat{a}\widehat{d}-\widehat{b}^2)-(ad-b^2)  }=\mathcal{O}_P\Big(\varrho_{T}d_TK^{3/2}s_T \Big)=o_P(1)$.
			\item $\abs{ad-b^2}=\mathcal{O}(1)$.
		\end{enumerate}
	\end{lem}
	\begin{proof}
		~
		\begin{enumerate}[label=(\alph*)]
			\item Part (a) is trivial and follows directly from $\vertiii{\bTheta}_2=\mathcal{O}(1)$ and $\norm{\boldm}_{\infty}=\mathcal{O}(1)$ from Assumption
			\ref{B1}. We show the proof for $d$: recall, $d = \boldm'\bTheta\boldm/p \leq \vertiii{\bTheta}_{2}^{2}\norm{\boldm}_{2}^{2}/p=\mathcal{O}(1)$.
			\item Using the H\"{o}lders inequality, we have
			\begin{align*} 
			\abs{\widehat{a}-a}=\abs{\frac{\biota'_{p}(\widehat{\bTheta}-\bTheta)\biota_p}{p}}\leq \frac{\norm{(\widehat{\bTheta}-\bTheta)\biota_p }_1\norm{\biota_p}_{\text{max}}}{p} &\leq \vertiii{ \widehat{\bTheta}-\bTheta}_1 \\
			&=\mathcal{O}_P\Big(\varrho_{T}d_TK^{3/2}(s_T+(1/p)) \Big)=o_P(1),
			\end{align*}
			where the last rate is obtained using the assumptions of Theorem 3.
			\item First, rewrite the expression of interest:
			\begin{align} \label{A6_0}
			\widehat{b}-b = \lbrack \biota'_{p}(\widehat{\bTheta}-\bTheta) (\widehat{\boldm}-\boldm)  \rbrack/p + \lbrack \biota'_{p}(\widehat{\bTheta}-\bTheta) \boldm  \rbrack/p +\lbrack \biota'_{p}\bTheta (\widehat{\boldm}-\boldm)  \rbrack/p.
			\end{align}
			We now bound each of the terms in \eqref{A6_0} using the expressions derived in Callot et al. (2019) (see their Proof of Lemma A.2) and the fact that $\log p/T=o(1)$.
			\begin{align}
			\abs{ \biota'_{p}(\widehat{\bTheta}-\bTheta) (\widehat{\boldm}-\boldm) }/p \leq \vertiii{\widehat{\bTheta}-\bTheta }_1 \norm{\widehat{\boldm}-\boldm }_{\text{max}} = \mathcal{O}_P\Big( \varrho_{T}d_TK^{3/2}s_T \cdot \sqrt{\frac{\log p}{T}} \Big).
			\end{align}
			\begin{align}
			\abs{ \biota'_{p}(\widehat{\bTheta}-\bTheta) \boldm  }/p \leq \vertiii{\widehat{\bTheta}-\bTheta }_1 = \mathcal{O}_P\Big( \varrho_{T}d_TK^{3/2}s_T \Big).
			\end{align}
			\begin{align}
			\abs{  \biota'_{p}\bTheta (\widehat{\boldm}-\boldm)  }/p \leq \vertiii{\bTheta}_1 \norm{\widehat{\boldm}-\boldm }_{\text{max}} = \mathcal{O}_P\Big( d_TK^{3/2} \cdot \sqrt{\frac{\log p}{T}} \Big).
			\end{align}
			\item First, rewrite the expression of interest:
			\begin{align} \label{A6}
			\widehat{d}-d &= \lbrack (\widehat{\boldm}-\boldm)'(\widehat{\bTheta}-\bTheta) (\widehat{\boldm}-\boldm)  \rbrack/p + \lbrack (\widehat{\boldm}-\boldm)'\bTheta (\widehat{\boldm}-\boldm)  \rbrack/p \nonumber\\
			&+\lbrack 2(\widehat{\boldm}-\boldm)'\bTheta\boldm  \rbrack/p + \lbrack 2 \boldm'(\widehat{\bTheta}-\bTheta) (\widehat{\boldm}-\boldm)  \rbrack/p \nonumber\\
			&+ \lbrack\boldm'(\widehat{\bTheta}-\bTheta) \boldm  \rbrack/p.
			\end{align}
			We now bound each of the terms in \eqref{A6} using the expressions derived in Callot et al. (2019) (see their Proof of Lemma A.3) and the facts that $\log p/T=o(1)$ and $\norm{\widehat{\boldm}-\boldm}_{\text{max}}=\mathcal{O}_P(\sqrt{\log p/T})$.
			\begin{align}\label{A7}
			\abs{ (\widehat{\boldm}-\boldm)'(\widehat{\bTheta}-\bTheta) (\widehat{\boldm}-\boldm)  }/p &\leq \norm{\widehat{\boldm}-\boldm }_{\text{max}}^2\vertiii{\widehat{\bTheta}-\bTheta }_1 \nonumber\\ &= \mathcal{O}_P\Big(\frac{\log p}{T} \cdot \varrho_{T}d_TK^{3/2}s_T \Big)
			\end{align}
			\begin{align}
			\abs{ (\widehat{\boldm}-\boldm)'\bTheta (\widehat{\boldm}-\boldm)  }/p \leq \norm{\widehat{\boldm}-\boldm }_{\text{max}}^2\vertiii{\bTheta}_1 = \mathcal{O}_P\Big(\frac{\log p}{T} \cdot d_TK^{3/2} \Big).
			\end{align}
			\begin{align}
			\abs{ (\widehat{\boldm}-\boldm)'\bTheta\boldm }/p \leq \norm{\widehat{\boldm}-\boldm }_{\text{max}}\vertiii{\bTheta}_1 = \mathcal{O}_P\Big(\sqrt{\frac{\log p}{T}} \cdot d_TK^{3/2} \Big).
			\end{align}
			\begin{align}
			\abs{\boldm'(\widehat{\bTheta}-\bTheta) (\widehat{\boldm}-\boldm)}/p &\leq \norm{\widehat{\boldm}-\boldm }_{\text{max}}\vertiii{\widehat{\bTheta}-\bTheta }_1 \nonumber\\ &= \mathcal{O}_P\Big(\sqrt{\frac{\log p}{T}} \cdot \varrho_{T}d_TK^{3/2}s_T \Big).
			\end{align}
			\begin{align}
			\abs{\boldm'(\widehat{\bTheta}-\bTheta) \boldm  }/p \leq \vertiii{\widehat{\bTheta}-\bTheta }_1 = \mathcal{O}_P\Big(\varrho_{T}d_TK^{3/2}s_T \Big).
			\end{align}
			\item This is a direct consequence of Part (d) and the fact that $\sqrt{\widehat{d}-d}\geq \sqrt{\widehat{d}}-\sqrt{d}$.
			\item First, rewrite the expression of interest:
			\begin{align*}
			(\widehat{a}\widehat{d}-\widehat{b}^2)-(ad-b^2)  = \lbrack (\widehat{a}-a)+a \rbrack\lbrack  (\widehat{d}-d)+d \rbrack - \lbrack  (\widehat{b}-b)+b \rbrack^2, 
			\end{align*}
			therefore, using Lemma \ref{lemA1caner}, we have
			\begin{align*}
			\abs{(\widehat{a}\widehat{d}-\widehat{b}^2)-(ad-b^2)  } &\leq \Big[\abs{\widehat{a}-a} \abs{\widehat{d}-d} + \abs{\widehat{a}-a}d+a\abs{\widehat{d}-d}+ (\widehat{b}-b)^2 +2\abs{b} \abs{\widehat{b}-b}\Big]\\
			&=\mathcal{O}_P\Big(\varrho_{T}d_TK^{3/2}s_T \Big)=o_P(1).
			\end{align*}
			\item This is a direct consequence of Part (a): $ad-b^2\leq ad = \mathcal{O}(1)$.
		\end{enumerate}		
	\end{proof}	
	\subsection{Proof of Theorem 4}
	Let us derive convergence rates for each portfolio weight formulas one by one. We start with GMV formulation. 
	\begin{align*}
	\norm{\widehat{\bw}_{\text{GMV}}-\bw_{\text{GMV}}}_1\leq \frac{a\frac{\norm{(\widehat{\bTheta}-\bTheta)\biota_p}_1 }{p}+\abs{a-\widehat{a} }\frac{\norm{\bTheta\biota_p}_1}{p}}{\abs{\widehat{a}}a}=\mathcal{O}_P\Big(\varrho_{T}d_{T}^2K^{3}s_T \Big)=o_P(1),
	\end{align*}
	where the first inequality was shown in Callot et al. (2019) (see their expression A.50), and the rate follows from Lemmas \ref{lemA1caner} and \ref{theorA1caner}.\\
	We now proceed with the MWC weight formulation. First, let us simplify the weight expression as follows: $\bw_{\text{MWC}} = \kappa_1(\bTheta\biota_p/p)+\kappa_2(\bTheta\boldm/p)$, where
	\begin{align*}
	&\kappa_1 = \frac{d-\mu b}{ad-b^2}\\
	&\kappa_2 = \frac{\mu a-b}{ad-b^2}.
	\end{align*}
	Let $\widehat{\bw}_{\text{MWC}} = \widehat{\kappa}_1(\widehat{\bTheta}\biota_p/p) + \widehat{\kappa}_2(\widehat{\bTheta}\widehat{\boldm}/p)$, where $\widehat{\kappa}_1$  and $\widehat{\kappa}_2$ are the estimators of $\kappa_1$ and $\kappa_2$ respectively. As shown in Callot et al. (2019) (see their equation A.57), we can bound the quantity of interest as follows:
	\begin{align}\label{A16}
	\norm{\widehat{\bw}_{\text{MWC}}-\bw_{\text{MWC}}}_1 &\leq \abs{(\widehat{\kappa}_1-\kappa_1) } \norm{(\widehat{\bTheta}-\bTheta)\biota_p}_1/p + \abs{(\widehat{\kappa}_1-\kappa_1) }\norm{\bTheta\biota_p}_1/p + \abs{\kappa_1 } \norm{(\widehat{\bTheta}-\bTheta)\biota_p}_1/p \nonumber\\
	&+ \abs{(\widehat{\kappa}_2-\kappa_2) } \norm{(\widehat{\bTheta}-\bTheta)(\widehat{\boldm}-\boldm)}_1/p + \abs{(\widehat{\kappa}_2-\kappa_2) } \norm{\bTheta(\widehat{\boldm}-\boldm)}_1/p \nonumber\\
	&+ \abs{(\widehat{\kappa}_2-\kappa_2) } \norm{(\widehat{\bTheta}-\bTheta)\boldm}_1/p + \abs{(\widehat{\kappa}_2-\kappa_2) } \norm{\bTheta\boldm}_1/p \nonumber\\
	&+\abs{\kappa_2 } \norm{(\widehat{\bTheta}-\bTheta)(\widehat{\boldm}-\boldm)}_1/p + \abs{\kappa_2 } \norm{(\widehat{\bTheta}-\bTheta)\boldm}_1/p.
	\end{align}
	For the ease of representation, denote $y=ad-b^2$. Then, using similar technique as in Callot et al. (2019) we get
	\begin{align*}
	\abs{(\widehat{\kappa}_1-\kappa_1) } \leq \frac{y\abs{\widehat{d}-d}+y\mu \abs{\widehat{b}-b}+\abs{\widehat{y}-y}\abs{d-\mu b}}{\widehat{y}y}=\mathcal{O}_P\Big(\varrho_{T}d_TK^{3/2}s_T \Big)=o_P(1),
	\end{align*}
	where the rate trivially follows from Lemma \ref{lemA1caner}.\\
	Similarly, we get 
	\begin{align*}
	\abs{(\widehat{\kappa}_2-\kappa_2) }= \mathcal{O}_P\Big(\varrho_{T}d_TK^{3/2}s_T\Big)=o_P(1).
	\end{align*}
	Callot et al. (2019) showed that $\abs{\kappa_1}=\mathcal{O}(1)$ and $\abs{\kappa_2}=\mathcal{O}(1)$. Therefore, we can get the rate of \eqref{A16}:
	\begin{align*}
	\norm{\widehat{\bw}_{\text{MWC}}-\bw_{\text{MWC}}}_1 = \mathcal{O}_P\Big( \varrho_{T}d_{T}^2K^{3}s_T \Big) = o_P(1).
	\end{align*}
	We now proceed with the MRC weight formulation:
	\begin{align*}
	&\norm{\widehat{\bw}_{\text{MRC}}-\bw_{\text{MRC}}}_1 \leq \frac{\frac{g}{p}\Big[\norm{ (\widehat{\bTheta}-\bTheta) (\widehat{\boldm}-\boldm) }_1+  \norm{ (\widehat{\bTheta}-\bTheta) \boldm }_1 + \norm{ \bTheta (\widehat{\boldm}-\boldm) }_1  \Big] + \abs{\widehat{g}-g}\norm{\bTheta\boldm}_1}{\abs{ \widehat{g}}g}\\
	&\leq \frac{\frac{g}{p}\Big[p\vertiii{ \widehat{\bTheta}-\bTheta }_1\norm{ (\widehat{\boldm}-\boldm) }_{\text{max}}+ p\vertiii{ \widehat{\bTheta}-\bTheta }_1 \norm{ \boldm }_{\text{max}} + p\vertiii{  \bTheta }_1\norm{ (\widehat{\boldm}-\boldm) }_{\text{max}}  \Big] + p\abs{\widehat{g}-g}\vertiii{ \bTheta }_1\norm{\boldm}_{\text{max}}}{\abs{ \widehat{g}}g}\\
	&=\mathcal{O}_P \Big( \varrho_{T}d_TK^{3/2}s_T \cdot \sqrt{\frac{\log p}{T}}  \Big) +  \mathcal{O}_P \Big(\varrho_{T}d_TK^{3/2}s_T \Big) \\&+ \mathcal{O}_P\Big(d_TK^{3/2}\cdot \sqrt{\frac{\log p}{T}} \Big) +  \mathcal{O}_P \Big(\lbrack\varrho_{T}d_TK^{3/2}s_T\rbrack^{1/2}\cdot d_TK^{3/2} \Big) = o_P(1),
	\end{align*}
	where we used Lemmas \ref{theorA1caner}-\ref{lemA1caner} and the fact that $\norm{\widehat{\boldm}-\boldm}_{\text{max}}=\mathcal{O}_P(\sqrt{\log p/T})$.
\subsection{Proof of Theorem 5}
We start with the GMV formulation. Using Lemma \ref{lemA1caner} (a)-(b), we get
\begin{align*}
\abs{\frac{\hat{a}^{-1}}{a^{-1}}-1}=\frac{\abs{a-\hat{a}}}{\abs{\hat{a}}} = \mathcal{O}_P(\varrho_{T}d_TK^{3/2}s_T )=o_P(1).
\end{align*}
Proceeding to the MWC risk exposure, we follow Callot et al. (2019) and introduce the following notation: $x=a\mu^2-2b\mu+d$ and $\hat{x}=\hat{a}\mu-2\hat{b}\mu+\hat{d}$ to rewrite $\widehat{\Phi}_{\text{MWC}}=p^{-1}(\hat{x}/\hat{y})$. As shown in Callot et al. (2019), $y/x=\mathcal{O}(1)$ (see their equation A.42). Furthermore, by Lemma \ref{lemA1caner} (b)-(d)
\begin{align*}
\abs{\hat{x}-x} \leq \abs{\hat{a}-a}\mu^2 + 2\abs{\hat{b}-b}\mu + \abs{\hat{d}-d}=\mathcal{O}_P(\varrho_{T}d_TK^{3/2}s_T)=o_P(1),
\end{align*}
and by Lemma \ref{lemA1caner} (f):
\begin{align*}
\abs{\hat{y}-y} = \abs{\hat{a}\hat{d}-\hat{b}^2-(ad-b^2)} = \mathcal{O}_P(\varrho_{T}d_TK^{3/2}s_T )=o_P(1).
\end{align*}
Using the above and the facts that $y=\mathcal{O}(1)$ and $x=\mathcal{O}(1)$ (which were derived by Callot et al. (2019) in A.45 and A.46), we have
\begin{align*}
\abs{\frac{\widehat{\Phi}_{\text{MWC}} - \Phi_{\text{MWC}} }{\Phi_{\text{MWC}}}} = \abs{ \frac{(\hat{x}-x)y +x(y-\hat{y}) }{\hat{y}y} }\mathcal{O}(1) \mathcal{O}_P(\varrho_{T}d_TK^{3/2}s_T )=o_P(1).
\end{align*}
Finally, to bound MRC risk exposure, we use Lemma \ref{lemA1caner} (e) and rewrite
\begin{align*}
\frac{\abs{g-\hat{g}}}{\abs{\hat{g}}} = \mathcal{O}_P\Big(\lbrack\varrho_{T}d_TK^{3/2}s_T\rbrack^{1/2})=o_P(1).
\end{align*}
 \subsection{Generalization: Sub-Gaussian and Elliptical Distributions} \label{appendixA10}
So far the consistency of the Factor Graphical Lasso in Theorem \ref{theor4} relied on the assumption of the exponential-type tails in \ref{A3}\red{(c)}. Since this tail-behavior may be too restrictive for financial portfolio, we comment on the possibility to relax it. First, recall where \ref{A3}\red{(c)} was used before: we required this assumption in order to establish convergence of unknown factors and loadings in Theorem \ref{theor1}, which was further used to obtain the convergence properties of $\widehat{\bSigma}_{\varepsilon}$ in Theorem \ref{theor2}. Hence, when Assumption \ref{A3}\red{(c)} is relaxed, one needs to find another way to consistently estimate $\bSigma_{\varepsilon}$. We achieve it using the tools developed in \cite{fan2018elliptical}. Specifically, let $\bSigma=\bGamma\bLambda\bGamma^{'}$, where $\bSigma$ is the covariance matrix of returns that follow a factor structure described in equation \eqref{e5.1}. Define $\widehat{\bSigma}, \widehat{\bLambda}_K,\widehat{\bGamma}_K$ to be the estimators of $\bSigma,\bLambda,\bGamma$. We further let $\widehat{\bLambda}_K=\text{diag}(\hat{\lambda}_1,\ldots,\hat{\lambda}_K)$ and $\widehat{\bGamma}_K=(\hat{v}_1,\ldots,\hat{v}_K)$ to be constructed by the first $K$ leading empirical eigenvalues and the corresponding eigenvectors of $\widehat{\bSigma}$ and $\widehat{\bB}\widehat{\bB}'=\widehat{\bGamma}_K\widehat{\bLambda}_K\widehat{\bGamma}_{K}^{'}$. Similarly to \cite{fan2018elliptical}, we require the following bounds on the componentwise maximums of the estimators:
\begin{enumerate}[\textbf{({C}.1)}]
	\item \label{C1} $\norm{\widehat{\bSigma}-\bSigma}_{\text{max}}=\mathcal{O}_P(\sqrt{\log p/T})$,
\end{enumerate}
\begin{enumerate}[\textbf{({C}.2)}]
	\item \label{C2} $\norm{(\widehat{\bLambda}_K-\bLambda)\bLambda^{-1}}_{\text{max}}=\mathcal{O}_P(K \sqrt{\log p/T})$,
\end{enumerate}
\begin{enumerate}[\textbf{({C}.3)}]
	\item \label{C3} $\norm{\widehat{\bGamma}_K-\bGamma}_{\text{max}}=\mathcal{O}_P(K^{1/2} \sqrt{\log p/(Tp)})$.
\end{enumerate}
Let $\widehat{\bSigma}^{SG}$ be the sample covariance matrix, with $\widehat{\bLambda}_{K}^{SG}$ and $\widehat{\bGamma}_{K}^{SG}$ constructed with the first $K$ leading empirical eigenvalues and eigenvectors of $\widehat{\bSigma}^{SG}$ respectively. Also, let $\widehat{\bSigma}^{EL1} = \widehat{\bD}\widehat{\bR}_1\widehat{\bD}$, where $\widehat{\bR}_1$ is obtained using the Kendall's tau correlation coefficients and $\widehat{\bD}$ is a robust estimator of variances constructed using the Huber loss. Furthermore, let $\widehat{\bSigma}^{EL2} = \widehat{\bD}\widehat{\bR}_2\widehat{\bD}$, where $\widehat{\bR}_2$ is obtained using the spatial Kendall's tau estimator. Define $\widehat{\bLambda}_{K}^{EL}$ to be the matrix of the first $K$ leading empirical eigenvalues of $\widehat{\bSigma}^{EL1}$, and $\widehat{\bGamma}_{K}^{EL}$ is the matrix of the first $K$ leading empirical eigenvectors of $\widehat{\bSigma}^{EL2}$. For more details regarding constructing $\widehat{\bSigma}^{SG}$, $\widehat{\bSigma}^{EL1}$ and $\widehat{\bSigma}^{EL2}$ see \cite{fan2018elliptical}, Sections 3 and 4.
\begin{prop}\label{prop2}
	For sub-Gaussian distributions, $\widehat{\bSigma}^{SG}$, $\widehat{\bLambda}_{K}^{SG}$ and $\widehat{\bGamma}_{K}^{SG}$ satisfy \ref{C1}-\ref{C3}.\\
	For elliptical distributions, $\widehat{\bSigma}^{EL1}$, $\widehat{\bLambda}_{K}^{EL}$ and $\widehat{\bGamma}_{K}^{EL}$ satisfy \ref{C1}-\ref{C3}.\\
	When \ref{C1}-\ref{C3} are satisfied, the bounds obtained in Theorems \ref{theor2}-\ref{theor5} continue to hold.
\end{prop}
Proposition \ref{prop2} is essentially a rephrasing of the results obtained in \cite{fan2018elliptical}, Sections 3 and 4. The difference arises due to the fact that we allow $K$ to increase, which is reflected in the modified rates in \ref{C2}-\ref{C3}. As evidenced from the above Proposition, $\widehat{\bSigma}^{EL2}$ is only used for estimating the eigenvectors. This is necessary due to the fact that, in contrast with $\widehat{\bSigma}^{EL2}$, the theoretical properties of the eigenvectors of $\widehat{\bSigma}^{EL}$ are mathematically involved because of the sin function. The FGL for the elliptical distributions will be called the \textit{Robust FGL}. 
	\end{spacing}
\newpage
\begin{spacing}{2}
\section{Additional Simulations} \label{appendixB}
\subsection{Verifying Theoretical Rates} \label{appendixB1}
To compare the empirical rate with the theoretical expressions derived in Theorems \ref{theor3}-\ref{theor5}, we use the facts from Theorem \ref{theor2} that $\omega_{3T}\defeq K^{2}\sqrt{\log p/T} +K^3/\sqrt{p}$ and  $\varrho_{T}^{-1}\omega_{3T}\xrightarrow{p}0$ to introduce the following functions that correspond to the theoretical rates for the choice of parameters in the empirical setting:
\begin{align} &\begin{rcases*}
f_{\vertiii{\cdot}_2} = C_1+ C_2\cdot\log_2(s_T\varrho_{T} ) \label{e5.4} \\
g_{\vertiii{\cdot}_1} = C_3+ C_2\cdot\log_2(d_TK^{3/2}s_T\varrho_{T} )\end{rcases*} \text{for} \ \widehat{\bTheta} \\
&h_{1} = C_4+ C_2\cdot\log_2(\varrho_{T}d_{T}^{2}K^3s_T) \hspace{30.5pt} \text{for} \ \widehat{\bw}_{\text{GMV}},\widehat{\bw}_{\text{MWC}} \label{e5.6}\\
&h_{2}= C_5+ C_6\cdot\log_2([\varrho_{T}s_T]^{1/2}d_{T}^{3/2}K^3) \quad \text{for} \ \widehat{\bw}_{\text{MRC}} \label{e5.7}\\
&h_{3} = C_7+ C_2\cdot\log_2(d_TK^{3/2}s_T\varrho_{T}) \hspace{22pt} \text{for} \ \widehat{\Phi}_{\text{GMV}},\widehat{\Phi}_{\text{MWC}} \label{e5.8}\\
&h_{4}= C_8+ C_{9}\cdot\log_2(d_TK^{3/2}s_T\varrho_{T}) \hspace{22pt} \text{for} \ \widehat{\Phi}_{\text{MRC}} \label{e5.9}
\end{align}
where $C_1, \ldots, C_9$ are constants with $C_6>C_2$ (by Theorem \ref{theor4}), $C_9>C_2$ (by Theorem \ref{theor5}).

\autoref{f1} shows the averaged (over Monte Carlo simulations) errors of the estimators of $\bTheta$, $\bw$ and $\Phi$ versus the sample size $T$ in the logarithmic scale (base 2). In order to confirm the theoretical findings from Theorems \ref{theor3}-\ref{theor5}, we also plot the theoretical rates of convergence given by the functions in \eqref{e5.4}-\eqref{e5.9}. We verify that the empirical and theoretical rates are matched. Since the convergence rates for GMV and MWC portfolio weights $\bw$ and risk exposures $\Phi$ are very similar, we only report the former. Note that as predicted by Theorem \ref{theor3}, the rate of convergence of the precision matrix in $\vertiii{\cdot}_2$-norm is faster than the rate in $\vertiii{\cdot}_1$-norm. Furthermore, the convergence rate of the GMV, MWC and MRC portfolio weights and risk exposures are close to the rate of the precision matrix $\bTheta$ in $\vertiii{\cdot}_1$-norm, which is confirmed by Theorem \ref{theor4}. As evidenced by \autoref{f1}, the convergence rate of the MRC risk exposure is slower than the rate of GMV and MWC exposures. This finding is in accordance with Theorem \ref{theor5} and it is also consistent with the empirical findings that indicate higher overall risk associated with MRC portfolios.	
\begin{figure}[hbt!]
	\centering
	\includegraphics[width=0.7\textwidth]{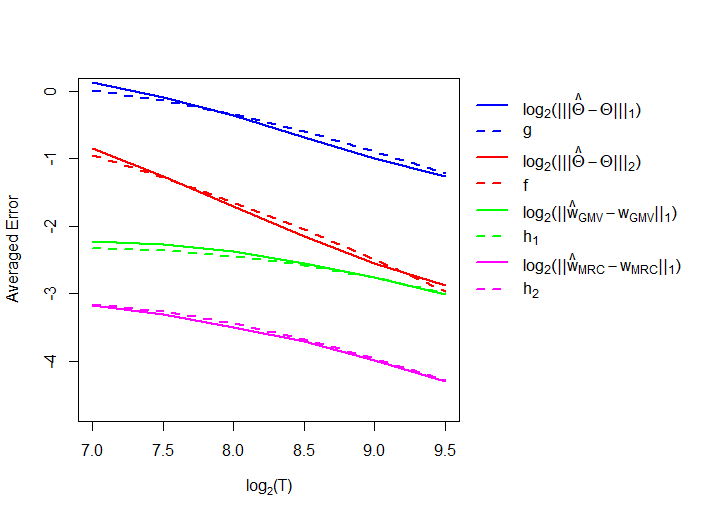}\\
	\vspace{-25pt}
	\includegraphics[width=0.7\textwidth]{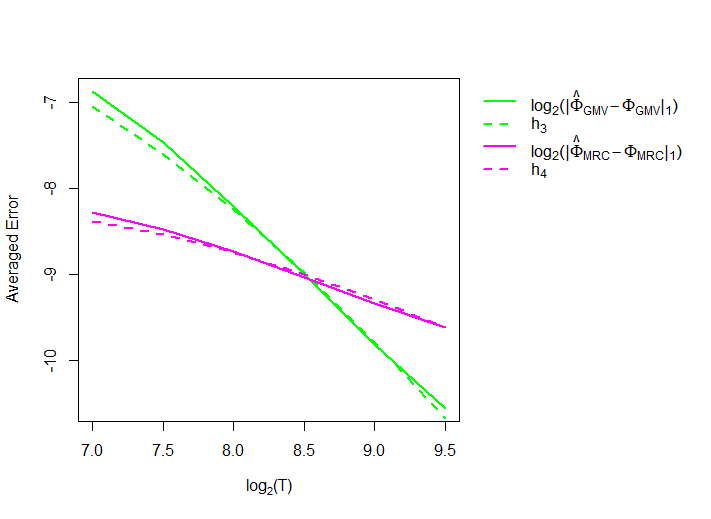}
	\bigskip
	\caption{\textbf{Averaged empirical errors (solid lines) and theoretical rates of convergence (dashed lines) on logarithmic scale: $p = T^{0.85}$, $K = 2(\log T)^{0.5}$, $s_T = \mathcal{O}(T^{0.05})$.}}
	\label{f1}
\end{figure}
\cleardoublepage
\subsection{Results for Case 1} \label{appendixB1A}
We compare the performance of FGL with the alternative models listed at the beginning of Section 5 for Case 1. The only instance when FGL is strictly but slightly dominated occurs in \autoref{f2}: POET outperforms FGL in terms of convergence of precision matrix in the spectral norm. This is different from Case 2 in \autoref{f4} where FGL outperforms all the competing models.
\begin{figure}[hbt!]
	\centering
	\includegraphics[width=0.98\textwidth]{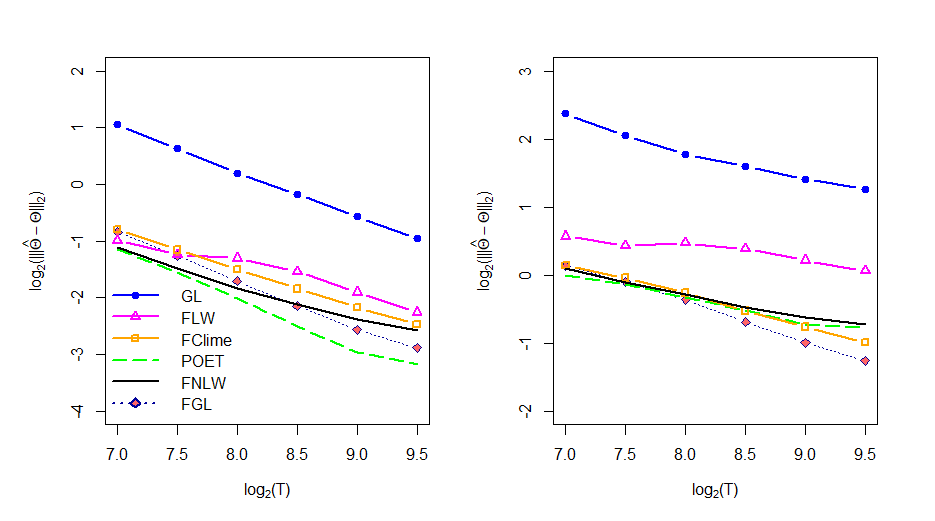}
	\bigskip
	\caption{\textbf{Averaged errors of the estimators of $\bTheta$ for Case 1 on logarithmic scale: $p = T^{0.85}$, $K = 2(\log T)^{0.5}$, $s_T = \mathcal{O}(T^{0.05})$.}}
	\label{f2}
\end{figure}
\begin{figure}[hbt!]
	\centering
	\includegraphics[width=0.98\textwidth]{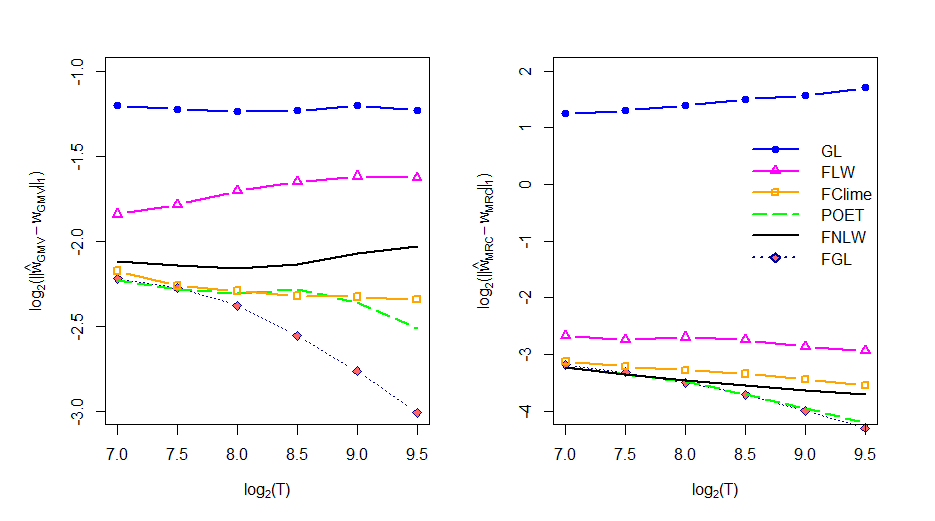}
	\bigskip
	\caption{\textbf{Averaged errors of the estimators of $\bw_{\text{GMV}}$ (left) and $\bw_{\text{MRC}}$ (right) for Case 1 on logarithmic scale: $p = T^{0.85}$, $K = 2(\log T)^{0.5}$, $s_T = \mathcal{O}(T^{0.05})$.}}
	\label{f3}
\end{figure}
\begin{figure}[hbt!]
	\centering
	\includegraphics[width=0.98\textwidth]{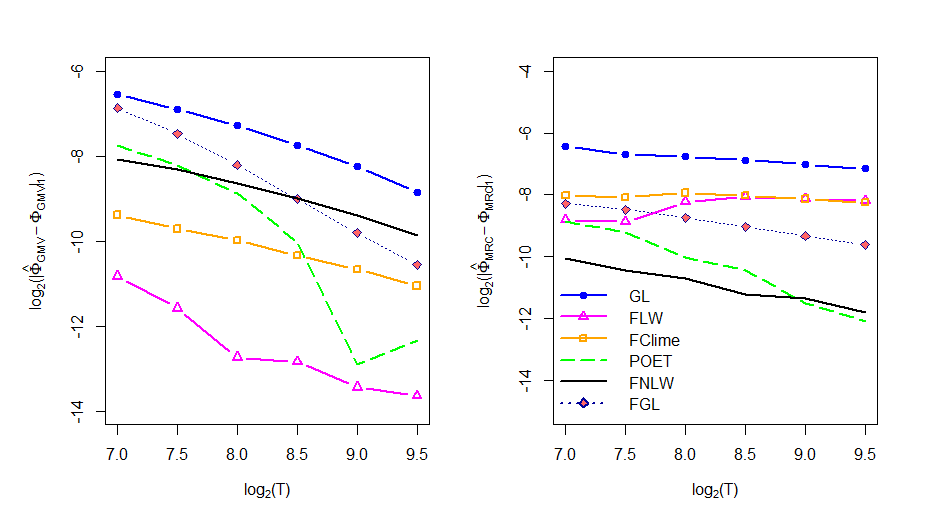}
	\bigskip
	\caption{\textbf{Averaged errors of the estimators of $\Phi_{\text{GMV}}$ (left) and $\Phi_{\text{MRC}}$ (right) for Case 1 on logarithmic scale: $p = T^{0.85}$, $K = 2(\log T)^{0.5}$, $s_T = \mathcal{O}(T^{0.05})$.}}
	\label{f3a}
\end{figure}
\newpage
\subsection{Robust FGL} \label{appendixB2}
The DGP for elliptical distributions is similar to \cite{fan2018elliptical}: let $(\bf_t, \bvarepsilon_{t})$ from \eqref{e5.1} jointly follow the multivariate t-distribution with the degrees of freedom $\nu$. When $\nu = \infty$, this corresponds to the multivariate normal distribution, smaller values of $\nu$ are associated with thicker tails. We draw $T$ independent samples of $(\bf_t, \bvarepsilon_{t})$ from the multivariate t-distribution with zero mean and covariance matrix $\bSigma = \text{diag}(\bSigma_{f}, \bSigma_{\varepsilon})$, where $\bSigma_{f} = \bI_{K}$. To construct $\bSigma_{\varepsilon}$ we use a Toeplitz structure parameterized by $\rho = 0.5$, which leads to the sparse $\bTheta_{\varepsilon} = \bSigma^{-1}_{\varepsilon}$. The rows of $\bB$ are drawn from $\mathcal{N}(\bm{0}, \bI_{K})$. We let $p = T^{0.85}$, $K = 2(\log T)^{0.5}$ and $T = \lbrack 2^h \rbrack, \ \text{for} \ h\in \{7,7.5,8,\ldots,9.5\}$. \autoref{f6}-\ref{f8} report the averaged (over Monte Carlo simulations) estimation errors (in the logarithmic scale, base 2) for $\bTheta$ and two portfolio weights (GMV and MRC) using FGL and Robust FGL for $\nu=4.2$. Noticeably, the performance of FGL for estimating the precision matrix is comparable with that of Robust FGL: this suggests that our FGL algorithm is insensitive to heavy-tailed distributions even without additional modifications. Furthermore, FGL outperforms its Robust counterpart in terms of estimating portfolio weights, as evidenced by \autoref{f8}. We further compare the performance of FGL and Robust FGL for different degrees of freedom: \autoref{f9} reports the log-ratios (base 2) of the averaged (over Monte Carlo simulations) estimation errors for $\nu=4.2$, $\nu = 7$ and $\nu = \infty$. The results for the estimation of $\bTheta$ presented in \autoref{f9} are consistent with the findings in \cite{fan2018elliptical}: Robust FGL outperforms the non-robust counterpart for thicker tails.
\begin{figure}[hbt!]
	\centering
	\includegraphics[width=0.98\textwidth]{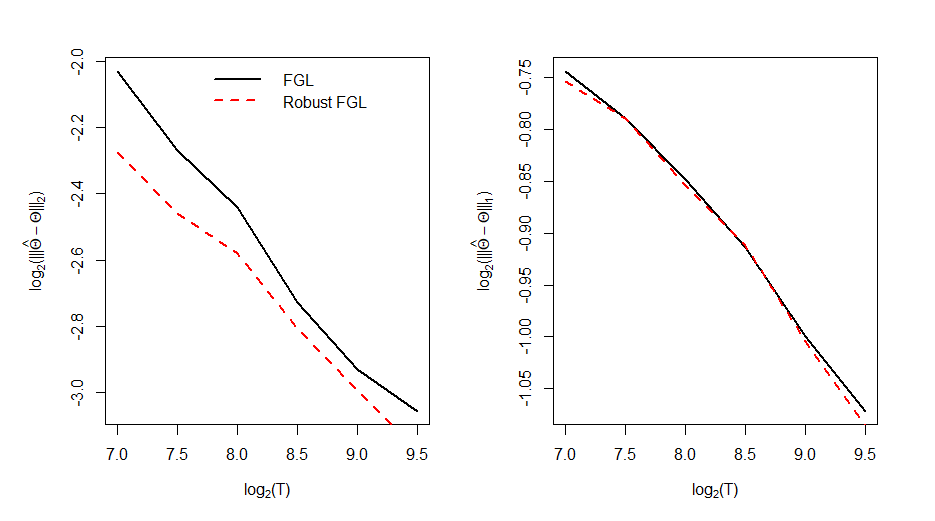}
	\bigskip
	\caption{\textbf{Averaged errors of the estimators of $\bTheta$ on logarithmic scale: $p = T^{0.85}$, $K = 2(\log T)^{0.5}$, $\nu=4.2$.}}
	\label{f6}
\end{figure}
\begin{figure}[hbt!]
	\centering
	\includegraphics[width=0.98\textwidth]{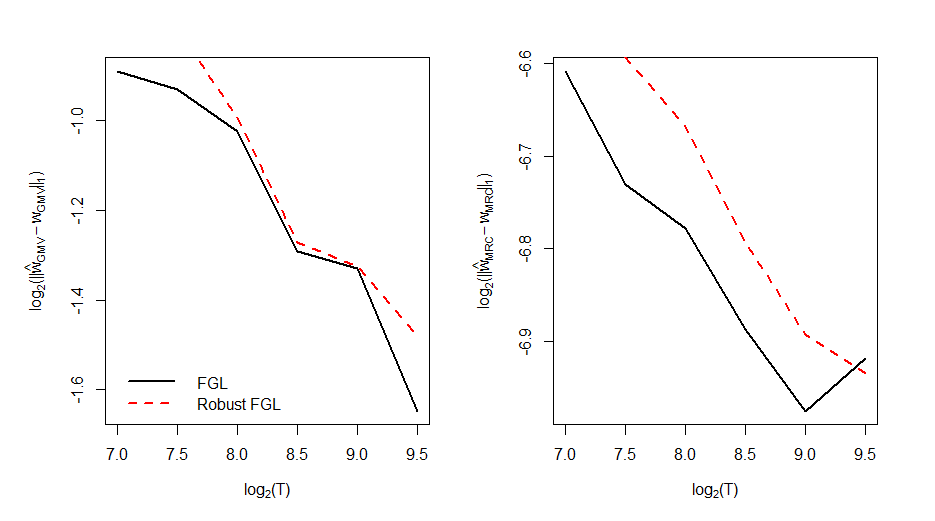}
	\bigskip
	\caption{\textbf{Averaged errors of the estimators of $\bw_{\text{GMV}}$ (left) and $\bw_{\text{MRC}}$ (right) on logarithmic scale: $p = T^{0.85}$, $K = 2(\log T)^{0.5}$, $\nu=4.2$.}}
	\label{f8}
\end{figure}
\newpage
\begin{figure}[!htbp]
	\centering
	\includegraphics[width=0.98\textwidth]{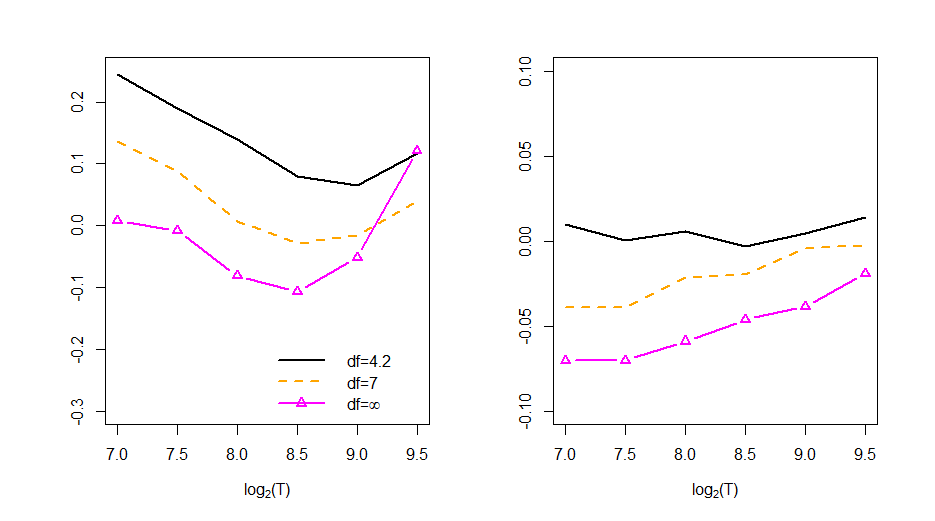}
	\bigskip
	\caption{\textbf{Log ratios (base 2) of the averaged errors of the FGL and the Robust FGL estimators of $\bTheta$: $\log_2\Big(\frac{\vertiii{\widehat{\bTheta}-\bTheta}_2}{\vertiii{\widehat{\bTheta}_{\text{R}}-\bTheta}_2}\Big)$ (left), $\log_2\Big(\frac{\vertiii{\widehat{\bTheta}-\bTheta}_1}{\vertiii{\widehat{\bTheta}_{\text{R}}-\bTheta}_1}\Big)$ (right): $p = T^{0.85}$, $K = 2(\log T)^{0.5}$.}}
	\label{f9}
\end{figure}
\cleardoublepage
\subsection{Relaxing Pervasiveness Assumption}\label{appendixB3}
As pointed out by \cite{onatski2013POET}, the data on 100 industrial portfolios shows that there are no large gaps between eigenvalues $i$ and $i+1$ of the sample covariance data except for $i=1$. However, as is commonly believed, such data contains at least three factors. Therefore, the factor pervasiveness assumption suggests the existence of a large gap for $i\geq3$. In order to examine sensitivity of portfolios to the pervasiveness assumption and quantify the degree of pervasiveness, we use the same DGP as in \eqref{e51}-\eqref{e52}, but with $\sigma_{\varepsilon,ij}=\rho^{\abs{i-j}}$ and $K=3$. We consider $\rho\in \{0.4,0.5,0.6,0.7,0.8,0.9\}$ which corresponds to $\lambda_3/\lambda_4\in\{3.1,2.7,2.6,2.2,1.5,1.1\}$. In other words, as $\rho$ increases, the systematic-idiosyncratic gap measured by $\hat{\lambda}_3/\hat{\lambda}_4$ decreases. \autoref{tab_sim1}-\ref{tab_sim2} report the mean quality of the estimators for portfolio weights and risk over 100 replications for $T=300$ and $p\in \{300,400\}$. The sample size and the number of regressors are chosen to closely match the values from the empirical application. POET and Projected POET are the most sensitive to a reduction in the gap between the leading and bounded eigenvalues which is evident from a dramatic deterioration in the quality of these estimators. The remaining methods, including FGL, exhibit robust performance. Since the behavior of the estimators for portfolio weights is similar to that of the estimators of precision matrix, we only report the former for the ease of presentation. For $(T,p)=(300,300)$, FClime shows the best performance followed by FGL and FLW, whereas for $(T,p)=(300,400)$ FGL takes the lead. Interestingly, despite inferior performance of POET and Projected POET in terms of estimating portfolio weights, risk exposure of the portfolios based on these estimators is competitive with the other approaches.
\begin{table}[]	
	\caption{Sensitivity of portfolio weights and risk exposure when the gap between the diverging and bounded eigenvalues decreases: $(T,p)=(300,300)$.}
	\label{tab_sim1}
	\centering
	\resizebox{0.95\textwidth}{!}{%
		\begin{tabular}{@{}ccccccc@{}}
			\toprule
			& $\rho=0.4$ & $\rho=0.5$ & $\rho=0.6$ & $\rho=0.7$ & $\rho=0.8$ & $\rho=0.9$ \\
			& ($\lambda_3/\lambda_4=3.1$) & ($\lambda_3/\lambda_4=2.7$) & ($\lambda_3/\lambda_4=2.6$) & ($\lambda_3/\lambda_4=2.2$) & ($\lambda_3/\lambda_4=1.5$) & ($\lambda_3/\lambda_4=1.1$) \\ \midrule
			& \multicolumn{6}{c}{$\norm{\widehat{\bw}_{\text{GMV}}-\bw_{\text{GMV}}}_1$} \\ \midrule
			FGL & 2.3198 & 2.3465 & 2.5177 & 2.4504 & 2.5010 & 2.7319 \\
			FClime & 1.9554 & 1.9359 & 1.9795 & 1.9103 & 1.9813 & 1.9948 \\
			FLW & 2.3445 & 2.3948 & 2.5328 & 2.4715 & 2.5918 & 3.0515 \\			
			FNLW & 2.2381 & 2.3009 & 2.3293 & 2.5497 & 2.9039 & 3.1980 \\			
			POET & 47.6746 & 82.1873 & 43.9722 & 54.1131 & 157.6963 & 235.8119 \\
			Projected POET & 9.6335 & 7.8669 & 10.1546 & 10.6205 & 12.1795 & 15.2581 \\ \midrule
			& \multicolumn{6}{c}{$\abs{\widehat{\Phi}_{\text{GMV}}-\Phi_{\text{GMV}}}$} \\ \midrule
			FGL & 0.0033 & 0.0032 & 0.0034 & 0.0027 & 0.0021 & 0.0023 \\
			FClime & 0.0012 & 0.0012 & 0.0012 & 0.0011 & 0.0010 & 0.0010 \\
			FLW & 0.0049 & 0.0052 & 0.0061 & 0.0056 & 0.0049 & 0.0059 \\			
			FNLW & 0.0055 & 0.0060 & 0.0054 & 0.0052 & 0.0066 & 0.0057 \\			
			POET & 0.0070 & 0.0122 & 0.0058 & 0.0063 & 0.0103 & 0.0160 \\
			Projected POET & 0.0021 & 0.0022 & 0.0019 & 0.0019 & 0.0018 & 0.0026 \\ \midrule
			& \multicolumn{6}{c}{$\norm{\widehat{\bw}_{\text{MWC}}-\bw_{\text{MWC}}}_1$} \\ \midrule
			FGL & 2.3766 & 2.4108 & 2.7411 & 2.6094 & 2.5669 & 3.4633 \\
			FClime & 2.0502 & 2.0279 & 2.2901 & 2.1400 & 2.1028 & 3.0737 \\
			FLW & 2.4694 & 2.5132 & 2.8902 & 2.7315 & 2.7210 & 4.0248 \\			
			FNLW & 2.7268 & 2.3060 & 2.8984 & 3.5902 & 2.9232 & 3.2076 \\			
			POET & 49.8603 & 34.2024 & 469.3605 & 108.1529 & 74.8016 & 99.4561 \\
			Projected POET & 9.0261 & 7.4028 & 8.1899 & 9.4806 & 11.9642 & 13.3890 \\ \midrule
			& \multicolumn{6}{c}{$\abs{\widehat{\Phi}_{\text{MWC}}-\Phi_{\text{MWC}}}$} \\ \midrule
			FGL & 0.0033 & 0.0032 & 0.0034 & 0.0027 & 0.0021 & 0.0024 \\
			FClime & 0.0012 & 0.0012 & 0.0013 & 0.0011 & 0.0010 & 0.0009 \\
			FLW & 0.0050 & 0.0053 & 0.0062 & 0.0057 & 0.0050 & 0.0059 \\			
			FNLW & 0.0055 & 0.0060 & 0.0055 & 0.0053 & 0.0066 & 0.0057 \\			
			POET & 0.0068 & 0.0047 & 0.0363 & 0.0092 & 0.0060 & 0.0056 \\
			Projected POET & 0.0022 & 0.0022 & 0.0020 & 0.0020 & 0.0018 & 0.0027 \\ \midrule
			& \multicolumn{6}{c}{$\norm{\widehat{\bw}_{\text{MRC}}-\bw_{\text{MRC}}}_1$} \\ \midrule
			FGL & 0.4872 & 0.1793 & 1.0044 & 0.6332 & 1.4568 & 2.3353 \\
			FClime & 0.5160 & 0.2148 & 1.0188 & 0.6694 & 1.4855 & 2.3519 \\
			FLW & 0.5333 & 0.2279 & 1.0345 & 0.6734 & 1.4904 & 2.3691 \\			
			FNLW & 0.8365 & 1.1285 & 1.1181 & 1.4419 & 1.7694 & 2.4612 \\			
			POET & NaN & NaN & NaN & NaN & NaN & NaN \\
			Projected POET & 0.7414 & 0.6383 & 1.6686 & 1.8013 & 2.3297 & 3.2791 \\ \midrule
			& \multicolumn{6}{c}{$\abs{\widehat{\Phi}_{\text{MRC}}-\Phi_{\text{MRC}}}$} \\ \midrule
			FGL & 0.0004 & 0.0003 & 0.0025 & 0.0007 & 0.0021 & 0.0071 \\
			FClime & 0.0005 & 0.0003 & 0.0024 & 0.0004 & 0.0016 & 0.0062 \\
			FLW & 0.0002 & 0.0002 & 0.0021 & 0.0003 & 0.0018 & 0.0066 \\			
			FNLW & 0.0062 & 0.0062 & 0.0069 & 0.0119 & 0.0059 & 0.0143 \\					
			POET & NaN & NaN & NaN & NaN & NaN & NaN \\
			Projected POET & 0.0003 & 0.0003 & 0.0027 & 0.0031 & 0.0069 & 0.0062 \\ \bottomrule
		\end{tabular}%
	}
\end{table}
\begin{table}[]
	\caption{Sensitivity of portfolio weights and risk exposure when the gap between the diverging and bounded eigenvalues decreases: $(T,p)=(300,400)$.}
	\centering
	\resizebox{0.95\textwidth}{!}{%
		\begin{tabular}{@{}ccccccc@{}}
			\toprule
			& $\rho=0.4$ & $\rho=0.5$ & $\rho=0.6$ & $\rho=0.7$ & $\rho=0.8$ & $\rho=0.9$ \\
			& ($\lambda_3/\lambda_4=3.1$) & ($\lambda_3/\lambda_4=2.7$) & ($\lambda_3/\lambda_4=2.6$) & ($\lambda_3/\lambda_4=2.2$) & ($\lambda_3/\lambda_4=1.5$) & ($\lambda_3/\lambda_4=1.1$) \\ \midrule
			& \multicolumn{6}{c}{$\norm{\widehat{\bw}_{\text{GMV}}-\bw_{\text{GMV}}}_1$} \\ \midrule
			FGL & 1.6900 & 1.8134 & 1.8577 & 1.8839 & 1.9843 & 2.0692 \\
			FClime & 1.9073 & 1.9524 & 1.9997 & 1.9490 & 1.9898 & 2.0330 \\
			FLW & 2.0239 & 2.0945 & 2.1195 & 2.1235 & 2.2473 & 2.4745 \\			
			FNLW & 2.0316 & 2.0790 & 2.1927 & 2.2503 & 2.4143 & 2.4710 \\			
			POET & 18.7934 & 28.0493 & 155.8479 & 32.4197 & 41.8098 & 71.5811 \\
			Projected POET & 7.8696 & 8.4915 & 8.8641 & 10.7522 & 11.2092 & 19.0424 \\ \midrule
			& \multicolumn{6}{c}{$\abs{\widehat{\Phi}_{\text{GMV}}-\Phi_{\text{GMV}}}$} \\ \midrule
			FGL & 8.62E-04 & 9.22E-04 & 7.23E-04 & 7.31E-04 & 6.83E-04 & 5.73E-04 \\
			FClime & 8.40E-04 & 8.27E-04 & 8.02E-04 & 7.87E-04 & 7.36E-04 & 6.71E-04 \\
			FLW & 1.59E-03 & 1.73E-03 & 1.57E-03 & 1.68E-03 & 1.69E-03 & 1.54E-03 \\			
			FNLW & 2.24E-03 & 2.10E-03 & 1.83E-03 & 1.88E-03 & 2.07E-03 & 1.29E-03\\			
			POET & 1.11E-03 & 1.46E-03 & 3.59E-03 & 1.27E-03 & 1.88E-03 & 2.51E-03 \\
			Projected POET & 8.97E-04 & 8.80E-04 & 6.83E-04 & 6.79E-04 & 7.98E-04 & 6.55E-04 \\ \midrule
			& \multicolumn{6}{c}{$\norm{\widehat{\bw}_{\text{MWC}}-\bw_{\text{MWC}}}_1$} \\ \midrule
			FGL & 1.9034 & 2.2843 & 1.9118 & 3.2569 & 2.7055 & 2.8812 \\
			FClime & 2.1193 & 2.4024 & 2.0540 & 3.3487 & 2.7277 & 2.8593 \\
			FLW & 2.2573 & 2.5809 & 2.1790 & 3.5728 & 3.0072 & 3.3164 \\			
			FNLW & 2.3207 & 3.3335 & 3.5518 & 3.4282 & 2.6446 & 4.8827 \\			
			POET & 15.8824 & 100.1419 & 56.9827 & 33.6483 & 38.8961 & 103.0434 \\
			Projected POET & 6.5386 & 7.2169 & 7.8583 & 9.7342 & 12.1420 & 17.7368 \\ \midrule
			& \multicolumn{6}{c}{$\abs{\widehat{\Phi}_{\text{MWC}}-\Phi_{\text{MWC}}}$} \\ \midrule
			FGL & 8.72E-04 & 9.41E-04 & 7.26E-04 & 7.99E-04 & 7.12E-04 & 6.08E-04 \\
			FClime & 8.52E-04 & 8.49E-04 & 8.06E-04 & 8.32E-04 & 7.50E-04 & 6.86E-04 \\
			FLW & 1.59E-03 & 1.74E-03 & 1.57E-03 & 1.71E-03 & 1.70E-03 & 1.56E-03 \\			
			FNLW & 2.25E-03 & 2.22E-03 & 1.89E-03 & 1.91E-03 & 2.08E-03 & 1.56E-03\\			
			POET & 1.14E-03 & 4.91E-03 & 1.78E-03 & 1.45E-03 & 1.57E-03 & 2.93E-03 \\
			Projected POET & 9.19E-04 & 9.20E-04 & 7.11E-04 & 7.04E-04 & 8.26E-04 & 6.78E-04 \\ \midrule
			& \multicolumn{6}{c}{$\norm{\widehat{\bw}_{\text{MRC}}-\bw_{\text{MRC}}}_1$} \\ \midrule
			FGL & 0.6683 & 0.7390 & 1.3103 & 1.5195 & 1.7124 & 3.0935 \\
			FClime & 0.6903 & 0.7635 & 1.3238 & 1.5403 & 1.7415 & 3.1180 \\
			FLW & 0.7132 & 0.7828 & 1.3430 & 1.5549 & 1.7517 & 3.1364 \\			
			FNLW & 0.4909 & 1.2121 & 1.4974 & 1.1996 & 1.8020 & 3.2989 \\			
			POET & NaN & NaN & NaN & NaN & NaN & NaN \\
			Projected POET & 1.6851 & 1.4434 & 1.9628 & 2.6182 & 2.7716 & 4.1753 \\ \midrule
			& \multicolumn{6}{c}{$\abs{\widehat{\Phi}_{\text{MRC}}-\Phi_{\text{MRC}}}$} \\ \midrule
			FGL & 1.02E-03 & 9.73E-04 & 4.63E-03 & 4.49E-03 & 3.23E-03 & 8.73E-03 \\
			FClime & 1.14E-03 & 1.01E-03 & 4.55E-03 & 4.22E-03 & 2.70E-03 & 7.72E-03 \\
			FLW & 6.62E-04 & 5.54E-04 & 4.19E-03 & 4.01E-03 & 2.71E-03 & 8.11E-03 \\			
			FNLW & 2.73E-04 & 6.93E-03 & 5.11E-03 & 1.93E-03 & 6.42E-03 & 2.98E-02\\			
			POET & NaN & NaN & NaN & NaN & NaN & NaN \\
			Projected POET & 3.59E-03 & 1.20E-03 & 1.49E-03 & 2.58E-03 & 7.86E-03 & 1.39E-02 \\ \bottomrule
		\end{tabular}%
	}
	\label{tab_sim2}
\end{table}
\newpage
\section{Additional Empirical Results} \label{appendixC}
This Appendix contains the description of the procedure used to estimate unknown factors and loadings using PCA (Appendix \ref{appendixC0}), and additional empirical results with portfolio performance for monthly data (Appendix \ref{appendixC4}) and verifying robustness of FGL towards different training periods (Appendix \ref{appendixC1}), different target risk and return (Appendix \ref{appendixC2}), and subperiod analyses for MWC and GMV portfolios (Appendix \ref{appendixC3}).
\subsection{Estimating Unknown Factors and Loadings} \label{appendixC0}
\begin{remark} \label{remark1}
	In practice, the number of common factors, $K$, is unknown and needs to be estimated. One of the standard and commonly used approaches is to determine $K$ in a data-driven way (\cite{Bai2002,kapetanios2010testing}). As an example, in their paper \cite{fan2013POET} adopt the approach from \cite{Bai2002}. However, all of the aforementioned papers deal with a fixed number of factors. Therefore, we need to adopt a different criteria since $K$ is allowed to grow in our setup. For this reason, we use the methodology by \cite{Li2017_Increasing_Factors}: let $\bb_{i,K}$ and $\bf_{t,K}$ denote $K\times 1$ vectors of loadings and factors when $K$ needs to be estimated, and $\bB_K$ is a $p\times K$ matrix of stacked $\bb_{i,K}$. Define 
	\begin{equation}\label{e3.38}
		V(K)=\min_{\bB_{K},\bF_K}\frac{1}{pT}\sum_{i=1}^{p} \sum_{t=1}^{T} \Big(r_{it}-\frac{1}{\sqrt{K}}\bb'_{i,K}\bf_{t,K}  \Big)^2,
	\end{equation}
	where the minimum is taken over $1\leq K\leq K_{\textup{max}}$, subject to normalization $\bB'_{K}\bB_{K}/p=\bI_{K}$. Hence, $\bar{\bF}'_K=\sqrt{K}\bR'\bB_K/p$. Define $\widehat{\bF}'_K = \bar{\bF}'_K(\bar{\bF}_K\bar{\bF}'_K/T)^{1/2}$, which is a rescaled estimator of the factors that is used to determine the number of factors when $K$ grows with the sample size. We then apply the following procedure described in \cite{Li2017_Increasing_Factors} to estimate $K$:
	\begin{equation}\label{e3.39}
		\widehat{K} = \argmin_{1\leq K\leq K_{\textup{max}}} \ln (V(K,\hat{\bF}_K)) + Kg(p,T),
	\end{equation}
	where $1\leq K\leq K_{\textup{max}}=o(\min\{ p^{1/17},T^{1/16} \})$ and $g(p,T)$ is a penalty function of $(p,T)$ such that (i) $K_{\textup{max}}\cdot g(p,T)\rightarrow 0$ and (ii) $C_{p,T,K_{\textup{max}}}^{-1} \cdot g(p,T) \rightarrow \infty$ with $C_{p,T,K_{\textup{max}}} = \mathcal{O}_P\Big(\max \Big[\frac{K^{3}_{\textup{max}}}{\sqrt{p}}, \frac{K^{5/2}_{\textup{max}}}{\sqrt{T}}  \Big]  \Big)$. The choice of the penalty function is similar to \cite{Bai2002}.
	Throughout the paper we let $\widehat{K}$ be the solution to \eqref{e3.39}.
\end{remark}
\newpage
\subsection{Monthly Data} \label{appendixC4}
Similarly to daily data, we use monthly returns of the components of the S\&P500. The data is fetched from CRSP and Compustat using SAS interface.
The full sample for the monthly data has 480 observations on 355 stocks from January 1, 1980 - December 1, 2019. We use January 1, 1980 - December 1, 1994 (180 obs) as a training (estimation) period and January 1, 1995 - December 1, 2019 (300 obs) as the out-of-sample test period. At the end of each month, prior to portfolio construction, we remove stocks with less than 15 years of historical stock return data.  We set the return target $\mu=0.7974\%$ which is equivalent to $10\%$ yearly return when compounded. The target level of risk for the weight-constrained and risk-constrained Markowitz portfolio (MWC and MRC) is set at $\sigma=0.05$ which is the standard deviation of the monthly excess returns of the S\&P500 index in the first training set. Transaction costs are taken to be the same as for the daily returns in Section 6.

\autoref{tab1} reports the results for monthly data. Some comments are in order: \textbf{(1)} interestingly, MRC produces portfolio return and Sharpe Ratio that are mostly higher than those for the weight-constrained allocations MWC and GMV. This means that relaxing the constraint that portfolio weights sum up to one leads to a large increase in the out-of-sample Sharpe Ratio and portfolio return which has not been previously well-studied in the empirical finance literature. \textbf{(2)} Similarly to the results from \autoref{tab3}, FGL outperforms the competitors including EW and Index in terms of the out-of-sample Sharpe Ratio and turnover. \textbf{(3)} Similarly to the results in \autoref{tab3}, the observable Fama-French factors produce the FGL portfolios with higher return and higher out-of-sample Sharpe Ratio compared to the FGL portfolios based on statistical factors. Again, this increase in return is not followed by higher risk. \textbf{(4)} To further verify that the shrinkage is functioning as desired and the estimated $\bTheta_{\varepsilon}$ is indeed sparse we include several visualizations. Figure \ref{fig_rev2} reports optimally tuned values of $\lambda$ (please refer to Section 3 of the main manuscript for a discussion on choosing the optimal shrinkage intensity) over the estimation period. Figure \ref{fig_rev1} plots the proportion of zero elements in the precision matrix of the idiosyncratic part, $\widehat{\bTheta}_{\varepsilon}$, corresponding to the optimal values of $\lambda = \hat{\lambda}$, and several fixed values of $\lambda$ for monthly data over the testing period. Extracting the common factors significantly reduces partial correlations of the error terms, rendering $\widehat{\bTheta}_{\varepsilon}$ sparse over the testing period: the number of zeroes for the optimally tuned $\lambda$ varies from 74.5\%-98.8\%. Figure \ref{fig_rev2fixed} plots the Sharpe Ratio of GMV portfolios for a set of fixed values of $\lambda \in \{0.005, 0.01, 0.05, 0.08, 0.1, 0.12, 0.15, 0.17, 0.2, 0.25, 0.3, 0.4, 0.5\}$. In other words, instead of using optimally tuned $\lambda$, we fix its value throughout the whole testing period and report the corresponding SR of such portfolios. For comparison, the SR that corresponds to the optimally tuned $\lambda$ is equal to $0.2023$, which is significantly higher than the SR achieved for any fixed $\lambda$ confirming the importance of selecting shrinkage intensity optimally.

We would like to emphasize that the selection of the tuning parameter is critically important in the literature on graphical models, which is why we build our tuning methodology on the Bayesian Information Criteria (BIC) as used and described in \cite{koike2019biased, Bishop2006,pourahmadi2013high,Sara2018} among others (the detailed treatment relevant to our paper can be found on p.13 of the main manuscript). The advantage of SR obtained using optimally tuned $\lambda$ highlights the importance of tuning and demonstrates that $\lambda$ changes over time. Hence, using a fixed value is expected to produce suboptimal performance.

We now elaborate on the discrepancy between the SR with the optimal vs fixed $\lambda$. Please note that Figure \ref{fig_rev2} should not be compared with Figure \ref{fig_rev2fixed}. In contrast to Figure \ref{fig_rev2fixed}, the range of $\lambda$ in Figure \ref{fig_rev2} is selected optimally by minimizing BIC. In other words, SR is not the objective function that we use for selecting the tuning parameter. To demonstrate the relevant range of $\lambda$ selected by the BIC we have included Figure \ref{fig_rev3fixed} that shows optimally selected $\lambda$ for six different rolling windows.

 
 \begin{figure}[!htbp]
 	\centering
 	\includegraphics[width=0.8\textwidth]{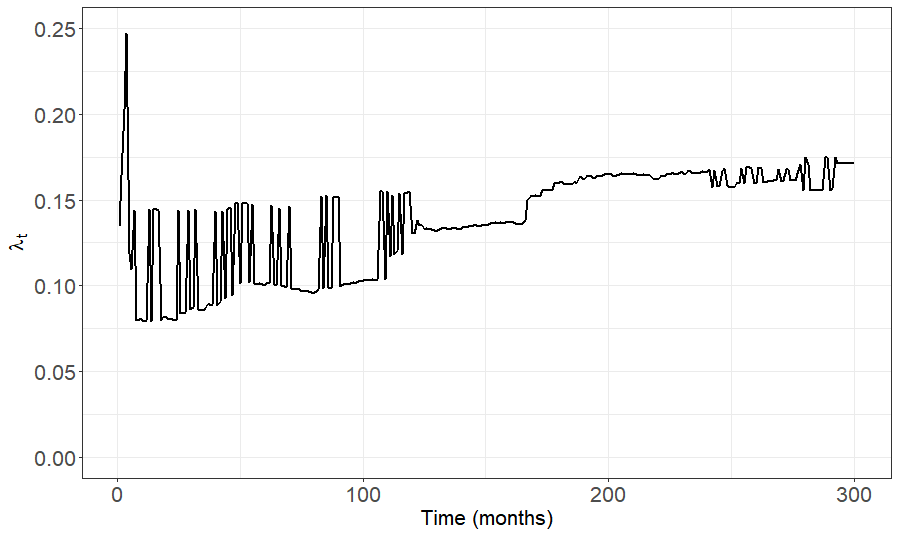}
 	\bigskip
 	\caption{Optimally tuned values of $\lambda$ over the testing period.}
 	\label{fig_rev2}
 \end{figure}
 \begin{figure}[!htbp]
	\centering
	\includegraphics[width=0.8\textwidth]{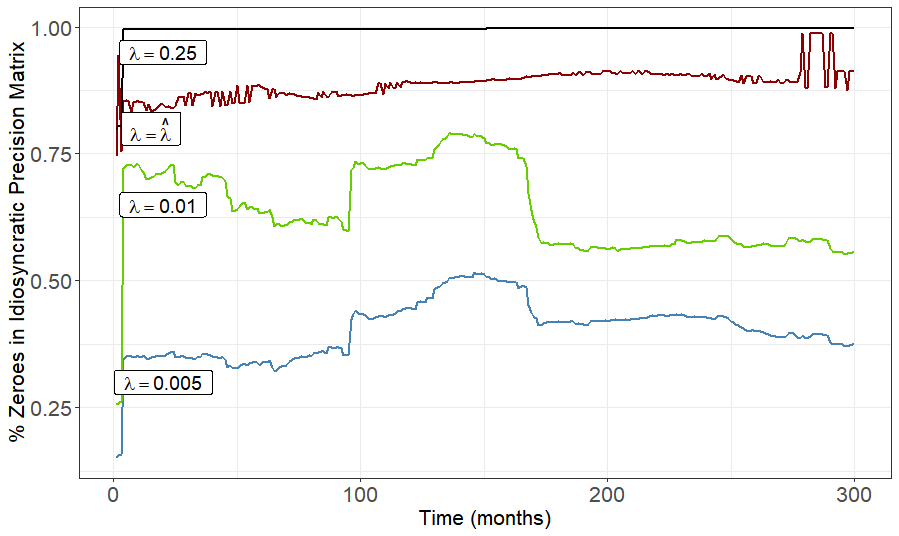}
	\bigskip
	\caption{Proportion of zero elements in $\widehat{\bTheta}_{\varepsilon}$ with respect to the total number of elements in a lower-triangular part of $\widehat{\bTheta}_{\varepsilon}$ (diagonals are excluded) corresponding to the optimal values of $\lambda = \hat{\lambda}$, and several fixed values of $\lambda$.}
	\label{fig_rev1}
\end{figure}
 \begin{figure}[!htbp]
	\centering
	\includegraphics[width=0.8\textwidth]{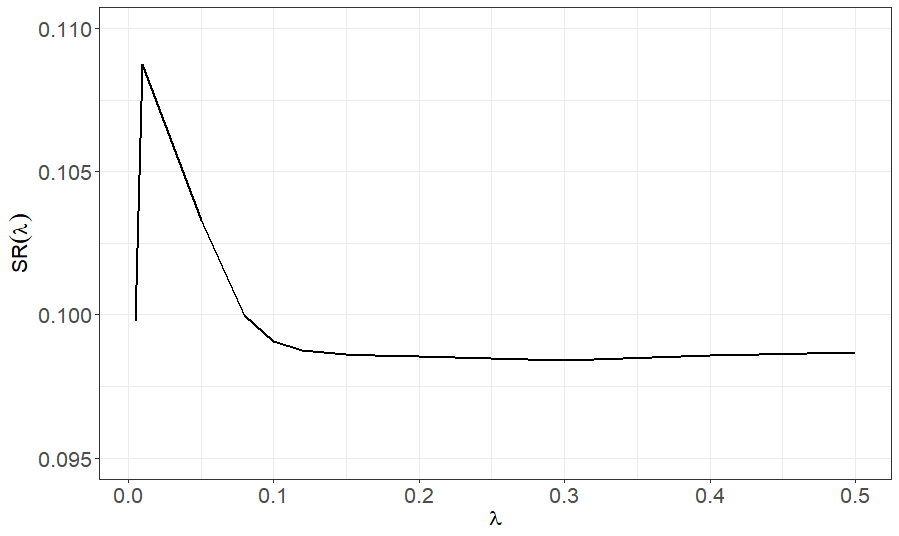}
	\bigskip
	\caption{Sharpe Ratios for GMV portfolios associated with fixed $\lambda$.}
	\label{fig_rev2fixed}
\end{figure}
 \begin{figure}[!htbp]
	\centering
	\includegraphics[width=\textwidth]{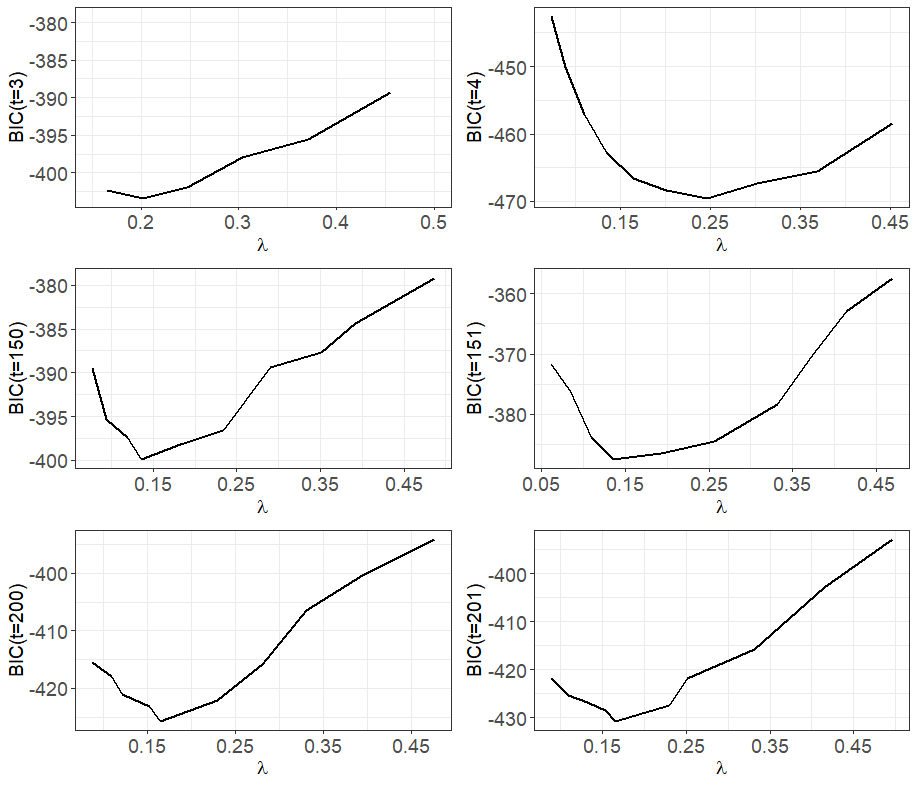}
	\bigskip
	\caption{BICs for several rolling windows indexed by $t=1,\ldots,300$.}
	\label{fig_rev3fixed}
\end{figure}
\begin{landscape}
	\begin{table}[]
		\caption{ \small{Monthly portfolio returns, risk, SR and turnover. In the upper part corresponding to the results w/o transactions costs, p-values are in parentheses. In the lower part corresponding to the results with transaction costs, $^{***}$ indicates p-value $<$ 0.01, $^{**}$ indicates p-value  $<$ 0.05, and $^{*}$ indicates p-value $<$ 0.10.  In-sample: January 1, 1980 - December 31, 1995 (180 obs), Out-of-sample: January 1, 1995 - December 31, 2019 (300 obs).} }
		\centering
		\label{tab1}
		\resizebox{0.8\textwidth}{!}{%
\begin{tabular}{ccccccccccccc} 
	\toprule
	& \multicolumn{4}{c}{Markowitz Risk-Constrained} & \multicolumn{4}{c}{Markowitz Weight-Constrained} & \multicolumn{4}{c}{Global Minimum-Variance} \\ 
	\midrule
	& \textbf{Return} & \textbf{Risk} & \textbf{SR} & \textbf{Turnover} & \textbf{Return} & \textbf{Risk} & \textbf{SR} & \textbf{Turnover} & \textbf{Return} & \textbf{Risk} & \textbf{SR} & \textbf{Turnover} \\ 
	\midrule
	\textbf{Without TC} &  &  &  &  &  &  &  &  &  &  &  &  \\
	EW & 0.0081 & 0.0519 & 0.1553 & - & 0.0081 & 0.0519 & 0.1553 & - & 0.0081 & 0.0519 & 0.1553 & - \\
	Index & 0.0063 & 0.0453 & 0.1389 & - & 0.0063 & 0.0453 & 0.1389 & - & 0.0063 & 0.0453 & 0.1389 & - \\
	FGL & 0.0256 & 0.0828 & \begin{tabular}[c]{@{}c@{}}0.3099\\(0.0799)~ ~\end{tabular} & - & 0.0059 & 0.0329 & \begin{tabular}[c]{@{}c@{}}0.1804\\(0.0430)\end{tabular} & - & 0.0065 & 0.0321 & \begin{tabular}[c]{@{}c@{}}0.2023\\(0.046)\end{tabular} & - \\
	FClime & 0.0372 & 0.2337 & \begin{tabular}[c]{@{}c@{}}0.1593\\(0.2715)\end{tabular} & - & 0.0067 & 0.0471 & \begin{tabular}[c]{@{}c@{}}0.1434\\(0.0791)\end{tabular} & - & 0.0076 & 0.0466 & \begin{tabular}[c]{@{}c@{}}0.1643\\(0.047)\end{tabular} & - \\
	FLW & 0.0296 & 0.1049 & \begin{tabular}[c]{@{}c@{}}0.2817\\(0.0879)\end{tabular} & - & 0.0059 & 0.0353 & \begin{tabular}[c]{@{}c@{}}0.1662\\(0.0791)\end{tabular} & - & 0.0063 & 0.0353 & \begin{tabular}[c]{@{}c@{}}0.1774\\(0.047)\end{tabular} & - \\
	FNLW & 0.0264 & 0.0925 & \begin{tabular}[c]{@{}c@{}}0.2853\\(0.0879)\end{tabular} & - & 0.0060 & 0.0333 & \begin{tabular}[c]{@{}c@{}}0.1793\\(0.0430)\end{tabular} & - & 0.0064 & 0.0332 & \begin{tabular}[c]{@{}c@{}}0.1930\\(0.046)\end{tabular} & - \\
	POET & NaN & NaN & NaN & - & -0.1041 & 2.0105 & \begin{tabular}[c]{@{}c@{}}-0.0518\\(0.9925)\end{tabular} & - & 0.5984 & 11.0064 & \begin{tabular}[c]{@{}c@{}}0.0544\\(0.6344)\end{tabular} & - \\
	Projected POET & 0.0583 & 0.3300 & \begin{tabular}[c]{@{}c@{}}0.1766\\(0.2715)\end{tabular} & - & 0.0058 & 0.0546 & \begin{tabular}[c]{@{}c@{}}0.1056\\(0.0791)\end{tabular} & - & 0.0069 & 0.0612 & \begin{tabular}[c]{@{}c@{}}0.1128\\(0.2693)\end{tabular} & - \\
	FGL (FF1) & 0.0275 & 0.0800 & \begin{tabular}[c]{@{}c@{}}0.3433\\(0.0659)\end{tabular} & - & 0.0061 & 0.0316 & \begin{tabular}[c]{@{}c@{}}0.1941\\(0.0415)\end{tabular} & - & 0.0073 & 0.0302 & \begin{tabular}[c]{@{}c@{}}0.2427\\(0.035)\end{tabular} & - \\
	FGL (FF3) & 0.0274 & 0.0797 & \begin{tabular}[c]{@{}c@{}}0.3437\\(0.0659)\end{tabular} & - & 0.0061 & 0.0314 & \begin{tabular}[c]{@{}c@{}}0.1955\\(0.0415)\end{tabular} & - & 0.0073 & 0.0300 & \begin{tabular}[c]{@{}c@{}}0.2440\\(0.035)\end{tabular} & - \\
	FGL (FF5) & 0.0273 & 0.0793 & \begin{tabular}[c]{@{}c@{}}0.3443\\(0.0659)\end{tabular} & - & 0.0061 & 0.0314 & \begin{tabular}[c]{@{}c@{}}0.1943\\(0.0415)\end{tabular} & - & 0.0073 & 0.0300 & \begin{tabular}[c]{@{}c@{}}0.2426\\(0.035)\end{tabular} & - \\
	FF1 & 0.0403 & 0.2250 & \begin{tabular}[c]{@{}c@{}}0.1789\\(0.2715)\end{tabular} & - & 0.0025 & 0.0548 & \begin{tabular}[c]{@{}c@{}}0.0452\\(0.9318)\end{tabular} & - & 0.0043 & 0.0546 & \begin{tabular}[c]{@{}c@{}}0.0781\\(0.6344)\end{tabular} & - \\
	FF3 & 0.0389 & 0.2022 & \begin{tabular}[c]{@{}c@{}}0.1926\\(0.2715)\end{tabular} & - & 0.0032 & 0.0528 & \begin{tabular}[c]{@{}c@{}}0.0610\\(0.9318)\end{tabular} & - & 0.0047 & 0.0517 & \begin{tabular}[c]{@{}c@{}}0.0915\\(0.6344)\end{tabular} & - \\
	FF5 & 0.0354 & 0.1803 & \begin{tabular}[c]{@{}c@{}}0.1962\\(0.2715)\end{tabular} & - & 0.0036 & 0.0531 & \begin{tabular}[c]{@{}c@{}}0.0670\\(0.9318)\end{tabular} & - & 0.0048 & 0.0513 & \begin{tabular}[c]{@{}c@{}}0.0945\\(0.6344)\end{tabular} & - \\ 
	\midrule
	\textbf{With TC} &  &  &  &  &  &  &  &  &  &  &  &  \\
	EW & 0.0080 & 0.0520 & 0.1538 & 0.0630 & 0.0080 & 0.0520 & 0.1538 & 0.0630 & 0.0080 & 0.0520 & 0.1538 & 0.0630 \\
	FGL & 0.0222 & 0.0828 & 0.2682* & 3.1202 & 0.0050 & 0.0329 & 0.1525* & 0.8786 & 0.0056 & 0.0321 & 0.1740** & 0.8570 \\
	FClime & 0.0334 & 0.2334 & 0.1429 & 4.9174 & 0.0062 & 0.0471 & 0.1307 & 0.5945 & 0.0071 & 0.0466 & 0.1522* & 0.5528 \\
	FLW & 0.0237 & 0.1052 & 0.2257 & 5.5889 & 0.0043 & 0.0353 & 0.1231 & 1.5166 & 0.0048 & 0.0354 & 0.1343 & 1.5123 \\
	FNLW & 0.0224 & 0.0927 & 0.2415* & 3.7499 & 0.0049 & 0.0334 & 0.1463* & 1.0812 & 0.0053 & 0.0333 & 0.1596* & 1.0793 \\
	POET & NaN & NaN & NaN & NaN & -0.1876 & 1.7274 & -0.1086 & 152.3298 & 1.0287 & 14.2676 & 0.0721 & 354.6043 \\
	Projected POET & 0.0166 & 0.2859 & 0.0579 & 69.7600 & -0.0002 & 0.0540 & -0.0044 & 5.9131 & -0.0002 & 0.0613 & -0.0027 & 7.0030 \\
	FGL (FF1) & 0.0243 & 0.0800 & 0.3036* & 2.8514 & 0.0054 & 0.0317 & 0.1692* & 0.7513 & 0.0066 & 0.0302 & 0.2176** & 0.7095 \\
	FGL (FF3) & 0.0242 & 0.0797 & 0.3037* & 2.8708 & 0.0054 & 0.0314 & 0.1703* & 0.7545 & 0.0066 & 0.0300 & 0.2186** & 0.7127 \\
	FGL (FF5) & 0.0241 & 0.0793 & 0.3037* & 2.8857 & 0.0053 & 0.0315 & 0.1686* & 0.7630 & 0.0065 & 0.0300 & 0.2167** & 0.7224 \\
	FF1 & 0.0169 & 0.2331 & 0.0767 & 23.3910 & -0.0023 & 0.0545 & -0.0415 & 4.6257 & -0.0004 & 0.0543 & -0.0079 & 4.5751 \\
	FF3 & 0.0185 & 0.2268 & 0.0924 & 20.6137 & -0.0013 & 0.0524 & -0.0243 & 4.3667 & 0.0003 & 0.0514 & 0.0059 & 4.2956 \\
	FF5 & 0.0164 & 0.2254 & 0.0918 & 18.5514 & -0.0008 & 0.0527 & -0.0145 & 4.2134 & 0.0005 & 0.0508 & 0.0108 & 4.1681 \\
	\bottomrule
\end{tabular}%
		}
		
	\end{table}
\end{landscape}
\subsection{Portfolio Performance for Longer Training Periods} \label{appendixC1}
This section examines the performance of the methods when training periods were increased. Tables \ref{tab2a} and \ref{tab4} report the results: the conclusions that we highlighted when analyzing Tables \ref{tab1} and \ref{tab3} continue to hold. We observed an interesting finding: for MRC portfolios (both monthly and daily), a larger training period changed the values of portfolio return and risk for all methods, however, their relative value illustrated by the SR remained unchanged. This is due to the fact that MRC portfolios maximize SR subject to either target risk or target return constraints:
	\begin{equation} \nonumber
		\max_{\bw} \frac{\boldm'\bw}{\sqrt{\bw'\bSigma\bw}} \ \text{s.t.}\ \text{(i)} \ \boldm'\bw\geq\mu \ \text{or} \text{(ii)} \ \bw'\bSigma\bw\leq \sigma^2,
	\end{equation}
	when  $\mu=\sigma\sqrt{\boldm'\bTheta\boldm}$, the solution to either of the constraints is given by
	$\bw_{\text{MRC}}=\frac{\sigma}{\sqrt{\boldm'\bTheta\boldm}}\bTheta\boldm$. Hence, even though the training period was increased, the maximum achievable SR remained the same since neither target risk nor target return were changed.
\clearpage
\newpage
	\begin{landscape}
	\begin{table}[]
		\caption{\small {Monthly portfolio returns, risk, SR and turnover. In the upper part corresponding to the results w/o transactions costs, p-values are in parentheses. In the lower part corresponding to the results with transaction costs, $^{***}$ indicates p-value $<$ 0.01, $^{**}$ indicates p-value  $<$ 0.05, and $^{*}$ indicates p-value $<$ 0.10.  In-sample: January 1, 1980 - December 31, 1999 (240 obs), Out-of-sample: January 1, 2000 - December 31, 2019 (240 obs).} }
		\centering
		\label{tab2a}
		\resizebox{0.8\textwidth}{!}{%
\begin{tabular}{ccccccccccccc} 
	\toprule
	& \multicolumn{4}{c}{Markowitz Risk-Constrained} & \multicolumn{4}{c}{Markowitz Weight-Constrained} & \multicolumn{4}{c}{Global Minimum-Variance} \\ 
	\midrule
	& \textbf{Return} & \textbf{Risk} & \textbf{SR} & \textbf{Turnover} & \textbf{Return} & \textbf{Risk} & \textbf{SR} & \textbf{Turnover} & \textbf{Return} & \textbf{Risk} & \textbf{SR} & \textbf{Turnover} \\ 
	\midrule
	\textbf{Without TC} &  &  &  &  &  &  &  &  &  &  &  &  \\
	EW & 0.0061 & 0.0029 & 0.1126 & - & 0.0061 & 0.0029 & 0.1126 & - & 0.0061 & 0.0029 & 0.1126 & - \\
	Index & 0.0032 & 0.0022 & 0.0692 & - & 0.0032 & 0.0022 & 0.0692 & - & 0.0032 & 0.0022 & 0.0692 & - \\
	FGL & 0.0223 & 0.0805 & \begin{tabular}[c]{@{}c@{}}0.2770\\(0.0609)\end{tabular} & - & 0.0054 & 0.0343 & \begin{tabular}[c]{@{}c@{}}0.1581\\(0.0939)\end{tabular} & - & 0.0062 & 0.0333 & \begin{tabular}[c]{@{}c@{}}0.1848\\(0.1189)\end{tabular} & - \\
	FClime & 0.0339 & 0.2642 & \begin{tabular}[c]{@{}c@{}}0.1285\\(0.0989)\end{tabular} & - & 0.0068 & 0.0510 & \begin{tabular}[c]{@{}c@{}}0.1341\\(0.1119)\end{tabular} & - & 0.0072 & 0.0483 & \begin{tabular}[c]{@{}c@{}}0.1482\\(0.1419)\end{tabular} & - \\
	FLW & 0.0268 & 0.1030 & \begin{tabular}[c]{@{}c@{}}0.2606\\(0.0609)\end{tabular} & - & 0.0047 & 0.0400 & \begin{tabular}[c]{@{}c@{}}0.1173\\(0.1119)\end{tabular} & - & 0.0056 & 0.0396 & \begin{tabular}[c]{@{}c@{}}0.1423\\(0.1419)\end{tabular} & - \\
	FNLW & 0.0234 & 0.0887 & \begin{tabular}[c]{@{}c@{}}0.2633\\(0.0609)\end{tabular} & - & 0.0050 & 0.0360 & \begin{tabular}[c]{@{}c@{}}0.1397\\(0.1119)\end{tabular} & - & 0.0060 & 0.0354 & \begin{tabular}[c]{@{}c@{}}0.1701\\(0.1189)\end{tabular} & - \\
	POET & NaN & NaN & NaN & - & 0.2800 & 3.7119 & \begin{tabular}[c]{@{}c@{}}0.0754\\(0.2617)\end{tabular} & - & -0.0239 & 0.6870 & \begin{tabular}[c]{@{}c@{}}-0.0348\\(0.7729)\end{tabular} & - \\
	Projected POET & -0.0371 & 0.8637 & \begin{tabular}[c]{@{}c@{}}-0.0430\\(0.9137)\end{tabular} & - & 0.0032 & 0.0447 & \begin{tabular}[c]{@{}c@{}}0.0708\\(0.2617)\end{tabular} & - & 0.0014 & 0.0823 & \begin{tabular}[c]{@{}c@{}}0.0174\\(0.7729)\end{tabular} & - \\
	FGL(FF1) & 0.0224 & 0.0775 & \begin{tabular}[c]{@{}c@{}}0.2891\\(0.0579)\end{tabular} & - & 0.0056 & 0.0325 & \begin{tabular}[c]{@{}c@{}}0.1721\\(0.0839)\end{tabular} & - & 0.0064 & 0.0307 & \begin{tabular}[c]{@{}c@{}}0.2070\\(0.0889)\end{tabular} & - \\
	FGL(FF3) & 0.0223 & 0.0772 & 0.2885 & - & 0.0056 & 0.0324 & \begin{tabular}[c]{@{}c@{}}0.1728\\(0.0839)\end{tabular} & - & 0.0064 & 0.0307 & \begin{tabular}[c]{@{}c@{}}0.2075\\(0.0889)\end{tabular} & - \\
	FGL(FF5) & 0.0222 & 0.0769 & \begin{tabular}[c]{@{}c@{}}0.2887\\(0.0579)\end{tabular} & - & 0.0056 & 0.0325 & \begin{tabular}[c]{@{}c@{}}0.1719\\(0.0839)\end{tabular} & - & 0.0063 & 0.0307 & \begin{tabular}[c]{@{}c@{}}0.2058\\(0.0889)\end{tabular} & - \\
	FF1 & -0.0074 & 1.0260 & \begin{tabular}[c]{@{}c@{}}-0.0072\\(0.9512)\end{tabular} & - & 0.0090 & 0.1158 & \begin{tabular}[c]{@{}c@{}}0.0777\\(0.2617)\end{tabular} & - & 0.0099 & 0.1162 & \begin{tabular}[c]{@{}c@{}}0.0853\\(0.1738)\end{tabular} & - \\
	FF3 & 0.0234 & 0.9301 & \begin{tabular}[c]{@{}c@{}}0.0252\\(0.9512)\end{tabular} & - & 0.0073 & 0.1011 & \begin{tabular}[c]{@{}c@{}}0.0726\\(0.2617)\end{tabular} & - & 0.0083 & 0.1040 & \begin{tabular}[c]{@{}c@{}}0.0796\\(0.1738)\end{tabular} & - \\
	FF5 & 0.0111 & 0.7406 & \begin{tabular}[c]{@{}c@{}}0.0149\\(0.9512)\end{tabular} & - & 0.0064 & 0.0906 & \begin{tabular}[c]{@{}c@{}}0.0708\\(0.2617)\end{tabular} & - & 0.0056 & 0.0942 & \begin{tabular}[c]{@{}c@{}}0.0590\\(0.1738)\end{tabular} & - \\ 
	\midrule
	\textbf{With TC} &  &  &  &  &  &  &  &  &  &  &  &  \\
	EW & 0.0060 & 0.0029 & 0.1109 & 0.0627 & 0.0060 & 0.0029 & 0.1109 & 0.0627 & 0.0060 & 0.0029 & 0.1109 & 0.0627 \\
	FGL & 0.0190 & 0.1826 & 0.2351* & 3.0398 & 0.0044 & 0.0344 & 0.1280 & 0.9940 & 0.0052 & 0.0334 & 0.1548 & 0.9519 \\
	FClime & 0.0305 & 0.7518 & 0.1158 & 5.1354 & 0.0063 & 0.0510 & 0.1231 & 0.5606 & 0.0067 & 0.0484 & 0.1378 & 0.4962 \\
	FLW & 0.0203 & 0.1992 & 0.1961* & 6.3365 & 0.0028 & 0.0401 & 0.0703 & 1.8813 & 0.0038 & 0.0397 & 0.0952 & 1.8528 \\
	FNLW & 0.0195 & 0.2199 & 0.2193* & 3.6238 & 0.0038 & 0.0361 & 0.1066 & 1.1847 & 0.0049 & 0.0355 & 0.1367 & 1.1596 \\
	POET & NaN & NaN & NaN & NaN & 0.1707 & 3.4886 & 0.0489 & 151.5813 & -0.0846 & 0.5652 & -0.1496 & 76.9912 \\
	Projected POET & -0.0556 & 0.1885 & -0.0757* & 33.2465 & -0.0014 & 0.0451 & -0.0313 & 4.3028 & -0.0031 & 0.0734 & -0.0429 & 15.5429 \\
	FGL(FF1) & 0.0194 & 0.2708 & 0.2497* & 2.7486 & 0.0047 & 0.0326 & 0.1454 & 0.8299 & 0.0055 & 0.0308 & 0.1803 & 0.7703 \\
	FGL(FF3) & 0.0192 & 0.1754 & 0.2485* & 2.7682 & 0.0047 & 0.0325 & 0.1460 & 0.8343 & 0.0055 & 0.0307 & 0.1807 & 0.7754 \\
	FGL(FF5) & 0.0192 & 0.1753 & 0.2488* & 2.7807 & 0.0047 & 0.0325 & 0.1451 & 0.8377 & 0.0055 & 0.0307 & 0.1789 & 0.7815 \\
	FF1 & -0.1143 & 0.1753 & -0.1499 & 338.4639 & -0.0039 & 0.1160 & -0.0337 & 13.1081 & -0.0041 & 0.1167 & -0.0353 & 14.2580 \\
	FF3 & -0.0763 & 0.3416 & -0.1047 & 285.5053 & -0.0038 & 0.1010 & -0.0380 & 11.3194 & -0.0042 & 0.1043 & -0.0407 & 12.6905 \\
	FF5 & -0.0547 & 0.3229 & -0.0922 & 604.2117 & -0.0038 & 0.0904 & -0.0423 & 10.3659 & -0.0052 & 0.0944 & -0.0548 & 10.9466 \\
	\bottomrule
\end{tabular}%
		}
		
	\end{table}
\end{landscape}
\begin{landscape}
	\begin{table}[]
		\centering
		\caption{\small{Daily portfolio returns, risk, SR and turnover. In the upper part corresponding to the results w/o transactions costs, p-values are in parentheses. In the lower part corresponding to the results with transaction costs, $^{***}$ indicates p-value $<$ 0.01, $^{**}$ indicates p-value  $<$ 0.05, and $^{*}$ indicates p-value $<$ 0.10. In-sample: January 20, 2000 - January 25, 2005  (1260 obs), Out-of-sample: January 26, 2005 - January 31, 2020 (3780 obs).}}
		\label{tab4}
		\resizebox{0.9\textwidth}{!}{%
\begin{tabular}{ccccccccccccc} 
	\toprule
	& \multicolumn{4}{c}{Markowitz Risk-Constrained} & \multicolumn{4}{c}{Markowitz Weight-Constrained} & \multicolumn{4}{c}{Global Minimum-Variance} \\ 
	\midrule
	& \textbf{Return} & \textbf{Risk} & \textbf{SR} & \textbf{Turnover} & \textbf{Return} & \textbf{Risk} & \textbf{SR} & \textbf{Turnover} & \textbf{Return} & \textbf{Risk} & \textbf{SR} & \textbf{Turnover} \\ 
	\midrule
	\textbf{Without TC} &  &  &  &  &  &  &  &  &  &  &  &  \\
	EW & 2.19E-04 & 1.98E-02 & 0.0111 & - & 2.19E-04 & 1.98E-02 & 0.0111 & - & 2.19E-04 & 1.98E-02 & 0.0111 & - \\
	Index & 2.15E-04 & 1.16E-02 & 0.0185 & - & 2.15E-04 & 1.16E-02 & 0.0185 & - & 2.15E-04 & 1.16E-02 & 0.0185 & - \\
	FGL & 8.86E-04 & 2.90E-02 & \begin{tabular}[c]{@{}c@{}}0.0305\\(0.0450)\end{tabular} & - & 3.51E-04 & 7.07E-03 & \begin{tabular}[c]{@{}c@{}}0.0496\\(0.0020)\end{tabular} & - & 3.51E-04 & 6.98E-03 & \begin{tabular}[c]{@{}c@{}}0.0503\\(0.0025)\end{tabular} & - \\
	FClime & 1.30E-03 & 8.36E-02 & \begin{tabular}[c]{@{}c@{}}0.0156\\(0.2513)\end{tabular} & - & 2.41E-04 & 1.04E-02 & \begin{tabular}[c]{@{}c@{}}0.0231\\(0.0315)\end{tabular} & - & 2.75E-04 & 1.10E-02 & \begin{tabular}[c]{@{}c@{}}0.0250\\(0.0415)\end{tabular} & - \\
	FLW & 4.24E-04 & 2.88E-02 & \begin{tabular}[c]{@{}c@{}}0.0147\\(0.2513)\end{tabular} & - & 3.12E-04 & 7.06E-03 & \begin{tabular}[c]{@{}c@{}}0.0443\\(0.0025)\end{tabular} & - & 3.15E-04 & 7.41E-03 & \begin{tabular}[c]{@{}c@{}}0.0425\\(0.0033)\end{tabular} & - \\
	FNLW & 3.20E-04 & 5.33E-02 & \begin{tabular}[c]{@{}c@{}}0.0060\\(0.6397)\end{tabular} & - & 3.23E-04 & 7.01E-03 & \begin{tabular}[c]{@{}c@{}}0.0461\\(0.0020)\end{tabular} & - & 3.49E-04 & 8.44E-03 & \begin{tabular}[c]{@{}c@{}}0.0414\\(0.0033)\end{tabular} & - \\
	POET & NaN & NaN & NaN & - & 5.39E-03 & 3.82E-01 & \begin{tabular}[c]{@{}c@{}}0.0141\\(0.1384)\end{tabular} & - & -8.23E-05 & 9.49E-02 & \begin{tabular}[c]{@{}c@{}}-0.0009\\(0.9218)\end{tabular} & - \\
	Projected POET & 7.86E-04 & 7.74E-02 & \begin{tabular}[c]{@{}c@{}}0.0101\\(0.2513)\end{tabular} & - & -1.70E-04 & 1.09E-02 & \begin{tabular}[c]{@{}c@{}}-0.0156\\(0.9713)\end{tabular} & - & -1.78E-04 & 1.15E-02 & \begin{tabular}[c]{@{}c@{}}-0.0155\\(0.9218)\end{tabular} & - \\
	FGL(FF1) & 6.03E-04 & 3.56E-02 & \begin{tabular}[c]{@{}c@{}}0.0169\\(0.2513)\end{tabular} & - & 3.58E-04 & 6.98E-03 & \begin{tabular}[c]{@{}c@{}}0.0513\\(0.0010)\end{tabular} & - & 3.68E-04 & 7.02E-03 & \begin{tabular}[c]{@{}c@{}}0.0523\\(0.0025)\end{tabular} & - \\
	FGL(FF3) & 6.02E-04 & 3.56E-02 & \begin{tabular}[c]{@{}c@{}}0.0169\\(0.2513)\end{tabular} & - & 3.58E-04 & 6.98E-03 & \begin{tabular}[c]{@{}c@{}}0.0514\\(0.0010)\end{tabular} & - & 3.68E-04 & 7.02E-03 & \begin{tabular}[c]{@{}c@{}}0.0524\\(0.0025)\end{tabular} & - \\
	FGL(FF5) & 6.01E-04 & 3.56E-02 & \begin{tabular}[c]{@{}c@{}}0.0169\\(0.2513)\end{tabular} & - & 3.57E-04 & 6.98E-03 & \begin{tabular}[c]{@{}c@{}}0.0512\\(0.0010)\end{tabular} & - & 3.67E-04 & 7.02E-03 & \begin{tabular}[c]{@{}c@{}}0.0522\\(0.0025)\end{tabular} & - \\
	FF1 & 6.13E-04 & 5.22E-02 & \begin{tabular}[c]{@{}c@{}}0.0117\\(0.2513)\end{tabular} & - & 2.93E-04 & 7.23E-03 & \begin{tabular}[c]{@{}c@{}}0.0405\\(0.0032)\end{tabular} & - & 2.99E-04 & 8.06E-03 & \begin{tabular}[c]{@{}c@{}}0.0371\\(0.0285)\end{tabular} & - \\
	FF3 & 6.13E-04 & 5.22E-02 & \begin{tabular}[c]{@{}c@{}}0.0117\\(0.2513)\end{tabular} & - & 2.93E-04 & 7.23E-03 & \begin{tabular}[c]{@{}c@{}}0.0405\\(0.0032)\end{tabular} & - & 2.99E-04 & 8.06E-03 & \begin{tabular}[c]{@{}c@{}}0.0371\\(0.0285)\end{tabular} & - \\
	FF5 & 6.13E-04 & 5.22E-02 & \begin{tabular}[c]{@{}c@{}}0.0117\\(0.2513)\end{tabular} & - & 2.93E-04 & 7.23E-03 & \begin{tabular}[c]{@{}c@{}}0.0405\\(0.0032)\end{tabular} & - & 2.99E-04 & 8.06E-03 & \begin{tabular}[c]{@{}c@{}}0.0371\\(0.0285)\end{tabular} & - \\ 
	\midrule
	\textbf{With TC} &  &  &  &  &  &  &  &  &  &  &  &  \\
	EW & 1.87E-04 & 1.98E-02 & 0.0094 & 0.0294 & 1.87E-04 & 1.98E-02 & 0.0094 & 0.0294 & 1.87E-04 & 1.98E-02 & 0.0094 & 0.0294 \\
	FGL & 5.07E-04 & 8.37E-02 & 0.0175 & 0.3845 & 2.64E-04 & 7.09E-03 & 0.0372*** & 0.0882 & 2.66E-04 & 7.00E-03 & 0.038** & 0.0863 \\
	FClime & 4.36E-04 & 3.03E-01 & 0.0052 & 0.9279 & 2.10E-04 & 1.04E-02 & 0.0201* & 0.0333 & 2.51E-04 & 1.10E-02 & 0.0228** & 0.0266 \\
	FLW & -1.26E-05 & 8.62E-02 & -0.0004 & 0.4399 & 1.86E-04 & 7.09E-03 & 0.0263** & 0.1267 & 1.91E-04 & 7.44E-03 & 0.0256** & 0.1251 \\
	FNLW & -4.74E-04 & 1.05E-01 & -0.009 & 0.806 & 1.38E-04 & 7.07E-03 & 0.0195 & 0.1856 & 1.65E-04 & 8.48E-03 & 0.0194 & 0.1846 \\
	POET & NaN & NaN & NaN & NaN & -6.70E-03 & 1.95E-01 & -0.0344 & 25.1772 & -7.54E-03 & 9.20E-02 & -0.082 & 11.6087 \\
	Projected POET & -3.51E-03 & 9.21E-02 & -0.0485 & 10.9077 & -6.81E-04 & 1.12E-02 & -0.0610 & 0.5127 & -7.37E-04 & 1.19E-02 & -0.0621 & 0.5609 \\
	FGL(FF1) & 2.02E-04 & 1.09E-01 & 0.0057 & 0.4028 & 2.68E-04 & 7.00E-03 & 0.0383*** & 0.0899 & 2.80E-04 & 7.04E-03 & 0.0397** & 0.0877 \\
	FGL(FF3) & 2.02E-04 & 8.39E-02 & 0.0057 & 0.4028 & 2.68E-04 & 7.00E-03 & 0.0383*** & 0.0901 & 2.80E-04 & 7.04E-03 & 0.0397** & 0.0879 \\
	FGL(FF5) & 1.99E-04 & 8.39E-02 & 0.0056 & 0.4032 & 2.67E-04 & 6.99E-03 & 0.0382*** & 0.0901 & 2.79E-04 & 7.04E-03 & 0.0396** & 0.088 \\
	FF1 & -7.16E-04 & 8.39E-02 & -0.0139 & 1.3523 & 1.61E-05 & 7.34E-03 & 0.0022 & 0.2748 & 2.35E-05 & 8.15E-03 & 0.0029 & 0.2736 \\
	FF3 & -7.16E-04 & 9.03E-02 & -0.0139 & 1.3523 & 1.61E-05 & 7.34E-03 & 0.0022 & 0.2748 & 2.35E-05 & 8.15E-03 & 0.0029 & 0.2736 \\
	FF5 & -7.16E-04 & 9.03E-02 & -0.0139 & 1.3523 & 1.61E-05 & 7.34E-03 & 0.0022 & 0.2748 & 2.35E-05 & 8.15E-03 & 0.0029 & 0.2736 \\
	\bottomrule
\end{tabular}%
		}
	\end{table}
\end{landscape}
\clearpage
\subsection{Less Risk-Averse Investors} \label{appendixC2}
Tables \ref{tab1a} and \ref{tab1b} provide the empirical results for higher target levels of risk and return for both monthly and daily data: target risk for monthly and daily data is set at $\sigma=0.08$ and $\sigma=0.02$, respectively. Target return for monthly and daily data is set at $1.1715\%$ and $0.0555\%$, respectively, both are equivalent to $15\%$ yearly return when compounded. Since GMV portfolio weight is not affected by target risk and return, only updated results for MRC and MWC are reported. Furthermore, since EW and Index portfolios are also not affected by target risk and return, their values are the same as in Table \ref{tab1} and, hence, are also not reported to avoid repetition. The conclusions that we highlighted when analyzing updated Tables \ref{tab1} and \ref{tab3} continue to hold.
\clearpage
\newpage
	\begin{table}[]
	\caption{ Monthly portfolio returns, risk, SR and turnover. Targeted risk is set at $\sigma=0.08$, monthly targeted return is $1.1715\%$ which is equivalent to $15\%$ yearly return when compounded. In the upper part corresponding to the results w/o transactions costs, p-values are in parentheses. In the lower part corresponding to the results with transaction costs, $^{***}$ indicates p-value $<$ 0.01, $^{**}$ indicates p-value  $<$ 0.05, and $^{*}$ indicates p-value $<$ 0.10. In-sample: January 1, 1980 - December 31, 1995 (180 obs), Out-of-sample: January 1, 1995 - December 31, 2019 (300 obs).} 
	\label{tab1a}
	\centering
	\resizebox{0.9\textwidth}{!}{%
\begin{tabular}{ccccccccc} 
	\toprule
	& \multicolumn{4}{c}{Markowitz Risk-Constrained} & \multicolumn{4}{c}{Markowitz Weight-Constrained} \\ 
	\hline
	& \textbf{Return} & \textbf{Risk} & \textbf{SR} & \textbf{Turnover} & \textbf{Return} & \textbf{Risk~} & \textbf{SR} & \textbf{Turnover} \\ 
	\hline
	\textbf{Without TC} &  &  &  &  &  &  &  &  \\
	FGL & 0.041 & 0.1324 & \begin{tabular}[c]{@{}c@{}}0.3099\\(0.0769)\end{tabular} & - & 0.0069 & 0.0317 & \begin{tabular}[c]{@{}c@{}}0.2187\\(0.028)\end{tabular} & - \\
	FClime & 0.0596 & 0.3739 & \begin{tabular}[c]{@{}c@{}}0.1593\\(0.1272)\end{tabular} & - & 0.0076 & 0.0441 & \begin{tabular}[c]{@{}c@{}}0.1717\\(0.034)\end{tabular} & - \\
	FLW & 0.0473 & 0.1679 & \begin{tabular}[c]{@{}c@{}}0.2817\\(0.0849)\end{tabular} & - & 0.007 & 0.0344 & \begin{tabular}[c]{@{}c@{}}0.2047\\(0.028)\end{tabular} & - \\
	FNLW & 0.0422 & 0.148 & \begin{tabular}[c]{@{}c@{}}0.2853\\(0.0849)\end{tabular} & - & 0.0071 & 0.0324 & \begin{tabular}[c]{@{}c@{}}0.2190\\(0.028)\end{tabular} & - \\
	POET & NaN & NaN & NaN & - & -0.1144 & 1.9928 & \begin{tabular}[c]{@{}c@{}}-0.0574\\(0.9471)\end{tabular} & - \\
	Projected POET & 0.0933 & 0.5281 & \begin{tabular}[c]{@{}c@{}}0.1766\\(0.1272)\end{tabular} & - & 0.0075 & 0.051 & \begin{tabular}[c]{@{}c@{}}0.1471\\(0.0837)\end{tabular} & - \\
	FGL(FF1) & 0.0439 & 0.128 & \begin{tabular}[c]{@{}c@{}}0.3433\\(0.0649)\end{tabular} & - & 0.0072 & 0.0303 & \begin{tabular}[c]{@{}c@{}}0.2369\\(0.0220)\end{tabular} & - \\
	FGL(FF3) & 0.0438 & 0.1275 & \begin{tabular}[c]{@{}c@{}}0.3437\\(0.0649)\end{tabular} & - & 0.0072 & 0.0301 & \begin{tabular}[c]{@{}c@{}}0.2385\\(0.0220)\end{tabular} & - \\
	FGL(FF5) & 0.0437 & 0.1269 & \begin{tabular}[c]{@{}c@{}}0.3443\\(0.0649)\end{tabular} & - & 0.0072 & 0.0301 & \begin{tabular}[c]{@{}c@{}}0.2377\\(0.0220)\end{tabular} & - \\
	FF1 & 0.0644 & 0.36 & \begin{tabular}[c]{@{}c@{}}0.1789\\(0.1272)\end{tabular} & - & 0.0038 & 0.0538 & \begin{tabular}[c]{@{}c@{}}0.0706\\(0.4833)\end{tabular} & - \\
	FF3 & 0.0623 & 0.3235 & \begin{tabular}[c]{@{}c@{}}0.1926\\(0.1272)\end{tabular} & - & 0.0045 & 0.0513 & \begin{tabular}[c]{@{}c@{}}0.0869\\(0.4833)\end{tabular} & - \\
	FF5 & 0.0566 & 0.2885 & \begin{tabular}[c]{@{}c@{}}0.1962\\(0.1272)\end{tabular} & - & 0.0047 & 0.0513 & \begin{tabular}[c]{@{}c@{}}0.0908\\(0.4833)\end{tabular} & - \\ 
	\hline
	\textbf{With TC} &  &  &  &  &  &  &  &  \\
	FGL & 0.0353 & 0.1792 & 0.2666* & 5.2184 & 0.006 & 0.0317 & 0.1897** & 0.8622 \\
	FClime & 0.0528 & 3.7772 & 0.1422 & 10.133 & 0.007 & 0.0442 & 0.1577* & 0.5971 \\
	FLW & 0.0375 & 0.1881 & 0.223 & 9.5001 & 0.0055 & 0.0345 & 0.1606* & 1.5019 \\
	FNLW & 0.0355 & 0.2159 & 0.2393* & 6.3769 & 0.006 & 0.0325 & 0.185** & 1.0653 \\
	POET & NaN & NaN & NaN & NaN & -0.1933 & 1.7451 & -0.1108 & 124.9832 \\
	Projected POET & 0.0313 & 0.1825 & 0.073 & 85.8766 & 0.0014 & 0.0505 & 0.0277 & 5.9556 \\
	FGL(FF1) & 0.0386 & 0.2476 & 0.3018* & 4.8017 & 0.0064 & 0.0303 & 0.2113** & 0.7219 \\
	FGL(FF3) & 0.0385 & 0.1738 & 0.3018* & 4.8312 & 0.0064 & 0.0301 & 0.2127** & 0.7245 \\
	FGL(FF5) & 0.0383 & 0.1733 & 0.3018* & 4.8537 & 0.0064 & 0.0302 & 0.2112** & 0.7335 \\
	FF1 & 0.0244 & 0.1733 & 0.0707 & 64.7017 & -0.0009 & 0.0535 & -0.0162 & 4.5438 \\
	FF3 & 0.028 & 0.2331 & 0.0896 & 168.9642 & 4.04E-05 & 0.051 & 0.0008 & 4.2854 \\
	FF5 & 0.0237 & 0.2268 & 0.0836 & 34.1596 & 0.0004 & 0.0509 & 0.0077 & 4.1438 \\
	\bottomrule
\end{tabular}
	}
\end{table}
\begin{table}[]
	\caption{ Daily portfolio returns, risk, SR and turnover. Targeted risk is set at $\sigma=0.02$, daily targeted return is $0.0555\%$ which is equivalent to $15\%$ yearly return when compounded.  In the upper part corresponding to the results w/o transactions costs, p-values are in parentheses. In the lower part corresponding to the results with transaction costs, $^{***}$ indicates p-value $<$ 0.01, $^{**}$ indicates p-value  $<$ 0.05, and $^{*}$ indicates p-value $<$ 0.10. In-sample: January 20, 2000 - January 24, 2002  (504 obs), Out-of-sample: January 17, 2002 - January 31, 2020 (4536 obs). }
	\label{tab1b}
	\centering
	\resizebox{0.9\textwidth}{!}{%
\begin{tabular}{ccccccccc} 
	\toprule
	& \multicolumn{4}{c}{Markowitz Risk-Constrained} & \multicolumn{4}{c}{Markowitz Weight-Constrained} \\ 
	\hline
	& \textbf{Return} & \textbf{Risk} & \textbf{SR} & \textbf{Turnover} & \textbf{Return} & \textbf{Risk~} & \textbf{SR} & \textbf{Turnover} \\ 
	\hline
	\textbf{Without TC} &  &  &  &  &  &  &  &  \\
	FGL & 1.25E-03 & 4.09E-02 & \begin{tabular}[c]{@{}c@{}}0.0305\\(0.0709)\end{tabular} & - & 3.10E-04 & 7.86E-03 & \begin{tabular}[c]{@{}c@{}}0.0394\\(0.0260)\end{tabular} & - \\
	FClime & 3.30E-03 & 1.30E-01 & \begin{tabular}[c]{@{}c@{}}0.0254\\(0.0814)\end{tabular} & - & 2.20E-04 & 9.61E-03 & \begin{tabular}[c]{@{}c@{}}0.0229\\(0.036)\end{tabular} & - \\
	FLW & 6.68E-04 & 4.08E-02 & \begin{tabular}[c]{@{}c@{}}0.0164\\(0.1539)\end{tabular} & - & 3.21E-04 & 9.36E-03 & \begin{tabular}[c]{@{}c@{}}0.0343\\(0.0280)\end{tabular} & - \\
	FNLW & 7.56E-04 & 1.02E-01 & \begin{tabular}[c]{@{}c@{}}0.0074\\(0.7312)\end{tabular} & - & 3.02E-04 & 1.16E-02 & \begin{tabular}[c]{@{}c@{}}0.0261\\(0.0360)\end{tabular} & - \\
	POET & NaN & NaN & NaN & - & -5.17E-04 & 2.89E-01 & \begin{tabular}[c]{@{}c@{}}-0.0018\\(0.7419)\end{tabular} & - \\
	Projected POET & 1.84E-03 & 2.63E-01 & \begin{tabular}[c]{@{}c@{}}0.0070\\(0.7312)\end{tabular} & - & -6.76E-05 & 1.58E-02 & \begin{tabular}[c]{@{}c@{}}-0.0043\\(0.7419)\end{tabular} & - \\
	FGL(FF1) & 1.24E-03 & 4.10E-02 & \begin{tabular}[c]{@{}c@{}}0.0303\\(0.0709)\end{tabular} & - & 3.10E-04 & 7.56E-03 & \begin{tabular}[c]{@{}c@{}}0.0410\\(0.0260)\end{tabular} & - \\
	FGL(FF3) & 1.25E-03 & 4.09E-02 & \begin{tabular}[c]{@{}c@{}}0.0306\\(0.0709)\end{tabular} & - & 3.15E-04 & 7.54E-03 & \begin{tabular}[c]{@{}c@{}}0.0417\\(0.0260)\end{tabular} & - \\
	FGL(FF5) & 1.24E-03 & 4.11E-02 & \begin{tabular}[c]{@{}c@{}}0.0301\\(0.0709)\end{tabular} & - & 3.15E-04 & 7.52E-03 & \begin{tabular}[c]{@{}c@{}}0.0419\\(0.0260)\end{tabular} & - \\
	FF1 & 1.14E-03 & 1.71E-01 & \begin{tabular}[c]{@{}c@{}}0.0067\\(0.7312)\end{tabular} & - & 3.78E-05 & 1.64E-02 & \begin{tabular}[c]{@{}c@{}}0.0023\\(0.5813)\end{tabular} & - \\
	FF3 & 1.16E-03 & 1.70E-01 & \begin{tabular}[c]{@{}c@{}}0.0068\\(0.7312)\end{tabular} & - & 3.14E-05 & 1.64E-02 & \begin{tabular}[c]{@{}c@{}}0.0019\\(0.5813)\end{tabular} & - \\
	FF5 & 1.17E-03 & 1.70E-01 & \begin{tabular}[c]{@{}c@{}}0.0069\\(0.7312)\end{tabular} & - & 2.47E-05 & 1.64E-02 & \begin{tabular}[c]{@{}c@{}}0.0015\\(0.5813)\end{tabular} & - \\ 
	\midrule
	\textbf{With TC} &  &  &  &  &  &  &  &  \\
	FGL & 6.14E-04 & 8.67E-02 & 0.0150 & 0.6385 & 2.43E-04 & 7.86E-03 & 0.0310** & 0.0673 \\
	FClime & 1.31E-03 & 6.49E-01 & 0.0101 & 2.4056 & 1.84E-04 & 9.61E-03 & 0.0191 & 0.0382 \\
	FLW & -1.58E-04 & 9.69E-02 & -0.0039 & 0.8283 & 2.01E-04 & 9.38E-03 & 0.0214** & 0.1218 \\
	FNLW & -4.50E-03 & 1.03E-01 & -0.0422 & 10.5211 & 5.71E-05 & 1.17E-02 & 0.0049 & 0.2461 \\
	POET & NaN & NaN & NaN & NaN & -2.50E-02 & 6.21E-01 & -0.0403 & 113.1667 \\
	Projected POET & -2.93E-02 & 1.15E-01 & -0.0315 & 84.1090 & -1.02E-03 & 1.65E-02 & -0.0615 & 0.9502 \\
	FGL(FF1) & 5.81E-04 & 1.43E-01 & 0.0141 & 0.6642 & 2.43E-04 & 7.57E-03 & 0.0321** & 0.0681 \\
	FGL(FF3) & 5.89E-04 & 8.51E-02 & 0.0144 & 0.6642 & 2.47E-04 & 7.55E-03 & 0.0327** & 0.0685 \\
	FGL(FF5) & 5.76E-04 & 8.50E-02 & 0.0140 & 0.6646 & 2.47E-04 & 7.53E-03 & 0.0328** & 0.0687 \\
	FF1 & -1.33E-02 & 8.49E-02 & -0.0858 & 15.6900 & -5.30E-04 & 1.66E-02 & -0.0319 & 0.5790 \\
	FF3 & -1.32E-02 & 1.28E-01 & -0.0854 & 15.6211 & -5.36E-04 & 1.66E-02 & -0.0323 & 0.5785 \\
	FF5 & -1.32E-02 & 1.28E-01 & -0.0852 & 15.5866 & -5.43E-04 & 1.66E-02 & -0.0327 & 0.5786 \\
	\bottomrule
\end{tabular}
	}
\end{table}
\clearpage
\subsection{Subperiod Analyses: MWC and GMV} \label{appendixC3}
Tables \ref{tab6} and \ref{tab7} report subperiod analyses for MWC and MRC portfolio formulations. The values of the EW and Index portfolios are the same as in Table \ref{tab5} and, hence, are also not reported to avoid repetition. In terms of relative comparison between the competing models, the conclusions are similar to those drawn when examining Table \ref{tab5} in the main text. However, in terms of relative magnitude, all models that use MWC or GMV portfolios exhibit deteriorated performance in terms of CER and SR during economic downturns (Downturn \#1 and Downturn \#2): MRC from Table \ref{tab5} is the only type of portfolio that produces positive CER during both recessions.
\begin{sidewaystable}[ph!]
	\centering
	\caption{Cumulative excess return (CER) and risk of MWC portfolios using daily data. Targeted risk is set at $\sigma=0.013$, daily targeted return is $0.0378\%$. P-values are in parentheses. In-sample: January 20, 2000 - January 24, 2002  (504 obs), Out-of-sample: January 17, 2002 - January 31, 2020 (4536 obs).}
	\label{tab6}
	\resizebox{\textwidth}{!}{%
\begin{tabular}{clccccccccccc} 
	\toprule
	& FGL & FClime & FLW & FNLW & POET & ProjPOET & FGL(FF1) & FGL(FF3) & FGL(FF5) & FF1 & FF3 & FF5 \\ 
	\midrule
	\multicolumn{5}{l}{\textbf{Downturn \#1: Argentine Great Depression (2002)}} & \multicolumn{1}{l}{} & \multicolumn{1}{l}{} & \multicolumn{1}{l}{} & \multicolumn{1}{l}{} & \multicolumn{1}{l}{} & \multicolumn{1}{l}{} & \multicolumn{1}{l}{} & \multicolumn{1}{l}{} \\
	CER & -0.0138 & -0.1045 & -0.0158 & -0.0195 & -0.2820 & -0.0217 & -0.0153 & -0.0176 & -0.0187 & -0.0334 & -0.0334 & -0.0334 \\
	Risk & 0.0082 & 0.0124 & 0.0080 & 0.0078 & 0.0324 & 0.0130 & 0.0078 & 0.0078 & 0.0078 & 0.0095 & 0.0095 & 0.0095 \\
	SR & -0.0031 & \begin{tabular}[c]{@{}c@{}}-0.0314\\(0.6753)\end{tabular} & \begin{tabular}[c]{@{}c@{}}-0.0045\\(0.6753)\end{tabular} & \begin{tabular}[c]{@{}c@{}}-0.0069\\(0.6753)\end{tabular} & \begin{tabular}[c]{@{}c@{}}-0.0265\\(0.6753)\end{tabular} & \begin{tabular}[c]{@{}c@{}}-0.0007\\(0.6753)\end{tabular} & \begin{tabular}[c]{@{}c@{}}-0.0044\\(0.6753)\end{tabular} & \begin{tabular}[c]{@{}c@{}}-0.0057\\(0.6753)\end{tabular} & \begin{tabular}[c]{@{}c@{}}-0.0063\\(0.6753)\end{tabular} & \begin{tabular}[c]{@{}c@{}}-0.0194\\(0.6414)\end{tabular} & \begin{tabular}[c]{@{}c@{}}-0.0194\\(0.6414)\end{tabular} & \begin{tabular}[c]{@{}c@{}}-0.0194\\(0.6414)\end{tabular} \\ 
	\midrule
	\multicolumn{5}{l}{\textbf{Downturn \#2: Financial Crisis (2008)}} & \multicolumn{1}{l}{} & \multicolumn{1}{l}{} & \multicolumn{1}{l}{} & \multicolumn{1}{l}{} & \multicolumn{1}{l}{} & \multicolumn{1}{l}{} & \multicolumn{1}{l}{} & \multicolumn{1}{l}{} \\
	CER & -0.1956 & -0.3974 & -0.2789 & -0.2811 & -0.9989 & -0.0842 & -0.2107 & -0.2074 & -0.2053 & -0.2669 & -0.2669 & -0.2669 \\
	Risk & 0.0135 & 0.0204 & 0.0126 & 0.0123 & 0.1198 & 0.0176 & 0.0134 & 0.0134 & 0.0133 & 0.0183 & 0.0183 & 0.0183 \\
	SR & \begin{tabular}[c]{@{}l@{}}0.0135\\(0.4555)\end{tabular} & \begin{tabular}[c]{@{}c@{}}0.0204\\(0.4715)\end{tabular} & \begin{tabular}[c]{@{}c@{}}0.0126\\(0.4715)\end{tabular} & \begin{tabular}[c]{@{}c@{}}0.0123\\(0.4715)\end{tabular} & \begin{tabular}[c]{@{}c@{}}0.1198\\(0.4715)\end{tabular} & \begin{tabular}[c]{@{}c@{}}0.0176\\(0.4715)\end{tabular} & \begin{tabular}[c]{@{}c@{}}0.0134\\(0.4645)\end{tabular} & \begin{tabular}[c]{@{}c@{}}0.0134\\(0.4685)\end{tabular} & \begin{tabular}[c]{@{}c@{}}0.0133\\(0.4715)\end{tabular} & \begin{tabular}[c]{@{}c@{}}0.0113\\(0.4486)\end{tabular} & \begin{tabular}[c]{@{}c@{}}0.0113\\(0.4486)\end{tabular} & \begin{tabular}[c]{@{}c@{}}0.0183\\(0.4486)\end{tabular} \\ 
	\midrule
	\multicolumn{5}{l}{\textbf{Boom \#1 (2017)}} & \multicolumn{1}{l}{} & \multicolumn{1}{l}{} & \multicolumn{1}{l}{} & \multicolumn{1}{l}{} & \multicolumn{1}{l}{} & \multicolumn{1}{l}{} & \multicolumn{1}{l}{} & \multicolumn{1}{l}{} \\
	CER & 0.1398 & 0.1309 & 0.1267 & -0.0361 & 0.5720 & -0.0877 & 0.1406 & 0.1407 & 0.1419 & -0.0361 & -0.0349 & -0.0676 \\
	Risk & 0.0044 & 0.0041 & 0.0037 & 0.0087 & 0.0630 & 0.0089 & 0.0046 & 0.0046 & 0.0046 & 0.0070 & 0.0070 & 0.0070 \\
	SR & \begin{tabular}[c]{@{}l@{}}0.1194\\(0.5814)\end{tabular} & \begin{tabular}[c]{@{}c@{}}0.1227\\(0.5884)\end{tabular} & \begin{tabular}[c]{@{}c@{}}0.1308\\(0.5814)\end{tabular} & \begin{tabular}[c]{@{}c@{}}-0.0124\\(0.7644)\end{tabular} & \begin{tabular}[c]{@{}c@{}}0.0510\\(0.5218)\end{tabular} & \begin{tabular}[c]{@{}c@{}}-0.0367\\(0.7644)\end{tabular} & \begin{tabular}[c]{@{}c@{}}0.1151\\(0.5814)\end{tabular} & \begin{tabular}[c]{@{}c@{}}0.1154\\(0.5814)\end{tabular} & \begin{tabular}[c]{@{}c@{}}0.1177\\(0.5814)\end{tabular} & \begin{tabular}[c]{@{}c@{}}-0.0173\\(0.7644)\end{tabular} & \begin{tabular}[c]{@{}c@{}}-0.0165\\(0.7644)\end{tabular} & \begin{tabular}[c]{@{}c@{}}-0.0361\\(0.7644)\end{tabular} \\ 
	\midrule
	\multicolumn{5}{l}{\textbf{Boom \#2 (2019)}} & \multicolumn{1}{l}{} & \multicolumn{1}{l}{} & \multicolumn{1}{l}{} & \multicolumn{1}{l}{} & \multicolumn{1}{l}{} & \multicolumn{1}{l}{} & \multicolumn{1}{l}{} & \multicolumn{1}{l}{} \\
	CER & 0.3787 & 0.2595 & 0.3018 & 0.4078 & 1.4756 & 0.5300 & 0.2492 & 0.2497 & 0.2506 & 0.3839 & 0.3845 & 0.3896 \\
	Risk & 0.0085 & 0.0078 & 0.0072 & 0.0098 & 0.0403 & 0.0176 & 0.0063 & 0.0064 & 0.0064 & 0.0175 & 0.0175 & 0.0175 \\
	SR & \begin{tabular}[c]{@{}l@{}}0.1533\\(0.5715)\end{tabular} & \begin{tabular}[c]{@{}c@{}}0.1215\\(0.5920)\end{tabular} & \begin{tabular}[c]{@{}c@{}}0.1495\\(0.5920)\end{tabular} & \begin{tabular}[c]{@{}c@{}}0.1432\\(0.5920)\end{tabular} & \begin{tabular}[c]{@{}c@{}}0.1092\\(0.6512)\end{tabular} & \begin{tabular}[c]{@{}c@{}}0.1046\\(0.6512)\end{tabular} & \begin{tabular}[c]{@{}c@{}}0.1423\\(0.5920)\end{tabular} & \begin{tabular}[c]{@{}c@{}}0.1423\\(0.5920)\end{tabular} & \begin{tabular}[c]{@{}c@{}}0.1427\\(0.5920)\end{tabular} & \begin{tabular}[c]{@{}c@{}}0.0816\\(0.8912)\end{tabular} & \begin{tabular}[c]{@{}c@{}}0.0817\\(0.8912)\end{tabular} & \begin{tabular}[c]{@{}c@{}}0.0826\\(0.8912)\end{tabular} \\
	\bottomrule
\end{tabular}%
	}
\end{sidewaystable}
\begin{sidewaystable}[ph!]
	\centering
	\caption{Cumulative excess return (CER) and risk of GMV portfolios using daily data. Targeted risk is set at $\sigma=0.013$, daily targeted return is $0.0378\%$. P-values are in parentheses. In-sample: January 20, 2000 - January 24, 2002  (504 obs), Out-of-sample: January 17, 2002 - January 31, 2020 (4536 obs).}
	\label{tab7}
	\resizebox{\textwidth}{!}{%
\begin{tabular}{clccccccccccc} 
	\toprule
	& FGL & FClime & FLW & FNLW & POET & ProjPOET & FGL(FF1) & FGL(FF3) & FGL(FF5) & FF1 & FF3 & FF5 \\ 
	\midrule
	\multicolumn{5}{l}{\textbf{Downturn \#1: Argentine Great Depression (2002)}} & \multicolumn{1}{l}{} & \multicolumn{1}{l}{} & \multicolumn{1}{l}{} & \multicolumn{1}{l}{} & \multicolumn{1}{l}{} & \multicolumn{1}{l}{} & \multicolumn{1}{l}{} & \multicolumn{1}{l}{} \\
	CER & -0.0044 & -0.1061 & -0.0151 & -0.0206 & -0.3190 & -0.0662 & -0.0038 & -0.0059 & -0.0076 & -0.0335 & -0.0335 & -0.0335 \\
	Risk & 0.0081 & 0.0129 & 0.0080 & 0.0078 & 0.0330 & 0.0135 & 0.0077 & 0.0077 & 0.0077 & 0.0096 & 0.0096 & 0.0096 \\
	SR & \begin{tabular}[c]{@{}l@{}}0.0017\\(0.6543)\end{tabular} & \begin{tabular}[c]{@{}c@{}}-0.0306\\(0.7564)\end{tabular} & \begin{tabular}[c]{@{}c@{}}-0.0041\\(0.6583)\end{tabular} & \begin{tabular}[c]{@{}c@{}}-0.0075\\(0.6583)\end{tabular} & \begin{tabular}[c]{@{}c@{}}-0.0322\\(0.7564)\end{tabular} & \begin{tabular}[c]{@{}c@{}}-0.0148\\(0.6583)\end{tabular} & \begin{tabular}[c]{@{}c@{}}0.0017\\(0.6543)\end{tabular} & \begin{tabular}[c]{@{}c@{}}0.0006\\(0.6543)\end{tabular} & \begin{tabular}[c]{@{}c@{}}-0.0004\\(0.6543)\end{tabular} & \begin{tabular}[c]{@{}c@{}}-0.0193\\(0.6583)\end{tabular} & \begin{tabular}[c]{@{}c@{}}-0.0193\\(0.6583)\end{tabular} & \begin{tabular}[c]{@{}c@{}}-0.0193\\(0.6583)\end{tabular} \\ 
	\midrule
	\multicolumn{5}{l}{\textbf{Downturn \#2: Financial Crisis (2008)}} & \multicolumn{1}{l}{} & \multicolumn{1}{l}{} & \multicolumn{1}{l}{} & \multicolumn{1}{l}{} & \multicolumn{1}{l}{} & \multicolumn{1}{l}{} & \multicolumn{1}{l}{} & \multicolumn{1}{l}{} \\
	CER & -0.2113 & -0.4410 & -0.2926 & -0.2959 & -0.9928 & 0.0829 & -0.2291 & -0.2251 & -0.2226 & -0.2938 & -0.2938 & -0.2938 \\
	Risk & 0.0138 & 0.0241 & 0.0128 & 0.0124 & 0.0931 & 0.0247 & 0.0136 & 0.0136 & 0.0136 & 0.0186 & 0.0186 & 0.0186 \\
	SR & \begin{tabular}[c]{@{}l@{}}0.0138\\(0.4146)\end{tabular} & \begin{tabular}[c]{@{}c@{}}0.0241\\(0.4296)\end{tabular} & \begin{tabular}[c]{@{}c@{}}0.0128\\(0.4296)\end{tabular} & \begin{tabular}[c]{@{}c@{}}0.0124\\(0.4296)\end{tabular} & \begin{tabular}[c]{@{}c@{}}0.0931\\(0.4296)\end{tabular} & \begin{tabular}[c]{@{}c@{}}0.0247\\(0.4296)\end{tabular} & \begin{tabular}[c]{@{}c@{}}0.0136\\(0.4196)\end{tabular} & \begin{tabular}[c]{@{}c@{}}0.0136\\(0.4276)\end{tabular} & \begin{tabular}[c]{@{}c@{}}0.0136\\(0.4296)\end{tabular} & \begin{tabular}[c]{@{}c@{}}0.0116\\(0.4096)\end{tabular} & \begin{tabular}[c]{@{}c@{}}0.0116\\(0.4096)\end{tabular} & \begin{tabular}[c]{@{}c@{}}0.0186\\(0.4096)\end{tabular} \\ 
	\midrule
	\multicolumn{5}{l}{\textbf{Boom \#1 (2017)}} & \multicolumn{1}{l}{} & \multicolumn{1}{l}{} & \multicolumn{1}{l}{} & \multicolumn{1}{l}{} & \multicolumn{1}{l}{} & \multicolumn{1}{l}{} & \multicolumn{1}{l}{} & \multicolumn{1}{l}{} \\
	CER & 0.1384 & 0.1264 & 0.1323 & -0.0388 & -1.0000 & -0.1106 & 0.1387 & 0.1388 & 0.1402 & -0.0389 & -0.0366 & -0.0698 \\
	Risk & 0.0045 & 0.0041 & 0.0037 & 0.0090 & 0.2414 & 0.0115 & 0.0047 & 0.0047 & 0.0046 & 0.0065 & 0.0065 & 0.0066 \\
	SR & \begin{tabular}[c]{@{}l@{}}0.1177\\(0.5994)\end{tabular} & \begin{tabular}[c]{@{}c@{}}0.1183\\(0.6044)\end{tabular} & \begin{tabular}[c]{@{}c@{}}0.1366\\(0.6044)\end{tabular} & \begin{tabular}[c]{@{}c@{}}-0.0131\\(0.8024)\end{tabular} & \begin{tabular}[c]{@{}c@{}}-0.0723\\(0.9115)\end{tabular} & \begin{tabular}[c]{@{}c@{}}-0.0347\\(0.8024)\end{tabular} & \begin{tabular}[c]{@{}c@{}}0.1131\\(0.6014)\end{tabular} & \begin{tabular}[c]{@{}c@{}}0.1133\\(0.6044)\end{tabular} & \begin{tabular}[c]{@{}c@{}}0.1157\\(0.6044)\end{tabular} & \begin{tabular}[c]{@{}c@{}}-0.0211\\(0.8024)\end{tabular} & \begin{tabular}[c]{@{}c@{}}-0.0196\\(0.8024)\end{tabular} & \begin{tabular}[c]{@{}c@{}}-0.0404\\(0.9023)\end{tabular} \\ 
	\midrule
	\multicolumn{1}{l}{\textbf{Boom \#2 (2019)}} &  & \multicolumn{1}{l}{} & \multicolumn{1}{l}{} & \multicolumn{1}{l}{} & \multicolumn{1}{l}{} & \multicolumn{1}{l}{} & \multicolumn{1}{l}{} & \multicolumn{1}{l}{} & \multicolumn{1}{l}{} & \multicolumn{1}{l}{} & \multicolumn{1}{l}{} & \multicolumn{1}{l}{} \\
	CER & 0.3703 & 0.2829 & 0.2994 & 0.3287 & 1.6301 & 0.6870 & 0.2503 & 0.2504 & 0.2504 & 0.4031 & 0.4038 & 0.4087 \\
	Risk & 0.0072 & 0.0081 & 0.0084 & 0.0097 & 0.0318 & 0.0186 & 0.0063 & 0.0063 & 0.0063 & 0.0185 & 0.0185 & 0.0184 \\
	SR & \begin{tabular}[c]{@{}l@{}}0.1521\\(0.5644)\end{tabular} & \begin{tabular}[c]{@{}c@{}}0.1266\\(0.5714)\end{tabular} & \begin{tabular}[c]{@{}c@{}}0.1478\\(0.5714)\end{tabular} & \begin{tabular}[c]{@{}c@{}}0.1419\\(0.5714)\end{tabular} & \begin{tabular}[c]{@{}c@{}}0.1366\\(0.5714)\end{tabular} & \begin{tabular}[c]{@{}c@{}}0.1209\\(0.5714)\end{tabular} & \begin{tabular}[c]{@{}c@{}}0.1441\\(0.5684)\end{tabular} & \begin{tabular}[c]{@{}c@{}}0.1440\\(0.5684)\end{tabular} & \begin{tabular}[c]{@{}c@{}}0.1439\\(0.5684)\end{tabular} & \begin{tabular}[c]{@{}c@{}}0.0810\\(0.7592)\end{tabular} & \begin{tabular}[c]{@{}c@{}}0.0811\\(0.7592)\end{tabular} & \begin{tabular}[c]{@{}c@{}}0.0819\\(0.7592)\end{tabular} \\
	\bottomrule
\end{tabular}%
	}
\end{sidewaystable}
\end{spacing}
\end{appendices}
\end{document}